%% file: arxiv.tex
\documentclass[11pt, letterpaper]{article}
\usepackage[utf8]{inputenc}
\usepackage{booktabs,multirow}
\usepackage{array}
\input{a-info-packages}

\title{Metric Distortion in Peer Selection}

 \author[1,2]{Javier Cembrano \thanks{\href{mailto:jcembran@mpi-inf.mpg.de}{jcembran@mpi-inf.mpg.de}}}
 \author[1,3]{Golnoosh Shahkarami \thanks{\href{mailto:gshahkar@mpi-inf.mpg.de}{gshahkar@mpi-inf.mpg.de}}}

 \affil[1]{Department of Agorithms and Complexity, Max Planck Institut für Informatik}
 \affil[2]{Department of Industrial Engineering, Universidad de Chile}
 \affil[3]{Graduate School of Computer Science, Universität des Saarlandes}

\date{}

\begin{document}

\maketitle

\begin{abstract}
In the \textit{metric distortion} problem, a set of voters and candidates lie in a common metric space, and a committee of $k$ candidates must be elected. The objective is to minimize a social cost, defined as a function of the distances between voters and their chosen representatives, while the voting rule only has access to ordinal preferences. The \textit{distortion} of a rule is the worst-case ratio between the social cost of its outcome and that of the optimal committee, taken over all consistent preferences and metrics.

We initiate the study of metric distortion in peer selection, where voters and candidates coincide. 
We consider four objectives, obtained by combining two aggregation rules with two types of social cost. 
Under \textit{additive aggregation}, an individual’s cost is the sum of their distances to all committee members; under \textit{$q$-cost}, it is their distance to the $q$th closest member. 
The overall social cost is either \textit{utilitarian}, given by the sum of all individual costs, or \textit{egalitarian}, given by the maximum individual cost.
Surprisingly, we find that even on the line metric, peer selection retains much of the hardness of the general case: Lower bounds remain strictly larger than one for all objectives, and cases where bounded distortion is impossible in general remain so here as well.
On a positive note, cases with bounded distortion in the general setting achieve better constants in peer selection. For utilitarian cost, selecting the $k$ middle agents achieves a distortion between $1$ and $2$ under additive aggregation. Under $q$-cost, we show positive results for $q=k=2$, but impossibility results largely carry over. For egalitarian cost, selecting the extremes yields an optimal distortion of $2$ under additive aggregation and for $q$-cost with $q>k/3$, while no bounded distortion is possible when $q \leq k/3$.
Overall, our results show that peer selection on the line metric admits better constants than the general case, yet fundamental hardness barriers persist.
\end{abstract}

\newpage
\section{Introduction}
\input{a-Introduction}

\section{Preliminaries}
\label{sec:prelim}
\input{a-Preliminaries}

\section{Utilitarian Social Cost}\label{sec:utilititarian}
\input{a-Utilitarian}

\section{Egalitarian Social Cost}\label{sec:egalitarian}
\input{a-Egalitarian}

\section{Discussion}

In this work, we have introduced the study of metric distortion in committee elections where voters and candidates coincide and provided a first step towards an understanding of this setting by focusing on the line metric.
Our results span a variety of social costs and include both analyses of voting rules and constructions of negative instances to provide impossibility results.
Although most of our results are tight, an intriguing gap remains for utilitarian $q$-cost when $q$ is greater than $\frac{k}{2}$.
We believe that rules with a distortion better than the current upper bound of $3$ exist and their design may benefit from the insights provided by our rule for $q=k=2$.

The study of the distortion of voting rules in more general metric spaces constitutes an interesting direction for future work, even in the general setting. 
The main open question concerns the design of voting rules that achieve small distortion beyond the line metric under the utilitarian additive cost.

Another challenge in the design of elections is preventing strategic behavior.
A mild assumption in the context of peer selection, adopted by the growing literature on impartial selection, is that agents' primary concern is whether they are selected themselves, and a voting rule is deemed impartial if an agent cannot affect this fact by changing their reported preferences.
On the other hand, a rule is called strategyproof in the voting literature if no agent can misreport their preferences and lead to a better outcome with respect to their actual preferences.
Both notions---impartiality and strategyproofness--can be readily applied to our setting, the former being a relaxed version of the latter in this case.
Most of the voting rules developed in this work depend on the order of the agents and are thus strategyproof if one restricts voters' deviations to those that are consistent with this order.
This constitutes a sensible way to define these axioms, as inconsistent reports could be easily detected and punished by the designer.
A notable exception is the \favCouple\ rule, which does not depend exclusively on the order and is not even impartial: For instance, an agent next to the median agent could in some cases modify their ranking, reporting the median agent immediately after themselves, to create a couple and become selected.
Designing impartial and strategyproof voting rules with bounded distortion for peer selection constitutes an interesting challenge for future work in the area.

\pagebreak
\bibliographystyle{plainnat}
\bibliography{bibliography}

\end{document}

%% file: a-info-packages.tex
\usepackage{booktabs} 
\usepackage[textsize=scriptsize]{todonotes}
\usepackage{natbib}
\usepackage{authblk}
\usepackage{a4wide}
\usepackage[colorlinks=true, linkcolor=black, urlcolor=black, citecolor=black]{hyperref}
\usepackage{amsmath}
\usepackage{amsthm}
\usepackage{cleveref}

\crefname{equation}{inequality}{inequalities}
\Crefname{equation}{Inequality}{Inequalities}

\usepackage{hyperref}
\usepackage{enumitem}
\usepackage{amssymb}
\usepackage{multirow}
\usepackage{thm-restate}
\usepackage{etoolbox}

\usepackage{tikz}
\usepackage{tikz-network}
\usetikzlibrary{patterns}
\usetikzlibrary{decorations.pathreplacing,calligraphy}
\usetikzlibrary{decorations.pathmorphing}
\usetikzlibrary{backgrounds}
\usetikzlibrary{arrows}
\usetikzlibrary{plotmarks}
\usetikzlibrary{fadings}
\tikzset{discont/.style={decoration={zigzag,segment length=12pt, amplitude=4pt},decorate}}
\def\discontarrow(#1)(#2)(#3)(#4);{
  \draw[discont,thick] (#2) -- (#3);
  \draw[thick,-latex] (#1) -- (#2)  (#3) -- (#4);
}
\usepackage{subcaption}

\newcommand{\NN}{\mathbb{N}}
\newcommand{\RR}{\mathbb{R}}
\newcommand{\calL}{\mathcal{L}}

\newcommand{\SC}{\text{SC}}
\newcommand{\instance}{\mathcal{E}}
\newcommand{\dist}{\mathrm{dist}}
\newcommand{\ve}[1]{\mathbf{#1}}

\newcommand{\medianAlt}{\textsc{Median Alternation}}
\newcommand{\favCouple}{\textsc{Favorite Couple}}

\newcommand{\kExtremes}{\textsc{$k$-Extremes}}

 \newtheorem{lemma}{Lemma}
 \newtheorem{claim}{Claim}[section]
 \newtheorem{Vrule}{Voting Rule}

\crefname{claim}{Claim}{Claims}
\crefname{lemma}{Lemma}{Lemmas}
\crefname{theorem}{Theorem}{Theorems}
\crefname{part}{part}{parts}
\crefname{alt}{alternative}{alternatives}
\crefname{step}{step}{steps}
\crefname{algocf}{Algorithm}{Algorithms}
\crefname{eq}{equality}{equalities}
\crefname{eqs}{equalities}{equalities}
\crefname{eqn}{equation}{equations}
\crefname{ineq}{inequality}{inequalities}
\crefname{ineqs}{inequalities}{inequalities}
\crefname{exp}{expression}{expressions}
\crefname{prop}{property}{properties}
\crefname{props}{properties}{properties}
\crefname{imp}{implication}{implications}
\crefname{eqv}{equivalence}{equivalences}
\crefname{def}{definition}{definitions}

\creflabelformat{eq}{(#2#1#3)}
\creflabelformat{eqs}{(#2#1#3)}
\creflabelformat{eqn}{(#2#1#3)}
\creflabelformat{ineq}{(#2#1#3)}
\creflabelformat{ineqs}{(#2#1#3)}
\creflabelformat{exp}{(#2#1#3)}
\creflabelformat{prop}{(#2#1#3)}
\creflabelformat{props}{(#2#1#3)}
\creflabelformat{imp}{(#2#1#3)}
\creflabelformat{eqv}{(#2#1#3)}
\creflabelformat{def}{(#2#1#3)}

%% file: a-Introduction.tex
A fundamental problem in social choice is the aggregation of individual preferences, expressed as rankings over a set of candidates, into a social preference consisting of a subset of elected candidates.
For centuries, social choice theorists have proposed several axioms to capture desirable properties that these aggregation or \textit{voting} rules should satisfy, usually leading to strong impossibility results \citep{arrow1963social,condorcet1785essai,Gibbard73,Satterthwaite75}.

As an alternative approach, attempting to quantify the extent to which a certain voting rule can faithfully translate voters' preferences into the selected committee, \citet{procaccia2006distortion} introduced the notion of \textit{distortion} of a rule.
The underlying assumption is that a voter's (dis)affinity with a candidate can be represented by a certain cost, and voters' rankings express these cardinal preferences.
The cost of a committee for a voter is then defined by aggregating the costs of the committee members, and the overall \textit{social cost} of the committee by aggregating the costs for all voters.
The distortion corresponds to the worst-case ratio between the social cost of the selected committee and that of the optimal committee, over all possible preferences and consistent metrics.

The study of the distortion of voting rules has usually focused on two ways of modeling the social cost: utilitarian and egalitarian~\citep{caragiannis2011voting, 10.1145/3230654.3230658, caragiannis2017subset}.
In the utilitarian case, the social cost is defined as the sum of the individual costs of the voters, so that all voters' costs contribute equally to the objective.
In contrast, the egalitarian social cost considers the maximum individual cost among all voters, capturing a notion of fairness where no voter is excessively disadvantaged.

In voting theory, it is common to assume that voters' preferences are not fully arbitrary but exhibit some structural properties.
A relevant line of work has sought structural restrictions that are natural and have powerful implications, such as single-peaked \citep{black48grouDecision} or single-crossing \citep{mirrlees1971exploration}; see \citet{elkind2022preference} for a survey.
A rather general framework among these is that of \textit{spatial} or \textit{metric voting}, where voters and candidates are assumed to lie in a common low-dimensional metric space and voters' costs correspond to their distance to each candidate~\citep{aziz2020justifications,jessee2012ideology,enelow1984spatial,merrill1999unified}.
For instance, a line metric is commonly employed to capture political affinity on the left-right spectrum, whereas geographical preferences are typically modeled in two-dimensional spaces.

This structural assumption naturally fits in the metric distortion framework: the distances to candidates fully define the social cost of a committee, but the voting rules only receive their expression as preference rankings.
With this structural restriction, a tight distortion bound of $3$ has been shown for any single-winner deterministic voting rule~\citep{ANSHELEVICH201827, kizilkaya2022plurality, gkatzelis2020resolving}.
Extending distortion to multi-winner elections requires defining how a voter's cost is aggregated over the selected committee.
Two ways have been considered in the literature: the \textit{additive cost}, where a voter's cost is the sum of their distances to all members of the committee~\citep{babashah2024distortionmultiwinnerelectionsline}, and the \textit{$q$-cost}, where the cost is determined by their distance to their $q$th closest committee member~\citep{caragiannis2022metric,Chen_Li_Wang_2020}.

Work on metric distortion has so far focused on the case where voters and candidates constitute disjoint sets, which forms a natural model for large-scale elections.
However, in many decision-making scenarios, a group of agents seeks to elect a subset of their own members.
One can think, for example, of a political organization selecting a committee.
Each member ranks others according to their political affinity, and the organization aims to select a committee that represents the variety of preferences of its members.
Since the voting rule only receives ordinal preferences, a small distortion constitutes a suitable objective to ensure a close-to-optimal outcome under this limited information.
In general, this situation arises in the context of \textit{peer selection}, where individuals evaluate each other to choose a group for governance, leadership, or resource allocation.
Further examples include academic hiring and promotion, student representative elections, self-organized committees in cooperatives, and local governance selection. 

While peer selection rules have been extensively studied in other contexts, particularly regarding the effect of strategic behavior~\citep[e.g.][]{holzman2013impartial,alon2011sum,caragiannis2022impartial}, little is known about their ability to accurately reflect agents' cardinal preferences.
On the other hand, previous work on metric distortion for single-selection has often parameterized an election via its \textit{decisiveness}, corresponding to the maximum ratio between a voter's distance to their top choice and to any other candidate~\citep{anshelevich2017randomized,gkatzelis2020resolving}.
These works have motivated this parameter by the fact that it becomes zero in the peer selection setting, as each agent becomes their own top choice.
However, directly modeling a common set of voters and candidates constitutes a structural change to the problem that has not been explored so far. 

\subsection{Our Contributions and Techniques}
We initiate the study of metric distortion when the set of voters and candidates coincide and bound the distortion achievable by voting rules selecting $k$ out of $n$ agents on the line metric for several social costs.
A summary of our results appears in \Cref{tab:results}.

\newcommand{\raisedrule}[2][2em]{\leaders\hbox{\rule[#1]{1pt}{#2}}\hfill}

{\footnotesize
\begin{table}[tb]
\centering

\makebox[\columnwidth][c]{%
\begin{tabular}{@{}l c cc@{}}
\toprule\\[-8pt]
& \multirow{2}{*}{\textbf{additive}} & \multicolumn{2}{c}{\textbf{$q$-cost}\phantom{aaaaaaaaa}} \\[2pt]
&  & $q\le \frac{k}{2}$ & $\frac{k}{2}< q\le k$ \\[3pt]
\midrule\\[-4pt]
\multirow{2}{*}{\rotatebox{90}{\bf utilitarian}}
& {\color{gray}$1+\sqrt{1+\frac{2}{k}}$; $\frac{7}{3}+\frac{4}{k}\big(\sqrt{2}-\frac{4}{3}\big)$ \scriptsize[BKSS]}
& \multicolumn{1}{c}{{\color{gray}$\infty$ \scriptsize[CSV]}}
& {\color{gray}$3$; $3$ \scriptsize[CSV]} \\[18pt]
& $1.0914$; $\displaystyle 
\tfrac{2}{k}\big(n-2\sqrt{n(n-k)}\big)^{(*)}$
\scriptsize[T.~\ref{thm:Util-Add}]
& $\infty$ \scriptsize[T.~\ref{thm:utilitarian-1cost}]
& $\displaystyle 2 - \tfrac{k-q}{4q-k-3}^{(**)}$; $3$ \scriptsize[T.~\ref{thm:utilitarian-lowerbound-qcost}, CSV] \\[8pt]

\midrule\\[-8pt]
& \multirow{2}{*}{\textbf{additive}} & \multicolumn{2}{c}{\textbf{$q$-cost}\phantom{aaaaaaaaa}} \\[2pt]
&  & $q\le \frac{k}{3}$ & $\frac{k}{3}< q\le k$ \\[3pt]
\midrule\\[-8pt]
\multirow{2}{*}{\rotatebox{90}{\bf egalitarian}}
& {\color{gray}$\cdot$; $2$ \scriptsize[BKSS]}
& {\color{gray}$\infty$ \scriptsize[CSV]}
& {\color{gray}$3$; $3$ \scriptsize[CSV]} \\[6pt]
& $\tfrac{3}{2}-\tfrac{1}{k}$; $\tfrac{3}{2}-\tfrac{1}{2(k-1)}^{(*)}$ \scriptsize[T.~\ref{thm:egalitarian-additive}]
& $\infty$ \scriptsize[T.~\ref{thm:egalitarian-qcost-smallq}]
& $2$; $2$ \scriptsize[T.~\ref{thm:egalitarian-qcost-qlarge}] \\[22pt]
\bottomrule
\end{tabular}
}
\caption{
\label{tab:results}
Summary of our and previous bounds on the distortion that voting rules can achieve in different settings.
Values before and after the semicolon represent lower and upper bounds, respectively.
Gray entries correspond to the previously studied setting with disjoint voters and candidates, either under a general metric~[CSV] or under the line metric~[BKSS].
The upper bounds for utilitarian additive and egalitarian additive social costs marked with $(*)$ correspond to the cases of even $n-k$ and even $k$, respectively; the slightly different bounds for the odd cases are stated in the corresponding theorems.
The upper bound for utilitarian $q$-cost marked with $(**)$ in the text is valid for $q\ge \frac{k}{2}+1$ (slightly stronger than $q>\frac{k}{2}$). BKSS: \citet{babashah2024distortionmultiwinnerelectionsline}; CSV: \citet{caragiannis2022metric}.
}
\end{table}
}

We start by observing a simple yet powerful property of metric voting on the line with a single set of voters and candidates, which follows from previous work \citep{elkind2014recognizing, babashah2024distortionmultiwinnerelectionsline}: we can fully compute the order of the agents from their rankings.
This constitutes a powerful tool for the design of our mechanisms, as in the following we can always take this order as given.

\paragraph{Utilitarian Additive Cost.}
We first consider the utilitarian social cost, in which the social cost of a committee is defined as the sum of all individual costs.
Intuitively, selecting $k$ consecutive agents results in lower utilitarian social cost.
In \Cref{subsec:uti-add}, we focus on the case of additive aggregation: the cost of a committee for a voter is given by the sum of all distances from the candidates to this voter.
As a natural extension of the optimal rules for one or two agents, which select the median and closest-to-median agents, we consider a rule called \medianAlt\ that selects $k$ middle agents.
We show that \medianAlt\ provides a distortion close to $1$ when $k$ is small compared to $n$ and approaches $2$ as $k$ goes to $n$.
Despite its simplicity, the analysis of this rule poses significant challenges.
In short, we reduce any metric to another with only two locations by showing the existence of a non-improving direction of movement for each agent, and then compute the worst-case distortion for this class.

We complement the upper bound described above with a lower bound of $1.0914$ for any $k\geq 3$.
This implies that, even in this seemingly simple case, we cannot always select the optimal committee.

\paragraph{Utilitarian $q$-Cost.}
In \Cref{subsec:uti-q}, we consider utilitarian $q$-cost, where the cost of a committee for an agent is given by the agent's distance to their $q$th closest candidate in the committee.
In \Cref{thm:utilitarian-1cost}, we show that no voting rule can provide a constant distortion when $q\leq \frac{k}{2}$, implying that this known impossibility from the setting with disjoint voters and candidates and a general metric space~\citep{caragiannis2022metric} remains in place in our restricted setting.
To prove this bound, we partition all but $q$ agents into $\big\lfloor \frac{k}{q}\big\rfloor \geq 2$ sets and consider two metrics that differ in the position of the remaining $q$ agents: relatively close to the other agents in one metric; very far in the other.
Intuitively, selecting these $q$ agents leads to an unbounded distortion in the former case but is necessary for a bounded distortion in the latter.
For $q>\frac{k}{2}$, the existence of rules with distortion $3$ follows from a general result by \citet{caragiannis2022metric}.
We provide a lower bound that varies between $\frac{3}{2}$ and $2$ as $q$ varies between $\frac{k}{2}+1$ and $2$, by considering three different metrics consistent with the same rankings and showing that, in one of them, there are $q$ agents in one extreme that cannot be consistently selected.
We take a closer look at the case with $k = q = 2$, where a best-possible distortion of $2$ can be achieved by selecting the median agents when $n$ is even.
For odd $n$, we show that a rule selecting a \textit{couple} of agents—a pair of agents who prefer each other over all other agents—among the five middle agents achieves an improved distortion of $\frac{4}{3}$, which is again best-possible.
The \favCouple\ rule leverages two key principles: (1) selecting agents close to the median so as to balance overall distances from agents on each side of the median, and (2) selecting consecutive agents with a small distance between them since this distance is part of the cost of all agents.
This intuition of selecting consecutive agents that are as close to each other as possible while also being close to the median in principle holds for larger $k$, but determining how tightly a group of $k$ agents is clustered based solely on ordinal rankings remains a challenge.

\paragraph{Egalitarian Additive Cost.}
In \Cref{sec:egalitarian}, we turn our attention to the egalitarian social cost, where we focus on the maximum cost of a committee for a voter.
We consider the simple \kExtremes\ rule, which selects half of the committee from each extreme.
On an intuitive level, this constitutes a natural rule in this setting as it avoids that extreme voters are excessively disadvantaged.
For the additive setting, we show in \Cref{subsec:egalitarian-additive} that \kExtremes\ achieves an optimal distortion up to $O\big(\frac{1}{k}\big)$ terms.
In particular, the optimal distortion of $1$ is attained for $k=2$, and distortions of $\frac{3}{2}-\frac{1}{2(k-1)}$ and $\frac{3}{2}-\frac{1}{k(k-1)}$ are achieved for even and odd $k\geq 3$, respectively, almost matching a lower bound of $\frac{3}{2}-\frac{1}{k}$.
The worst-case instances involve $k+1$ agents in one extreme, a single agent in the other extreme, and $k$ agents in the middle, which are selected in the optimal committee but cannot be detected by any rule when considering two symmetric distance metrics. 

\paragraph{Egalitarian $q$-Cost.}
In \Cref{subsec:egalitarian-qcost}, we show that \kExtremes\ attains a distortion of $2$ for $q$-cost as long as $q>\frac{k}{3}$.
To do so, we prove that the social cost of the set selected by this rule is at most the distance from the agent closest to the center to their nearest extreme, and bound the social cost of the optimal set from below by half of this distance.
We provide a matching lower bound by revisiting the instance used for the additive case.
Finally, we show that no constant distortion is possible when $q\leq \frac{k}{3}$, again implying that the general impossibility result of \citet{caragiannis2022metric} still holds in our setting.
In the worst-case instances, we partition the agents into $\big\lfloor \frac{k}{q}\big\rfloor$ sets and consider two symmetric distance metrics where all but one set are placed at a unit distance from one another and two sets in one extreme are at the same location.
We show that no rule can pick $q$ agents from each location.

\subsection{Further Related Work}
Distortion of voting rules was first introduced by \citet{procaccia2006distortion}.
Since then, extensive research has been conducted to establish lower and upper bounds on the distortion of different rules under various scenarios, both within the metric and non-metric frameworks.
For a comprehensive survey, we refer to \citet{anshelevich2021distortion}. 

\paragraph{Single-Winner Voting.}
In the non-metric framework, \citet{caragiannis2011voting} showed that the distortion of any voting rule is at least $\Omega(m^2)$ and that simple rules such as Plurality achieve a distortion of at most $O(m^2)$, where $m$ is the number of candidates.
In the metric framework, \citet{ANSHELEVICH201827} established a general lower bound of $3$ on the distortion of any deterministic voting rule.
They also analyzed the distortion of common voting rules such as Majority, Borda, and Copeland, showing that the latter achieves the lowest distortion of $5$ among them.
\citet{goel2017metric} disproved a conjecture by \citeauthor{ANSHELEVICH201827} regarding a better-than-$5$ distortion of the Ranked Pairs rule and introduced the notion of \textit{fairness ratio} of a rule, which captures the egalitarian social cost as a special case.
These results were later improved by \citet{munagala2019improved}, who extended the analysis to uncovered set rules and reduced the upper bound to $4.236$.
\citet{gkatzelis2020resolving} closed the gap by improving this bound to $3$, and showed the validity of this bound in terms of fairness ratio and thus egalitarian social cost.
Randomized voting rules have also been extensively explored in the metric framework \citep{pulyassary2021randomized, fain2019random}.
The best-known upper bound for a randomized voting rule was recently obtained by \citet{charikar2024breaking}, who showed that a carefully designed randomization over existing and novel voting rules achieves a distortion of at most $2.753$.
As for lower bounds, \citet{charikar2022metric} disproved a conjecture by \citet{goel2017metric} regarding the existence of a randomized voting rule with distortion $2$, by constructing instances whose distortion approaches $2.113$ as the number of candidates grows.

\paragraph{Multi-Winner Voting.}
In the study of metric distortion for multi-winner voting, various objective functions have been proposed to capture the cost incurred by each voter for the elected committee~\citep{elkind2017a, faliszewski2017multiwinner}.
A foundational result by \citet{10.1145/3230654.3230658} showed that, for the additive cost function, iterating a single-winner voting rule with distortion $\delta$ for $k$ rounds produces a $k$-winner committee with the same distortion.
\citet{Chen_Li_Wang_2020} studied the $1$-cost objective in the metric framework when each voter casts a vote for a single candidate.
They proposed a deterministic rule with a tight distortion of $3$ and a randomized rule with a distortion of $3 - \frac{2}{m}$.
More generally, \citet{caragiannis2022metric} introduced the $q$-cost objective, where a voter's cost for a committee is determined by the distance to their $q$th closest member.
They showed that the distortion is unbounded for $q \leq \frac{k}{3}$ and linear in $n$ for $\frac{k}{3} < q \leq \frac{k}{2}$.
For $q > \frac{k}{2}$, they presented a non-polynomial voting rule that achieves a distortion of $3$ and a polynomial rule with a distortion of $9$.
They discussed how these upper bounds for $q > \frac{k}{2}$ and the unbounded distortion for $q\leq \frac{k}{3}$ carry over to egalitarian social cost, but interestingly showed that a constant distortion is possible for this objective when $\frac{k}{3} < q \leq \frac{k}{2}$.
\citet{kizilkaya2022plurality} later proposed a polynomial-time rule with a distortion of $3$.
Recently, \citet{babashah2024distortionmultiwinnerelectionsline} studied the distortion of multi-winner elections with additive cost on the line, devising a rule with a distortion of roughly $\frac{7}{3}$.
\citet{caragiannis2017subset} studied distortion in multi-winner voting for the non-metric framework, defining a voter's utility for a committee as the highest utility derived from any of its members.
They proposed a rule achieving a distortion of $1 + \frac{m(m-k)}{k}$ for deterministic committee selection when selecting $k$ out of $m$ candidates. 

\paragraph{Restricted Voting Settings.}
A specialized setting in metric voting studies single-peaked and $1$-Euclidean preferences, where both voters and alternatives are embedded on the real line \citep{black48grouDecision, moulin1980strategy, miyagawa2001locating, fotakis2016conference, fotakis2022distortion, VOUDOURIS2023266, ghodsi2019distortion}.
In particular, the work of \citet{DBLP:journals/corr/abs-2408-11755} investigated the distortion of deterministic algorithms for $k$-committee selection on the line under the $1$-cost objective, leveraging additional distance queries.

\paragraph{Mechanism Design in Committee Selection.}
Several recent studies have explored alternative models for committee selection.
The concept of stable committees and stable lotteries has been considered in various settings, focusing on fairness and individual incentives \citep{jiang2020approximately, DBLP:conf/ijcai/Borodin0L022}.
An active area of research in the last years has focused on impartial mechanisms, where agents approve a subset of other agents and the voting rule must incentivize truthful reports while selecting well-evaluated agents~\citep{alon2011sum,holzman2013impartial,fischer2015optimal,tamura2016characterizing,mackenzie2015symmetry,caragiannis2022impartial,cembrano2023single,cembrano2024impartial,bousquet2014near}.
Finally, another line of work investigates distortion when agents have known locations, enabling mechanisms to explicitly consider distances in selection \citep{kalayci2024proportional, pulyassary2022algorithm, Anshelevich21KnownLocations}.

%% file: a-Preliminaries.tex
We let $\NN$ denote the strictly positive integers and, for $n\in \NN$, we write $[n]=\{1,\ldots,n\}$ for the first $n$.
A \textit{linear order} $\succ$ on a set $S$ is a complete, transitive, and antisymmetric binary relation on $S$; we denote the set of all linear orders on $[n]$ by $\calL(n)$.

\paragraph{Election.}
 An instance of a committee election, or simply an \textit{election} is described by the triple $\instance = (A, k, \succ)$, where:
\begin{itemize}
    \item $A=[n]$ is the set of agents,
    \item $k\in \NN$ is the number of agents to be selected for the committee, and
    \item $\ve{\succ} = (\succ_1,\succ_2,\ldots,\succ_n)\in \calL^n(n)$ comprises the agents' preference profiles, where $\succ_a\in \calL(n)$ is a linear order on $[n]$ for every $a\in [n]$.
\end{itemize}
We let ${A\choose k} = \{S\subseteq A \mid |S|=k\}$ denote the feasible committees for a given election; i.e., the set of all subsets of $A$ of size $k$.

\paragraph{Line metric.} A \textit{distance metric} on $A$ is a function $d\colon A\times A\to \RR_+$ satisfying (i) $d(a,a)=0$, 
(ii) $d(a,b)=d(b,a)$ for every $a,b\in A$, and (iii) $d(a,c)\leq d(a,b)+d(b,c)$ for every $a,b,c\in A$.
In this paper, we focus on the line metric:
We associate each agent $a\in A$ with a position $x_a\in (-\infty,\infty)$, and the metric $d$ is defined by $d(a,b)=|x_a-x_b|$ for every $a,b\in A$.
A metric $d$ is said to be \textit{consistent} with a ranking profile $\succ\ \in \calL^n(n)$, denoted as $d\ \rhd \succ$, if for every triple of agents $a,b,c\in A$, the condition $d(a, b) < d(a, c)$ implies $b \succ_a c$.\footnote{Note that this definition allows for agent-dependent tie-breaking; i.e., when $d(a, b) = d(a, c)$ agent $a$ can rank either $b \succ_a c$ or $c \succ_a b$, independently of other agents. This assumption makes the problem in principle harder, so that our upper bounds on the distortion remain valid if a common tie-breaking rule is employed, and it allows us to construct simpler examples for lower bounds. It is not hard to see that the same lower bounds can be obtained without the assumption: Whenever a metric has ties, distances can be perturbed by a small enough constant $\varepsilon$ so that there are no longer ties and the distortion does not improve.}
Since $d$ is fully defined by the position vector $x\in (-\infty,\infty)^A$, we often refer directly to this vector being consistent with a ranking profile $\succ\ \in \calL^n(n)$ and denote it by $x\ \rhd \succ$.
Likewise, we often exchange $d$ by $x$ in the definitions that follow.
Finally, for a fixed election $\instance=(A,k,\succ)$, consistent vector of locations $x\in (-\infty,\infty)^n$, and interval $I=(y,z)$ with $y < z$, we let $A(I)=\{a\in A\mid x_a\in I\}$ denote the agents with locations in $I$. 
When $I$ is a single point $\bar{x}$, we write $A(\bar{x})$ for the agents located at this point.

 \paragraph{Social cost.} 
For a certain set of agents $A$, a committee size $k\in \NN$, and a \textit{candidate-aggregation function} $h\colon \RR^k_+\to \RR_+$, the \textit{cost} of $S\in {A\choose k}$ for agent $a\in A$ is simply $\SC(S,a;d)= h((d(a,b))_{b\in S})$.
For a set of agents $A$, a committee size $k\in \NN$, and a \textit{voter-aggregation function} $g\colon \RR^n_+\to \RR$, the \textit{social cost} of $S\in {A\choose k}$ is $\SC(S,A;d)= g((\SC(S,a;d))_{a\in A})$.
In this paper, we study a handful of candidate- and voter-aggregation functions.
In terms of the voter-aggregation function $g\colon \RR^n\to \RR_+$, we focus on the \textit{utilitarian social cost}, given by $g(y)=\sum_{i\in [n]} y_i$, and the \textit{egalitarian social cost}, given by $g(y) = \max\{y_i\mid i\in [n]\}$.
In terms of the candidate-aggregation function $h\colon \RR^k_+\to \RR_+$, we focus on the \textit{additive social cost}, given by $h(y)=\sum_{i\in [k]}y_i$, and the \textit{$q$-cost}, given by $h(y)=\tilde{y}_q$, where $\tilde{y}$ is the vector with the entries of $y$ sorted in increasing order.
Thus, for example, the \textit{$1$-cost} is given by $h(y)=\min\{y_i\mid i\in [k]\}$; and the \textit{$k$-cost} is given by $h(y)=\max\{y_i\mid i\in [k]\}$.

\paragraph{Voting rules and distortion.}
For $n,k\in \NN$ with $n\geq k$, an \textit{$(n,k)$-voting rule} is a function $f$ that takes a preference profile $\succ\ \in \calL^n(n)$ and returns a subset $S\in {[n]\choose k}$, to which we often refer as a \textit{committee}.
For an election $\instance=([n],k,\succ)$ and a metric $d$, the \textit{distortion $\dist(S,\instance;d)$ of $S\subseteq A$ under $d$} is the ratio between the social cost of the committee and the minimum social cost of any committee; i.e.,
\[
    \dist(S, \instance; d) = \frac{\SC(S,A;d)}{\min_{S' \in {A\choose k}} \SC(S',A;d)}.
\]
For an election $\instance=(A,k,\succ)$, the \textit{distortion $\dist(S,\instance)$ of a committee $S\subseteq A$} is then defined as the worst-case distortion over all metrics consistent with the ranking profile $\succ$; i.e.,
\[
    \dist(S, \instance) = \sup_{d \rhd \succ} \dist(S, \instance; d).
\]
Finally, for an $(n,k)$-voting rule $f$, the distortion of $f$ is defined as the worst-case distortion of its output across all possible elections; i.e.,
\[
\dist(f) = \sup_{\succ\in \calL^n(n)} \dist(f(\succ), ([n],k,\succ)).
\]
Throughout the paper, we study the distortion that voting rules can achieve under different social costs.

\subsection{Computing the Order From an Election}
An essential property in line metric settings is the ability to determine the order of agents based on their preferences. This result has been established in prior work. Specifically, \citet{elkind2014recognizing} and \citet{babashah2024distortionmultiwinnerelectionsline} demonstrate that if voters’ preference lists are pairwise distinct, it is possible to uniquely determine the ordering on the line, along with the ordering of the non-Pareto-dominated alternatives. While their setting differentiates between voters and alternatives, this result naturally extends to our context, where agents serve as both voters and candidates. 
In our setting, this follows directly from a simpler observation: for any three agents, their relative order on the line can be reconstructed from their preference rankings.
We state this result as a lemma, which serves as a foundation for many results in this paper as it guarantees that the order of agents in any election can be uniquely identified. 

\begin{lemma}[\citet{elkind2014recognizing}, \citet{babashah2024distortionmultiwinnerelectionsline}]\label{lem:agents-order}
For every election $\instance=([n],k,\succ)$, we can compute a permutation $\pi\colon [n]\to[n]$ of the agents such that, for any consistent position vector $x\in (-\infty,\infty)^n$ with $x\ \rhd \succ$, we have either 
$x_{\pi(1)} \leq x_{\pi(2)} \leq \dots \leq x_{\pi(n)}$ 
or 
$x_{\pi(n)} \leq x_{\pi(n-1)} \leq \dots \leq x_{\pi(1)}$.
\end{lemma}

We briefly highlight some structural features that distinguish the peer selection setting from the general setting. First, unlike the general voting setup with disjoint voters and candidates, the peer selection problem provides us with not only the order of the candidates but also the exact order of the voters. Second, this order can be efficiently reconstructed: we can identify extreme agents by observing who consistently appears at the bottom of others' rankings. Once such an agent is identified, their ranking determines the full order of agents along the line. Third, the mutual evaluations are more restricted, so we gain more information about relative preferences, such as how many voters prefer one candidate over another, than in the general voting setting. These features inform the design of improved mechanisms in peer selection and present new algorithmic and structural challenges.

For simplicity, whenever we fix an election throughout the paper, we will assume w.l.o.g., that the agents are already ordered, i.e., that the permutation $\pi$ stated in the lemma is the identity.
Hence, we denote the ordered agents by $1, \ldots, n$ and informally refer to this order as \textit{from left to right}.

%% file: a-Utilitarian.tex
Using Lemma~\ref{lem:agents-order}, we know that the order of agents can be fully determined from the preference profile $\ve{\succ}$. 
This allows us to identify the \textit{median agent}, who minimizes the total distance to all other agents. 
Hence, when selecting a single agent ($k=1$) under the utilitarian objective, the median agent is optimal. 
This observation holds regardless of whether the individual cost is defined additively or via the $q$-cost, since both notions coincide when $k=1$.

For larger committee sizes ($k > 1$), it becomes necessary to define how a voter's distances to the selected agents are aggregated. 
In this section, we study two aggregation rules: one that considers the sum of all distances to selected agents in \Cref{subsec:uti-add}, and one that considers the distance to the $q$th closest selected agent in \Cref{subsec:uti-q}.

\subsection{Utilitarian Additive Cost}\label{subsec:uti-add}

In this section, we focus on the \textit{utilitarian additive} objective for committee selection. This objective aims to minimize the \textit{utilitarian additive social cost}, which is defined as the total distance from all agents to the selected committee. Formally, the \textit{utilitarian additive social cost} of a committee $S' \in {A \choose k}$ is given by
\[
\SC(S', A; d) = \sum_{a \in A} \sum_{b \in S'} d(a, b).
\]
The cost of each agent $a \in A$ is the sum of their distances to all members of the selected committee $S'$, and the overall social cost is the sum of these individual costs across all agents in $A$.

It is straightforward to verify that the optimal committee can be directly identified from the preference profile when the committee size is both $k = 1$ and $k = 2$. 
For $k = 1$, the optimal committee consists of the median agent, as discussed above. 
For $k = 2$, the optimal committee consists of the two median agents if $n$ is even, or the median agent together with the agent closest to them, if $n$ is odd. 
In both cases, the optimal agents can be determined solely from $\ve{\succ}$, without knowledge of the underlying metric, yielding a distortion of $1$.

We now introduce our voting rule for selecting a committee of size $k \geq 2$, called the \medianAlt\ rule. 
This rule can be viewed as a generalization of the \textit{Polar Comparison} rule proposed by \citet{babashah2024distortionmultiwinnerelectionsline} for multi-winner elections. 
In contrast to the Polar Comparison rule, the \medianAlt\ rule (i) applies to all committee sizes $k$, (ii) is simpler since it does not require case distinctions or pairwise comparisons, and (iii) achieves improved distortion guarantees under peer selection.

\begin{Vrule}[\medianAlt]
Compute the order of the agents $1,\ldots,n$ on the line. 
If $n-k$ is even, return $S=\big\{ \tfrac{n-k}{2} + 1,\; \tfrac{n-k}{2} + 2,\; \dots,\; \tfrac{n+k}{2} \big\}$.
If $n-k$ odd and $\big\lceil \tfrac{n-k}{2} \big\rceil \succ_m \big\lceil \tfrac{n+k}{2} \big\rceil$, return $S=\big\{ \big\lceil \tfrac{n-k}{2} \big\rceil, \, \big\lceil \tfrac{n-k}{2} \big\rceil + 1, \, \dots, \, \big\lceil \tfrac{n+k}{2} \big\rceil - 1 \big\}$.
Otherwise, return $S=\big\{ \big\lceil \tfrac{n-k}{2} \big\rceil + 1, \, \big\lceil \tfrac{n-k}{2} \big\rceil + 2, \, \dots, \, \big\lceil \tfrac{n+k}{2} \big\rceil \big\}$.
\end{Vrule}

On an intuitive level, the rule can be understood as constructed by going through the rank list of the median(s) agent(s), selecting agents in the order reported by them but alternating between those to their left and to their right.
This ensures a balanced representation of agents on both sides.

For larger committees, a key ingredient in our analysis is that there always exists an optimal committee consisting of consecutive agents.
We formalize this observation in the following lemma.

\begin{restatable}{lemma}{lemUtilConsecutive}\label{lem:util-additive-consecutive}
For any election $\instance = (A, k, \succ)$ and consistent metric $d\ \rhd \succ$, there exists $i \in [n-k+1]$ such that, defining $S^* = \{i, i+1, \dots, i+k-1\}$, we have $\SC(S^*,A;d)=\min\big\{\SC(S',A;d)\mid S'\in {A\choose k}\big\}$.
\end{restatable}
\begin{proof}
Let $\instance=(A,k,\succ)$ with $A=[n]$ and $d$ be as in the statement, and let also $x\ \rhd \succ$ be a consistent position vector defining $d$.
The result is trivial if $k=1$, so we assume that $k\geq 2$ in what follows.
We first observe that, by Lemma~$2$ in \citet{babashah2024distortionmultiwinnerelectionsline}, $\SC(a,A;d)\leq \SC(b,A;d)$ holds for $a,b\in A$ are such that either (1) $a,b\geq \frac{n+1}{2}$ and $a-\frac{n+1}{2} \leq b-\frac{n+1}{2}$, or (2) $a,b\leq \frac{n+1}{2}$ and $\frac{n+1}{2} -a \leq \frac{n+1}{2}-b$.
In simple words, if two agents lie on the same side of the median agent(s), the agent closer to them has a lower cost. 
Thus, there exist $S^*\in {A\choose k}$ that minimizes the social cost such that $\{m_1,m_2\} \subseteq S^*$, where $m_1=\big\lfloor \frac{n+1}{2}\big\rfloor$ and $m_2=\big\lceil \frac{n+1}{2}\big\rceil$ denote the median agent(s) (note that $m_1=m_2$ if $n$ is odd).
Now, suppose that $S^*$ is not consecutive.
Since $m_1,m_2\in S^*$, there exists an agent $a \notin S^*$ and $b \in S^*$ such that either (1) $a,b\geq \frac{n+1}{2}$ and $a-\frac{n+1}{2} \leq b-\frac{n+1}{2}$, or (2) $a,b\leq \frac{n+1}{2}$ and $\frac{n+1}{2} -a \leq \frac{n+1}{2}-b$.
But then, using the result by \citeauthor{babashah2024distortionmultiwinnerelectionsline} again, we obtain that $\SC((S^*\setminus \{b\})\cup \{a\},A;d)\leq \SC(S^*,A;d)$; i.e., we can exchange $b$ by $a$ and the social cost of the committee does not increase.
By repeating this procedure, we reach a committee with consecutive agents and minimum social cost, as claimed in the statement.
\end{proof}
We now present our main result in terms of utilitarian additive social cost, regarding the distortion guaranteed by \medianAlt.  
\begin{restatable}{theorem}{thmUtilAdd}\label{thm:Util-Add}
    For every $n,k\in \NN$ with $n\ge k\ge 2$, under the utilitarian additive social cost the distortion of \medianAlt\ is at most
    \[
    \frac{1}{k}\left(2n+\chi-2\sqrt{\,n\bigl(n-k+\chi\bigr)}\right),
    \]
    where $\chi=1$ if $n-k$ is odd and $\chi=0$ otherwise.
\end{restatable}

The bound ranges between $1$ and $2$: it is closer to $1$ when $k$ is small relative to $n$, and approaches $2$ as $k$ approaches $n$.
\Cref{fig:plot-uti-add} illustrates the bound for $n=100$ and $k$ between $2$ and $n-1$. 
\begin{figure}
    \centering
\begin{tikzpicture}[xscale=0.07,yscale=0.07]
\draw[thick,-latex] node[left] {$1$} (-1,0) -- (105,0) node[right] {$k$};
\draw[thick,-latex] node[below] {$0$} (0,-1) -- (0,55) node[above] {distortion};
\foreach \n in {10,20,30,40,50,60,70,80,90,100}{
    \draw (\n,1) -- (\n,-1) node[below] {$\n$};
}
\draw (1,10) -- (-1,10) node[left] {$1.2$};
\draw (1,20) -- (-1,20) node[left] {$1.4$};
\draw (1,30) -- (-1,30) node[left] {$1.6$};
\draw (1,40) -- (-1,40) node[left] {$1.8$};
\draw (1,50) -- (-1,50) node[left] {$2$};

\draw[semithick,blue] plot[mark=., mark color=none,mark options={fill=black,draw opacity=0},mark size=3pt] coordinates { 
(2, 0.2525) (3, 0.1684) (4, 0.5103) (5, 0.4082) (6, 0.7734) (7, 0.6629) (8, 1.0421) (9, 0.9263) (10, 1.3167) (11, 1.197) (12, 1.5974) (13, 1.4745) (14, 1.8844) (15, 1.7588) (16, 2.178) (17, 2.0499) (18, 2.4786) (19, 2.3482) (20, 2.7864) (21, 2.6537) (22, 3.1018) (23, 2.9669) (24, 3.4251) (25, 3.2881) (26, 3.7567) (27, 3.6176) (28, 4.0971) (29, 3.9558) (30, 4.4467) (31, 4.3032) (32, 4.8059) (33, 4.6603) (34, 5.1753) (35, 5.0275) (36, 5.5556) (37, 5.4054) (38, 5.9472) (39, 5.7947) (40, 6.3508) (41, 6.1959) (42, 6.7673) (43, 6.6099) (44, 7.1974) (45, 7.0374) (46, 7.642) (47, 7.4794) (48, 8.102) (49, 7.9367) (50, 8.5786) (51, 8.4104) (52, 9.073) (53, 8.9018) (54, 9.5865) (55, 9.4122) (56, 10.1205) (57, 9.943) (58, 10.6769) (59, 10.4959) (60, 11.2574) (61, 11.0729) (62, 11.8643) (63, 11.676) (64, 12.5) (65, 12.3077) (66, 13.1674) (67, 12.9709) (68, 13.8698) (69, 13.6688) (70, 14.6111) (71, 14.4053) (72, 15.3958) (73, 15.1849) (74, 16.2295) (75, 16.0131) (76, 17.1187) (77, 16.8964) (78, 18.0716) (79, 17.8428) (80, 19.0983) (81, 18.8625) (82, 20.2117) (83, 19.9682) (84, 21.4286) (85, 21.1765) (86, 22.7714) (87, 22.5097) (88, 24.2716) (89, 23.9989) (90, 25.9747) (91, 25.6893) (92, 27.9519) (93, 27.6513) (94, 30.3246) (95, 30.0054) (96, 33.3333) (97, 32.9897) (98, 37.6101) (99, 37.2302) };
\end{tikzpicture}
    \caption{Distortion of \medianAlt\ stated in \Cref{thm:Util-Add} for $n=100$ and $k\in \{2,\ldots,99\}$.}
    \label{fig:plot-uti-add}
\end{figure}

In order to prove \Cref{thm:Util-Add}, we will show that we can reduce any metric to another one where all agents are in one out of two locations.
As a first step, we prove that an agent (or set of agents at the same location) can always be moved in one direction such that the distortion does not improve, as long as they do not pass through other agents' locations.
To this end, for a position vector $x\in (-\infty,\infty)^n$, a position $\bar{x}\in (-\infty,\infty)$ such that $A(\bar{x})\neq \emptyset$, and $\delta>0$, we define the \textit{shifted position vectors} $x^-(\bar{x},\delta),x^+(\bar{x},\delta)\in (-\infty,\infty)^n$ as follows:
\begin{align*}
    x^-_a(\bar{x},\delta)& =x_a-\delta \text{ for every } a\in A(\bar{x}),\quad x^-_a(\bar{x},\delta)=x_a \text{ for every } a\in A\setminus A(\bar{x}),\\
    x^+_a(\bar{x},\delta)& =x_a+\delta \text{ for every } a\in A(\bar{x}),\quad x^+_a(\bar{x},\delta)=x_a \text{ for every } a\in A\setminus A(\bar{x}).
\end{align*}

\begin{restatable}{lemma}{lemAgentMove}\label{lem:agent-move-distortion}
Let $\instance=(A,k,\succ)$ be an election with $A=[n]$, let $S\in {A\choose k}$ be the committee selected by \medianAlt\ on this election, and let $x\in (-\infty,\infty)^n$ with $x\ \rhd \succ$ be a consistent position vector.
Let $\bar{x}\in (-\infty,\infty)$ be such that $A(\bar{x})\neq \emptyset$, let $\delta>0$ be such that $A((\bar{x}-\delta,\bar{x}+\delta))=A(\bar{x})$ and let $x^-=x^-(\bar{x},\delta)$ and $x^+=x^+(\bar{x},\delta)$.
Then, for all preference profiles $\succ^-,\succ^+$ such that $x^-\ \rhd \succ^-$ and $x^+\ \rhd \succ^+$, at least one of the following inequalities holds: 
\[
    \dist(S,(A,k,\succ^-);x^-)\geq \dist(S,\instance;x),\quad \text{or}\quad \dist(S,(A,k,\succ^+);x^+)\geq \dist(S,\instance;x).
\]
\end{restatable}

The proof of this lemma relies on the linearity of the objective function: If moving an agent or set of agents to the right has a certain effect on the social cost, moving them to the left has the opposite effect.
Then, the ratio between the social cost of any two fixed committees must not improve in one of these directions.
Since the committee selected by \medianAlt\ remains fixed as long as the order of agents does not change, and changing the optimal set can only lead to a worse distortion, the result follows.
We now proceed with the formal proof.

\begin{proof}[Proof of \Cref{lem:agent-move-distortion}]
    Let $\instance = (A, k, \succ)$, $S$, $x$, $\bar{x}$, $\delta$, $x^-$, $x^+$, $\succ^-$, and $\succ^+$ be as in the statement.
    We denote by $d$, $d^-$, and $d^+$ the distance metrics associated to $x$, $x^-$, and $x^+$, respectively.

    We first consider an arbitrary committee $S'\in {A\choose k}$ and compute the difference between the social cost of this committee under metric $d$ and under both of the other metrics.
    From the definition of the additive social cost, for any $a\in A$ such that $x_a< \bar{x}$ we have that
    \begin{align}
        \SC(S',a;x^-) & = \sum_{b\in S'\cap A(\bar{x})}d^-(a,b) + \sum_{b\in S'\setminus A(\bar{x})}d^-(a,b)\nonumber \\
        & =\sum_{b\in S'\cap A(\bar{x})}(d(a,b)-\delta) + \sum_{b\in S'\setminus A(\bar{x})}d(a,b)\nonumber\\
        & = \SC(S',a;x) - \delta\, |S'\cap A(\bar{x})|.\label{eq:diff-sc-d-dmin-left}
    \end{align}
    Similarly, for any $a\in A$ such that $x_a> \bar{x}$ we have that
    \begin{align}
        \SC(S',a;x^-) & = \sum_{b\in S'\cap A(\bar{x})}d^-(a,b) + \sum_{b\in S'\setminus A(\bar{x})}d^-(a,b)\nonumber \\
        & =\sum_{b\in S'\cap A(\bar{x})}(d(a,b)+\delta) + \sum_{b\in S'\setminus A(\bar{x})}d(a,b)\nonumber\\
        & = \SC(S',a;x) + \delta\, |S'\cap A(\bar{x})|.\label{eq:diff-sc-d-dmin-right}
    \end{align}
    Finally, for every $a$ with $x_a=\bar{x}$, i.e., $a\in A(\bar{x})$, we have that
    \begin{align}
        \SC(S',a;x^-) & = \sum_{b\in S'\cap A((-\infty,\bar{x}))}d^-(a,b) + \sum_{b\in S'\cap A((\bar{x},+\infty))}d^-(a,b)\nonumber \\
        & =\sum_{b\in S'\cap A((-\infty,\bar{x}))}(d(a,b)-\delta) + \sum_{b\in S'\cap A((\bar{x},+\infty))}(d^-(a,b)+\delta) \nonumber\\
        & = \SC(S',a;x) + \delta\, (|S'\cap A((\bar{x},+\infty))|-|S'\cap A((-\infty,\bar{x}))|).\label{eq:diff-sc-d-dmin-xbar}
    \end{align}
    Combining \cref{eq:diff-sc-d-dmin-left,eq:diff-sc-d-dmin-right,eq:diff-sc-d-dmin-xbar}, we obtain from the definition of utilitarian social cost that
    \begin{align*}
        \SC(S',A;x^-) & = \sum_{a\in A}\SC(S',a;d^-) \\
        & = \SC(S',A;x) - \delta\, |S'\cap A(\bar{x})| \,\big(|A(-\infty,\bar{x})|-|A(\bar{x},+\infty)|\big) \\
        & \phantom{{}={}} - \delta\,|A(\bar{x})|\, \big(|S'\cap A((-\infty,\bar{x}))|-|S'\cap A((\bar{x},+\infty))|\big).
    \end{align*}
    One can proceed analogously for $d^+$ to obtain
    \begin{align*}
        \SC(S',A;x^+) & = \SC(S',A;x) + \delta\, |S'\cap A(\bar{x})| \,\big(|A(-\infty,\bar{x})|-|A(\bar{x},+\infty)|\big) \\
        & \phantom{{}={}} + \delta\,|A(\bar{x})|\, \big(|S'\cap A((-\infty,\bar{x}))|-|S'\cap A((\bar{x},+\infty))|\big).
    \end{align*}
    Hence, there exists a value $\Delta(S')$, that only depends on the committee $\delta$, such that
    \begin{equation}
        \SC(S',A;x^-)=\SC(S',A;x)-\Delta(S'),\qquad \SC(S',A;x^+)=\SC(S',A;x)+\Delta(S').\label{eq:diff-committees}
    \end{equation}
    
    We let $S^*$ denote an optimal committee for the metric $d$ in what follows, i.e., a committee such that $\SC(S^*,A;x)=\min\big\{\SC(S',A;x)\mid S'\in {A\choose k}\big\}$.
    We observe that
    \begin{equation}
        \dist(S,(A,k,\succ^-);x^-) = \frac{\SC(S,A;x^-)}{\min_{S'\in {A\choose k}}\SC(S',A;x^-)} \geq \frac{\SC(S,A;x^-)}{\SC(S^*,A;x^-)} = \frac{\SC(S,A;x)-\Delta(S)}{\SC(S^*,A;x)-\Delta(S^*)},\label{eq:dist-xmin-lb}
    \end{equation}
    and
    \begin{equation}
        \dist(S,(A,k,\succ^-);x^+) = \frac{\SC(S,A;x^+)}{\min_{S'\in {A\choose k}}\SC(S',A;x^+)} \geq \frac{\SC(S,A;x^+)}{\SC(S^*,A;x^+)} = \frac{\SC(S,A;x)+\Delta(S)}{\SC(S^*,A;x)+\Delta(S^*)}.\label{eq:dist-xplus-lb}
    \end{equation}
    If either $\SC(S^*,A;x)=\Delta(S^*)$ or $\SC(S^*,A;x)=-\Delta(S^*)$ holds, the distortion becomes unbounded in one of the new instances and the result follows directly. Otherwise, it follows from the simple property stated in the following claim.
    \begin{claim}
        For any values $y,z\in \RR_+$ and $w\in (-z,z)$, we have either $\frac{y+w}{z+w} \geq \frac{y}{z}$ or $\frac{y-w}{z-w} \geq \frac{y}{z}$.
    \end{claim}
    \begin{proof}
        Suppose towards a contradiction that both $\frac{y+w}{z+w} < \frac{y}{z}$ and $\frac{y-w}{z-w} < \frac{y}{z}$ hold.
        Since $w<z$, the first inequality is equivalent to 
        \[
            z(y+w)< y(z+w) \Longleftrightarrow zw<yw.
        \]
        Since $w>-z$, the second inequality is equivalent to 
        \[
            z(y-w)< y(z-w) \Longleftrightarrow yw<zw.
        \]
        As the inequalities contradict each other, we conclude.
    \end{proof}
    
    Applying these properties to inequalities \eqref{eq:dist-xmin-lb} and \eqref{eq:dist-xplus-lb}, we obtain that either  
    \[
        \dist(S,(A,k,\succ^-);x^+) \geq \frac{\SC(S,A;x)+\Delta(S)}{\SC(S^*,A;x)+\Delta(S^*)} \geq \frac{\SC(S,A;x)}{\SC(S^*,A;x)} = \dist(S,\instance;x)
    \]
    or 
    \[
        \dist(S,(A,k,\succ^-);x^-) \geq \frac{\SC(S,A;x)-\Delta(S)}{\SC(S^*,A;x)-\Delta(S^*)} \geq \frac{\SC(S,A;x)}{\SC(S^*,A;x)} = \dist(S,\instance;x)
    \]
    holds, concluding the proof.
\end{proof}

We can use the previous lemma to conclude that, for every election and consistent metric, \medianAlt\ selects a committee such that, under another metric with only two locations, the distortion does not improve.
Indeed, we can iterating the argument in \Cref{lem:agent-move-distortion} to move (sets of) agents in non-extreme positions in their non-improving direction.
This procedure terminates with all agents in one of the original extreme positions $x_1$ or $x_n$ and that the distortion has not improved.
The following lemma formally states this fact.
\begin{restatable}{lemma}{lemWorstDistortion}\label{lem:worst-distortion}
    Let $\instance=(A,k,\succ)$ be an election with $A=[n]$, let $S\in {A\choose k}$ be the committee selected by \medianAlt\ on this election, and let $x\in (-\infty,\infty)^n$ with $x\ \rhd \succ$ be a consistent position vector.
    Then, there exists a position vector $x'\in (-\infty,\infty)^n$ such that $x'_a\in \{x_1,x_n\}$ for every $a\in A$ and $\dist(S',(A,k,\succ');x') \geq \dist(S,\instance,x)$, where $\succ'$ is any preference profile such that $x'\ \rhd \succ'$ and $S'\in {A\choose k}$ is the committee selected by \medianAlt\ on the election $(A,k,\succ')$. 
\end{restatable}
\begin{proof}
    Let $\instance = (A, k, \succ)$ and $x$ be as in the statement, where, as usual, $x_1$ and $x_n$ represent the positions of the two extreme agents. 
    To construct $x'$ as claimed in the statement, we iteratively move agents toward the positions of the extreme agents using Lemma~\ref{lem:agent-move-distortion}.
    Specifically, we initialize $x'=x$ and, as long as $x'_a\in (x_1,x_n)$ for some $a\in A$, we fix $\bar{x}=x_a$, we define
    \[
        \delta^*=\max\{\delta>0\mid A((\bar{x}-\delta,\bar{x}+\delta))=A(\bar{x})\},
    \]
    and we update $x'_b\gets x'_b\pm \delta^*$ for every $b\in A(\bar{x})$ and the sign that ensures not increasing the distortion $\dist(S,A;x')$ of $S$.
    Note that the definition of $\delta^*$ ensures both the existence of this sign, due to \Cref{lem:agent-move-distortion}, and the fact that the number of different positions $|\{y\in (-\infty,\infty)\mid \exists a\in [n]: x'_a=y\}|$ is reduced in each step.
    Thus, the procedure terminates with a vector $x'\in (-\infty,\infty)$ such that (1) $x'_a\in \{x_1,x_n\}$ for every $a\in A$, and (2) the distortion of $S$ under the resulting metric has not decreased.
    Note that, since the order of the agents has not been changed besides ties, we have either $S'=S$ if the committee selected by \medianAlt\ has not changed or $S'\neq S$ but $\SC(S',A,x')=\SC(S,A;x')$ if the committee has changed due to a different tie-breaking.
\end{proof}

We now proceed with the proof of \Cref{thm:Util-Add}. 

\begin{proof}[Proof of Theorem \ref{thm:Util-Add}]
    Let $\instance = (A, k, \succ)$ be an arbitrary election, where $A = [n]$ is the set of agents. Let $d\ \rhd \succ$ be any consistent distance metric  induced by positions $x\in (-\infty,\infty)^n$, and let $S$ denote the committee selected by \medianAlt\ on this election.
    From \Cref{lem:worst-distortion}, we know that there exists a new position vector $x'\in (-\infty,\infty)^n$ and associated election $\instance' = (A, k, \succ')$, with $x'\ \rhd \succ'$, such that where all agents are positioned at the two extreme positions of the original instance and the distortion in $\instance'$ is at least as bad as the distortion in $\instance$; i.e., $x'_a\in \{x_1,x_n\}$ for every $a\in A$ and
    $\dist(S',(A,k,\succ');x') \geq \dist(S,\instance,x)$, where $S'$ denotes the committee selected by \medianAlt\ on $\instance'$.
    Thus, it suffices to compute the distortion for this election $\instance'$ to bound the distortion of the voting rule.
    As usual, we denote by $d'$ the metric induced by the position vector $x'$.

    We partition the set of agents into two groups, $A=A_1\dot{\cup}A_n$, where
    \[
    A_1 = \{a \in A \mid x'_a = x_1\} \text{ and } A_n = \{a \in A \mid x'_a = x_n\}
    \]
    denote the sets of agents located at positions $x_1$ and $x_n$ under the position vector $x'$, respectively. 
    We let $S_1=S'\cap A_1$ and $S_n=S'\cap A_n$ denote the agents selected by \medianAlt\ on $\instance'$ from agents in $A_1$ and $A_n$, respectively.
    Then, the social cost of $S'$ is given by
    \begin{align*}
    \SC(S', A; d') & = \sum_{a \in A_1} \sum_{b \in S'} d'(x_1, x_b) + \sum_{a \in A_n} \sum_{b \in S'} d'(x_n, x_b) \\
    & = |A_1| \cdot |S_n| \cdot d'(x_1, x_n) + |A_n| \cdot |S_1| \cdot d'(x_1, x_n).
    \end{align*} 

    On the other hand, the optimal committee $S^*$ clearly minimizes the total social cost by selecting as many agents as possible from the larger group between $A_1$ and $A_n$, as this cost is only incurred by agents in the smaller set. 
    We suppose that $|A_n| \geq |A_1|$ w.l.o.g. 
    We have two cases: either $|A_n| \geq k$ or $|A_n| < k$. 
    In the former case,
    \[
    \SC(S^*, A; d) = |A_1| \cdot k \cdot d'(x_1, x_n),
    \]
    while in the latter case,
    \[
    \SC(S^*, A; d) = |A_1| \cdot |A_n| \cdot d'(x_1, x_n) + |A_n| \cdot (k - |A_n|) \cdot d'(x_1, x_n).
    \]
    Since $|A_n|\geq |A_1|$ implies
    \[
    |A_1| \cdot |A_n| \cdot d'(x_1, x_n) + |A_n| \cdot (k - |A_n|) \cdot d'(x_1, x_n) \geq |A_1| \cdot k \cdot d'(x_1, x_n),
    \]
    the social cost induced by $S^*$ is smaller when $|A_n| \geq k$ and it suffices to bound the distortion in this case.
    Therefore,
    \begin{align}
    \dist(f) & \leq \frac{\SC(S, A; d')}{\SC(S^*, A; d')}\nonumber\\
    & = \frac{|A_1| \cdot |S_n| \cdot d'(x_1, x_n) + |A_n| \cdot |S_1| \cdot d'(x_1, x_n)}{|A_1| \cdot k \cdot d'(x_1, x_n)} \nonumber\\
    & = \frac{|A_1| \cdot |S_n| + |A_n| \cdot |S_1|}{|A_1| \cdot k}.\label{eq:distortion-uti-add-1}
    \end{align}
    If $|S_n|=k$, we obtain $\dist(f)=1$.
    In what follows, we thus assume $S_1\neq \emptyset$.
    From the definition of the \medianAlt\ voting rule, we know that $|A_n| - |S_n| = |A_1| - |S_1|$ if $n-k$ is even, and $|A_n| - |S_n| = |A_1| - |S_1|-1$ if $n-k$ is odd. 
    
    We can express the previous equations simply as 
    \[
        |A_n| - |S_n| = |A_1| - |S_1| - \chi.
    \]
    From this equality, alongside $|A_1|+|A_n| = n$ and $|S_1|+|S_n|=k$, we can express all $|A_1|$, $|S_1|$, and $|S_n|$ in terms of $|A_n|$ as follows:
    \[
        |A_1|=n-|A_n|,\quad |S_1|=\frac{n+k-\chi}{2}-|A_n|,\quad |S_n|=|A_n|-\frac{n-k-\chi}{2}.
    \]
    
    Replacing in \cref{eq:distortion-uti-add-1}, we obtain
    \begin{align}
        \dist(f) & \leq \frac{(n-|A_n|)\big(|A_n|-\frac{n-k-\chi}{2}\big) + |A_n|\big(\frac{n+k-\chi}{2}-|A_n| \big)}{(n-|A_n|)k}\nonumber\\
        & =\frac{2|A_n|}{k}-\frac{\chi |A_n|}{k(n-|A_n|)} - \frac{n(n-k-\chi)}{2k(n-|A_n|)}\nonumber \\
        & = h(|A_n|),\label{eq:distortion-uti-add-2}
    \end{align}
    where we have defined a function $h\colon \big\{\big\lceil \frac{n}{2}\big\rceil,\ldots, n-1\big\} \to \RR$, which evaluated at $|A_n|$ gives the last expression.
    Its first and second derivatives are given by
    \[
        h'(y) = \frac{1}{k}\bigg(2-\frac{(n-k+\chi)n}{2(n-y)^2}\bigg),\qquad h''(y) = -\frac{(n-k+\chi)n}{k(n-y)^3}.
    \]
    Since $h''(y) \leq 0$ for every $y$ in the domain of $h$, an upper bound for the value of $h$ is given by its value at $y^*$, where $y^*$ is such that
    \[
        h'(y^*)=0 \Longleftrightarrow y^* = n - \frac{1}{2}\sqrt{(n-k+\chi)n}.
    \]
    Combining this fact with \cref{eq:distortion-uti-add-2}, we conclude that
    \begin{align*}
        \dist(f) & \leq h(y^*) \\
        & = \frac{2}{k}\bigg(n - \frac{1}{2}\sqrt{(n-k+\chi)n}\bigg) - \frac{\chi \big(n-\frac{1}{2}\sqrt{(n-k+\chi)n}\big)}{\frac{k}{2}\sqrt{(n-k+\chi)n}}\\
        & \phantom{{}={}}- \frac{n(n-k-\chi)}{k\sqrt{(n-k+\chi)n}} \\
        & = \frac{1}{k}\Big(2n+\chi-2\sqrt{(n-k+\chi)n}\Big),
    \end{align*}
    which is the same as the expression in the statement.
\end{proof}

In terms of lower bounds, we prove that for any $k\geq 3$, there exists $n$ such that any $(n,k)$-voting rule has distortion at least $1.0914$.
To do so, for any fixed $k$ we consider $n$ agents partitioned into sets $A_1,A_2,A_3$, with $|A_3|\leq |A_1|\leq k$ and $|A_2|+|A_3|=k$.
We consider two metrics depicted in \Cref{fig:uti-add-lb}, both with the agents in the same set located at the same point and $x_a\leq x_b\leq x_c$ for any $a\in A_1,b\in A_2$, and $c\in A_3$.
In one metric, $A_2$ is equidistant to $A_1$ and $A_3$; in the other, $A_2$ and $A_3$ are located at the same point.
Intuitively, a rule performs badly in the first instance if it selects many agents from $A_3$ and in the second instance otherwise.
We formalize this distinction through a threshold, obtain bounds parameterized on this threshold for each of the corresponding cases, and obtain the overall bound by optimally picking the threshold and the number of agents in each set that maximizes the minimum distortion between the two cases.

\begin{restatable}{theorem}{thmUtilitarianAdditiveLB}\label{thm:utilitarian-add-lb}
    For every $k\in \NN$ with $k\geq 3$, there exists $n\in \NN$ with $n\geq k$ such that, for every $(n,k)$-voting rule $f$, $dist(f)\geq 1.0914$ for utilitarian additive social cost.
\end{restatable}
\begin{proof}
Consider an arbitrary $(n,k)$-voting rule $f$.
We partition the agents into three sets $A_1,A_2,A_3$ with $|A_1|=n_1$, $|A_2|=n_2$, and $|A_3|=k-n_2$. The natural values $n_1$ and $n_2$ will be fixed later, respecting $n_2\leq k$ and $k-n_2\leq n_1\leq k$.
The total number of agents is then $n=n_1+k$.
We consider the profile $\succ\in \calL^n(n)$, where $S = f(\succ)$, and
\begin{enumerate}[label=(\roman*)]
    \item $b\succ_a c$ whenever $a\in A_i,b\in A_j,c\in A_\ell$ for some $i,j,\ell\in [3]$ with $|i-j|<|i-\ell|$;
    \item $b \succ_a c$ whenever $a\in A_2,b\in A_3,c\in A_1$;
\end{enumerate}
and the remaining pairwise comparisons are arbitrary.
We consider the election $\instance=(A,k,\succ)$ with $A=[n]$.

In what follows, we distinguish whether or not $f$ selects $k'$ or more agents from $A_3$, where $k'\in [k-n_2]$ will be fixed later, and construct appropriate distance metrics to show that, in either case, the distortion is at least the one claimed in the statement.
Intuitively, if $f$ selects many agents from $S_3$ (more than $k'$), we will consider $A_1$, $A_2$, and $A_3$ from left to right, each at a distance of $1$ from each other, so that picking all agents from $A_2$ and some agents from $A_1$ would lead to a lower or equal social cost (because $|A_1|\geq |A_3|$).
If $f$ selects few agents from $A_3$ (less than $k'$), we will consider $A_2$ and $A_3$ at the same location and a distance of $1$ from $A_1$, so that picking all $k$ agents from $A_2\cup A_3$ would lead to a lower or equal social cost (because $|A_1|\leq |A_2\cup A_3|$).

Formally, we first consider the case with $|S\cap A_3| \geq k'$ and define the distance metric $d_1$ on $A$ by the following positions $x\in (-\infty,\infty)^n$: $x_a=0$ for every $a\in A_1$, $x_a=1$ for every $a\in A_2$, and $x_a=2$ for every $a\in A_3$.
Note that $d_1\ \rhd \succ$; see \Cref{fig:uti-add-lb} for an illustration.
Since $|A_1|=n_1 \geq k-n_2=|A_3|$ by assumption, defining $T\subset A$ such that $|A_1\cap T|=k-k'-n_2$, $|A_2\cap T|=n_2$, and $|A_3\cap T|=k'$, we know that
\[
    \SC(S,A;d_1)\geq \SC(T,A;d_1) = n_1(n_2+2k')+n_2(k-n_2)+(k-n_2)(n_2+2(k-k'-n_2)).
\]
If we consider the alternative committee $S'$ such that $|A_1\cap S'|=k-n_2$, $|A_2\cap S'|=n_2$, and $|A_3\cap S'|=0$, we have $\SC(S',a;d_1)=n_2$ for every $a\in A_1$, $\SC(S',a;d_1)=k-n_2$ for every $a\in A_2$, and $\SC(S',a;d_1)=n_2+2(k-n_2)$ for every $a\in A_3$.
We obtain
\begin{align}
    \dist(f(\succ),\instance) \geq \frac{\SC(S,A;d_1)}{\SC(S',A;d_1)} & \geq \frac{n_1(n_2+2k')+n_2(k-n_2)+(k-n_2)(n_2+2(k-k'-n_2))}{n_1n_2+n_2(k-n_2)+(k-n_2)(n_2+2(k-n_2))}\nonumber\\
    & = 1 + \frac{2k'(n_1+n_2-k)}{n_1n_2 + n_2(k-n_2) + (k-n_2)(n_2+2(k-n_2))}\nonumber \\
    & = 1+ \frac{2k'(n_1+n_2-k)}{n_1n_2-2kn_2+2k^2}.\label{eq:lb-uti-add-ins1}
\end{align}
\begin{figure}
\centering
\begin{tikzpicture}[scale=1, every node/.style={font=\footnotesize}]
    \node[anchor=east] at (-0.5, 0) {\textbf{Metric $d_1$}};
    \filldraw[black] (0,0) circle (2.5pt) node[anchor=north] {$A_1$};
    \filldraw[black] (2.5,0) circle (2.5pt) node[anchor=north] {$A_2$};
    \filldraw[black] (5,0) circle (2.5pt) node[anchor=north] {$A_3$};
    \draw[-] (0,0) -- (2.5,0) node[midway, anchor=south] {$1$};
    \draw[-] (2.5,0) -- (5,0) node[midway, anchor=south] {$1$};

\begin{scope}[yshift=-1cm],
    \node[anchor=east] at (-0.5, 0) {\textbf{Metric $d_2$}};
    \filldraw[black] (1.25,0) circle (2.5pt) node[anchor=north] {$A_1$};
    \filldraw[black] (3.75,0) circle (2.5pt) node[anchor=north] {$A_2\cup A_3$};
    \draw[-] (1.25,0) -- (3.75,0) node[midway, anchor=south] {$1$};
\end{scope}
\end{tikzpicture}
\caption{Metrics considered in the proof of \Cref{thm:utilitarian-add-lb}. In this and all similar figures throughout the paper, the (sets of) agents are represented by circles, with the identity of the agents or sets below them, and the distances between them are written on top of the corresponding line segments. All figures consider indistinguishable metrics for a certain preference profile of the agents and thus any voting rule must select the same subsets for any of these metrics.}
    \label{fig:uti-add-lb}
\end{figure}

We now consider the case with $|S\cap A_3| \leq k'-1$ and define the distance metric $d_2$ on $A$ by the following positions $x\in (-\infty,\infty)^n$: $x_a=0$ for every $a\in A_1\cup A_2$ and $x_a=1$ for every $a\in A_3$.
Note that $d_2\ \rhd \succ$; see \Cref{fig:uti-add-lb} for an illustration.
Since $|A_1|=n_1\leq k=|A_2\cup A_3|$ by assumption, defining $T\subset A$ such that $|A_1\cap T|=k-(k'-1)-n_2$ and $|(A_2\cup A_3)\cap T|=n_2+(k'-1)$, we know that
\[
    \SC(S,A;d_2)\geq \SC(T,A;d_2) = n_1(n_2+(k'-1))+k(k-(k'-1)-n_2).
\]
If we consider the alternative committee $S'$ such that $|A_1\cap S'|=0$ and $|(A_2\cup A_3)\cap S'|=k$, we have $\SC(S',a;d_2)=k$ for every $a\in A_1$ and $\SC(S',a;d_2)=0$ for every $a\in A_2\cup A_3$.
We obtain
\begin{align}
    \dist(f(\succ),\instance) \geq \frac{\SC(S,A;d_2)}{\SC(S',A;d_2)} & \geq \frac{n_1(n_2+(k'-1))+k(k-(k'-1)-n_2)}{n_1k}\nonumber \\
    & =1 + \frac{(k-n_1)(k-n_2-k'+1)}{kn_1}.\label{eq:lb-uti-add-ins2}
\end{align}

We conclude that $\text{dist}(f(\succ),\instance) \geq 1+\gamma$, where $\gamma$ is the optimal value of the following optimization program:
\begin{equation}
\begin{aligned}\label{eq:lp-lb-uti-add}
    \max\quad & \min\bigg\{ \frac{2k'(n_1+n_2-k)}{n_1n_2-2kn_2+2k^2}, \frac{(k-n_1)(k-n_2-k'+1)}{kn_1}\bigg\}\\
    \text{s.t.}\quad & n_1,n_2 \leq k,\\
    & k'+n_2 \leq k,\\
    & n_1+n_2 \geq k,\\
    & k',n_1,n_2 \in \NN.
\end{aligned}
\end{equation}
Intuitively, for large enough $k$ it is not hard to see that the above program has a strictly positive objective value, leading to a distortion strictly greater than one, while for small values of $k$ in can be solved computationally by enumerating all possible solutions.
Naturally, the larger the lower bound we choose for $k$, the better the bound is. 

Formally, to obtain the claim bound of $1.0914$, it suffices to assume $k\geq 200$.
Indeed, for any fixed $k\geq 200$, we take
\[
    n_1=\bigg\lceil\frac{3}{5}k\bigg\rceil,\quad n_2=\bigg\lceil \frac{13}{20}k\bigg\rceil,\quad k'=\bigg\lceil \frac{1}{5}k\bigg\rceil.
\]
These values trivially satisfy the constraints of the program \eqref{eq:lp-lb-uti-add}.
Moreover, since $k\geq 200$ we have
\begin{align}
    n_1 \leq \frac{3}{5}k+1 \leq 
    \frac{121}{200}k,\quad n_2 \leq \frac{13}{20}k+1 \leq 
    \frac{131}{200}k, \quad k' \leq \frac{1}{5}k+1 \leq 
    \frac{41}{200}k.\label{eq:bounds-n1-n2}
\end{align}
Therefore,
\[
    \frac{2k'(n_1+n_2-k)}{n_1n_2-2kn_2+2k^2} \geq \frac{2\cdot \frac{1}{5}k\big(\frac{3}{5}k+\frac{13}{20}k-k)}{\frac{121}{200}k\cdot \frac{13}{20}k-2k\cdot \frac{13}{20}k+2k^2} = \frac{\frac{2}{5}\cdot \frac{1}{4}}{\frac{1573-5200+8000}{4000}} = \frac{400}{4373} \approx 0.09147,
\]
where the first inequality follows from \eqref{eq:bounds-n1-n2} and the fact that $n_1n_2-2kn_2$ is decreasing in $n_2$ since $n_1<2k$.
On the other hand, 
\[
    \frac{(k-n_1)(k-n_2-k'+1)}{kn_1} \geq \frac{(k-\frac{121}{200}k)(k-\frac{131}{200}k-\frac{41}{200}k)}{\frac{121}{200}k^2} = \frac{\frac{79}{200}\cdot \frac{28}{200}}{\frac{121}{200}} = \frac{2212}{24200} \approx 0.09140,
\]
where the first inequality again follows from \eqref{eq:bounds-n1-n2}.
Combining these two inequalities, we obtain $\text{dist}(f(\succ),\instance) > 1.0914$, concluding the proof for $k\geq 200$.

For values $k\leq 200$, we include in \Cref{tab:bounds-uti-add-smallk} feasible (and even optimal) solutions for the program \eqref{eq:lp-lb-uti-add} alongside their objective value.
Since all objective values are strictly greater than $0.0914$, this completes the proof.
\end{proof}

\begin{table}[!htbp]
\centering
\makebox[0cm]{
\footnotesize{
\begin{tabular}{@{\quad} c @{~~} c @{~~} c @{\quad} | @{\quad}  c @{~~} c @{~~} c @{\quad} | @{\quad}  c @{~~} c @{~~} c @{\quad} | @{\quad}  c @{~~} c @{~~} c @{\quad}}
\toprule \\[-10pt]
$k$ & $(n_1,n_2,k')$ & $\gamma$ & $k$ & $(n_1,n_2,k')$ & $\gamma$ & $k$ & $(n_1,n_2,k')$ & $\gamma$ & $k$ & $(n_1,n_2,k')$ & $\gamma$
\\[1pt] 
\midrule 
\\[-9pt]
$3$ & $(2, 2, 1)$ & $0.1667$ & $53$ & $(30, 37, 10)$ & $0.0998$ & $103$ & $(63, 67, 21)$ & $0.0974$ & $153$ & $(93, 99, 32)$ & $0.0970$ \\
$4$ & $(2, 3, 1)$ & $0.1429$ & $54$ & $(32, 36, 11)$ & $0.0995$ & $104$ & $(62, 68, 22)$ & $0.0977$ & $154$ & $(90, 104, 30)$ & $0.0969$ \\
$5$ & $(3, 3, 2)$ & $0.1333$ & $55$ & $(34, 36, 11)$ & $0.0996$ & $105$ & $(64, 68, 22)$ & $0.0976$ & $155$ & $(92, 104, 30)$ & $0.0969$ \\
$6$ & $(3, 5, 1)$ & $0.1481$ & $56$ & $(33, 37, 12)$ & $0.0996$ & $106$ & $(61, 72, 21)$ & $0.0974$ & $156$ & $(94, 103, 31)$ & $0.0970$ \\
$7$ & $(3, 6, 1)$ & $0.1250$ & $57$ & $(35, 36, 13)$ & $0.0992$ & $107$ & $(63, 72, 21)$ & $0.0978$ & $157$ & $(96, 103, 31)$ & $0.0970$ \\
$8$ & $(4, 6, 2)$ & $0.1250$ & $58$ & $(35, 38, 12)$ & $0.0986$ & $108$ & $(65, 72, 21)$ & $0.0978$ & $158$ & $(98, 101, 33)$ & $0.0969$ \\
$9$ & $(4, 7, 2)$ & $0.1250$ & $59$ & $(34, 40, 12)$ & $0.0997$ & $109$ & $(67, 71, 22)$ & $0.0978$ & $159$ & $(95, 106, 31)$ & $0.0967$ \\
$10$ & $(6, 7, 2)$ & $0.1176$ & $60$ & $(36, 40, 12)$ & $0.1000$ & $110$ & $(64, 75, 21)$ & $0.0974$ & $160$ & $(97, 105, 32)$ & $0.0967$ \\
$11$ & $(5, 8, 3)$ & $0.1091$ & $61$ & $(38, 39, 13)$ & $0.0992$ & $111$ & $(66, 74, 22)$ & $0.0974$ & $161$ & $(99, 105, 32)$ & $0.0968$ \\
$12$ & $(7, 9, 2)$ & $0.1185$ & $62$ & $(35, 43, 12)$ & $0.0995$ & $112$ & $(68, 74, 22)$ & $0.0975$ & $162$ & $(101, 104, 33)$ & $0.0969$ \\
$13$ & $(5, 11, 2)$ & $0.1121$ & $63$ & $(37, 43, 12)$ & $0.0992$ & $113$ & $(70, 73, 23)$ & $0.0975$ & $163$ & $(100, 104, 35)$ & $0.0966$ \\
$14$ & $(9, 9, 3)$ & $0.1086$ & $64$ & $(39, 42, 13)$ & $0.0992$ & $114$ & $(69, 73, 25)$ & $0.0973$ & $164$ & $(97, 109, 33)$ & $0.0969$ \\
$15$ & $(8, 11, 3)$ & $0.1154$ & $65$ & $(41, 41, 14)$ & $0.0991$ & $115$ & $(69, 76, 23)$ & $0.0971$ & $165$ & $(99, 110, 32)$ & $0.0970$ \\
$16$ & $(10, 11, 3)$ & $0.1111$ & $66$ & $(38, 45, 13)$ & $0.0986$ & $116$ & $(68, 77, 24)$ & $0.0974$ & $166$ & $(101, 108, 34)$ & $0.0969$ \\
$17$ & $(9, 12, 4)$ & $0.1046$ & $67$ & $(40, 45, 13)$ & $0.0986$ & $117$ & $(70, 77, 24)$ & $0.0976$ & $167$ & $(103, 107, 35)$ & $0.0967$ \\
$18$ & $(9, 14, 3)$ & $0.1111$ & $68$ & $(42, 44, 14)$ & $0.0986$ & $118$ & $(72, 76, 25)$ & $0.0975$ & $168$ & $(100, 112, 33)$ & $0.0967$ \\
$19$ & $(11, 13, 4)$ & $0.1078$ & $69$ & $(41, 46, 14)$ & $0.0990$ & $119$ & $(69, 81, 23)$ & $0.0974$ & $169$ & $(102, 112, 33)$ & $0.0968$ \\
$20$ & $(13, 13, 4)$ & $0.1069$ & $70$ & $(43, 45, 15)$ & $0.0987$ & $120$ & $(71, 81, 23)$ & $0.0974$ & $170$ & $(104, 111, 34)$ & $0.0968$ \\
$21$ & $(12, 14, 5)$ & $0.1071$ & $71$ & $(40, 48, 15)$ & $0.0982$ & $121$ & $(73, 80, 24)$ & $0.0974$ & $171$ & $(106, 110, 35)$ & $0.0968$ \\
$22$ & $(12, 16, 4)$ & $0.1053$ & $72$ & $(42, 49, 14)$ & $0.0991$ & $122$ & $(75, 80, 24)$ & $0.0975$ & $172$ & $(103, 114, 34)$ & $0.0965$ \\
$23$ & $(14, 15, 5)$ & $0.1038$ & $73$ & $(44, 49, 14)$ & $0.0989$ & $123$ & $(77, 78, 26)$ & $0.0971$ & $173$ & $(105, 114, 34)$ & $0.0966$ \\
$24$ & $(13, 17, 5)$ & $0.1058$ & $74$ & $(46, 47, 16)$ & $0.0987$ & $124$ & $(74, 83, 24)$ & $0.0971$ & $174$ & $(104, 114, 36)$ & $0.0967$ \\
$25$ & $(15, 17, 5)$ & $0.1067$ & $75$ & $(43, 52, 14)$ & $0.0985$ & $125$ & $(76, 82, 25)$ & $0.0972$ & $175$ & $(109, 113, 35)$ & $0.0967$ \\
$26$ & $(15, 18, 5)$ & $0.1020$ & $76$ & $(45, 51, 15)$ & $0.0984$ & $126$ & $(78, 82, 25)$ & $0.0972$ & $176$ & $(103, 118, 35)$ & $0.0966$ \\
$27$ & $(14, 20, 5)$ & $0.1032$ & $77$ & $(47, 51, 15)$ & $0.0984$ & $127$ & $(77, 82, 27)$ & $0.0971$ & $177$ & $(105, 119, 34)$ & $0.0968$ \\
$28$ & $(16, 20, 5)$ & $0.1042$ & $78$ & $(49, 50, 16)$ & $0.0986$ & $128$ & $(74, 86, 26)$ & $0.0969$ & $178$ & $(107, 118, 35)$ & $0.0968$ \\
$29$ & $(18, 19, 6)$ & $0.1041$ & $79$ & $(48, 51, 17)$ & $0.0981$ & $129$ & $(76, 86, 26)$ & $0.0973$ & $179$ & $(109, 118, 35)$ & $0.0968$ \\
$30$ & $(17, 20, 7)$ & $0.1020$ & $80$ & $(48, 53, 16)$ & $0.0979$ & $130$ & $(78, 86, 26)$ & $0.0974$ & $180$ & $(111, 116, 37)$ & $0.0967$ \\
$31$ & $(17, 22, 6)$ & $0.1030$ & $81$ & $(47, 54, 17)$ & $0.0982$ & $131$ & $(80, 85, 27)$ & $0.0973$ & $181$ & $(108, 121, 35)$ & $0.0966$ \\
$32$ & $(19, 22, 6)$ & $0.1021$ & $82$ & $(49, 54, 17)$ & $0.0986$ & $132$ & $(77, 90, 25)$ & $0.0971$ & $182$ & $(110, 120, 36)$ & $0.0966$ \\
$33$ & $(21, 21, 7)$ & $0.1022$ & $83$ & $(51, 53, 18)$ & $0.0983$ & $133$ & $(79, 89, 26)$ & $0.0971$ & $183$ & $(112, 120, 36)$ & $0.0967$ \\
$34$ & $(20, 23, 7)$ & $0.1029$ & $84$ & $(48, 58, 16)$ & $0.0982$ & $134$ & $(81, 89, 26)$ & $0.0972$ & $184$ & $(114, 119, 37)$ & $0.0967$ \\
$35$ & $(22, 22, 8)$ & $0.1013$ & $85$ & $(50, 57, 17)$ & $0.0983$ & $135$ & $(83, 88, 27)$ & $0.0972$ & $185$ & $(116, 117, 39)$ & $0.0965$ \\
$36$ & $(22, 24, 7)$ & $0.1006$ & $86$ & $(52, 57, 17)$ & $0.0983$ & $136$ & $(85, 86, 29)$ & $0.0971$ & $186$ & $(110, 123, 38)$ & $0.0966$ \\
$37$ & $(21, 26, 7)$ & $0.1029$ & $87$ & $(54, 56, 18)$ & $0.0983$ & $137$ & $(82, 91, 27)$ & $0.0969$ & $187$ & $(112, 123, 38)$ & $0.0967$ \\
$38$ & $(23, 25, 8)$ & $0.1024$ & $88$ & $(51, 60, 17)$ & $0.0979$ & $138$ & $(84, 91, 27)$ & $0.0969$ & $188$ & $(114, 122, 39)$ & $0.0967$ \\
$39$ & $(25, 24, 9)$ & $0.1005$ & $89$ & $(53, 60, 17)$ & $0.0978$ & $139$ & $(83, 91, 29)$ & $0.0971$ & $189$ & $(111, 127, 37)$ & $0.0966$ \\
$40$ & $(22, 28, 8)$ & $0.1015$ & $90$ & $(55, 59, 18)$ & $0.0979$ & $140$ & $(88, 90, 28)$ & $0.0971$ & $190$ & $(113, 127, 37)$ & $0.0966$ \\
$41$ & $(24, 28, 8)$ & $0.1013$ & $91$ & $(57, 59, 18)$ & $0.0980$ & $141$ & $(82, 95, 28)$ & $0.0970$ & $191$ & $(115, 126, 38)$ & $0.0966$ \\
$42$ & $(26, 27, 9)$ & $0.1009$ & $92$ & $(56, 59, 20)$ & $0.0978$ & $142$ & $(84, 95, 28)$ & $0.0971$ & $192$ & $(117, 126, 38)$ & $0.0967$ \\
$43$ & $(23, 31, 8)$ & $0.1009$ & $93$ & $(56, 62, 18)$ & $0.0974$ & $143$ & $(86, 95, 28)$ & $0.0972$ & $193$ & $(119, 125, 39)$ & $0.0967$ \\
$44$ & $(27, 28, 10)$ & $0.1002$ & $94$ & $(55, 63, 19)$ & $0.0981$ & $144$ & $(88, 94, 29)$ & $0.0972$ & $194$ & $(116, 129, 38)$ & $0.0965$ \\
$45$ & $(27, 30, 9)$ & $0.1000$ & $95$ & $(57, 63, 19)$ & $0.0982$ & $145$ & $(90, 92, 31)$ & $0.0969$ & $195$ & $(118, 129, 38)$ & $0.0965$ \\
$46$ & $(26, 31, 10)$ & $0.1003$ & $96$ & $(59, 62, 20)$ & $0.0980$ & $146$ & $(87, 97, 29)$ & $0.0969$ & $196$ & $(120, 128, 39)$ & $0.0965$ \\
$47$ & $(28, 31, 10)$ & $0.1011$ & $97$ & $(56, 66, 19)$ & $0.0978$ & $147$ & $(89, 97, 29)$ & $0.0969$ & $197$ & $(122, 128, 39)$ & $0.0966$ \\
$48$ & $(30, 30, 11)$ & $0.1000$ & $98$ & $(58, 66, 19)$ & $0.0978$ & $148$ & $(91, 96, 30)$ & $0.0970$ & $198$ & $(124, 125, 42)$ & $0.0964$ \\
$49$ & $(27, 34, 10)$ & $0.0998$ & $99$ & $(60, 65, 20)$ & $0.0978$ & $149$ & $(93, 95, 31)$ & $0.0970$ & $199$ & $(118, 132, 40)$ & $0.0966$ \\
$50$ & $(29, 34, 10)$ & $0.1005$ & $100$ & $(62, 65, 20)$ & $0.0979$ & $150$ & $(92, 95, 33)$ & $0.0967$ & $200$ & $(120, 133, 39)$ & $0.0967$ \\
$51$ & $(31, 34, 10)$ & $0.1004$ & $101$ & $(59, 68, 20)$ & $0.0974$ & $151$ & $(89, 100, 31)$ & $0.0969$ & & & \\
$52$ & $(33, 33, 11)$ & $0.0997$ & $102$ & $(61, 68, 20)$ & $0.0974$ & $152$ & $(91, 100, 31)$ & $0.0970$ & & & \\
\bottomrule
\end{tabular}
}}
\caption{
\label{tab:bounds-uti-add-smallk}
Optimal solutions for the program \eqref{eq:lp-lb-uti-add} for all integer values of $k$ between $3$ and $200$, alongside their objective value.
}
\end{table}

\subsection{Utilitarian $q$-Cost}\label{subsec:uti-q}

In this section, we study the distortion of voting rules in the context of utilitarian $q$-cost, in which the cost of a committee $S'$ for an agent is given by its distance to the $q$th closest agent in $S'$, and the social cost of a committee is the sum of its cost for all agents.
Formally, for a set of agents $A$, a committee size $k$, a committee $S'\in {A\choose k}$, and a distance metric $d$, the social cost of the committee is given by
\[
    \SC(S',A;d) = \sum_{a\in A}\tilde{d}(a)_q,
\]
where $\tilde{d}(a)\in \RR^S_+$ contains the values $\{d(a,s)\mid s\in S'\}$ in increasing order.

Similarly to the setting with disjoint voters and candidates, the distortion of voting rules heavily depends on the value of $q$. 
The key intuition is that when $q$ is large, the $q$-cost objective emphasizes the lower-ranked positions in each agent’s ranking, which leads to a higher optimal cost. 
This makes the problem easier, as the selected committee can better approximate the worst-case distances. 
In contrast, when $q$ is small, the objective focuses on the top few distances, so there may exist very good committees that are difficult to identify with only ordinal information, leading to greater potential distortion. 
Indeed, a result by \citet{caragiannis2022metric} implies the existence of $(n,k)$-voting rules with distortion $3$ for $q$-cost whenever \smash{$q>\frac{k}{2}$}, since their result holds in the setting with disjoint voters and candidates and general distance metrics. 
We complement this result with a lower bound that ranges from \smash{$\frac{3}{2}$} to $2$ as $q$ varies between \smash{$\big\lceil\frac{k}{2}\big\rceil+1$} and $k$. 
For \smash{$q\leq \frac{k}{2}$}, \citet{caragiannis2022metric} showed that no rule provides a bounded distortion; we prove that this still holds in our setting.

In \Cref{subsec:utilitarian-kcost}, we study the case where $q=k=2$ in further detail and achieve the best-possible distortions of $\frac{4}{3}$ and $2$ for odd and even $n$, respectively, through natural voting rules that are able to leverage the different objectives involved in the problem.

\subsubsection{Impossibility Results}\label{subsec:utilitarian-qcost-lbs}

In this section, we provide two strong impossibilities regarding distortion bounds for $q$-cost, analyzing the cases with $q\leq \frac{k}{2}$ and with $q \geq \big\lceil\frac{k}{2}\big\rceil+1$ separately.

We begin with a strong impossibility for the case where we focus, for each agent, on their $q$th closest selected agent with $q\leq \frac{k}{2}$.
We show that no constant distortion is possible in this setting, regardless of the number of agents to select.
\begin{restatable}{theorem}{thmUtilitarianSmallq}\label{thm:utilitarian-1cost}
    For every $k\in \NN$ with $k\geq 2$ and $q\in \NN$ with $q\leq \frac{k}{2}$, there exists $n\in \NN$ with $n\geq k$ such that, for every $(n,k)$-voting rule $f$, $\dist(f)$ is unbounded for utilitarian $q$-cost.
\end{restatable}
\begin{proof}
We let $k$ and $q$ be as in the statement, fix $n \in \NN$ to a large value, in particular with $n\geq 2k+q$ (we will ultimately take the limit $n\to \infty$), and consider an arbitrary $(n,k)$-voting rule $f$.
We denote $p=\big\lfloor \frac{k}{q} \big\rfloor \geq 2$ and partition the agents into $p+1$ sets $A=\dot{\bigcup}_{i=1}^{p}A_i \cup B$, such that $|A_i|\in \big\{\big\lfloor \frac{n-q}{p} \big\rfloor, \big\lceil \frac{n-q}{p} \big\rceil\big\}$ for every $i\in [p]$ and $|B|=q$. 
Note that this is possible since
\[
    p \bigg\lfloor \frac{n-q}{p} \bigg\rfloor + q \leq n \leq p \bigg\lfloor \frac{n-q}{p} \bigg\rfloor + q.
\]
We consider the profile $\succ\in \calL^n(n)$, where 
\begin{enumerate}[label=(\roman*)]
    \item $b\succ_a c$ whenever $a\in A_i,b\in A_j,c\in A_\ell$ for some $i,j,\ell \in [p]$ with $|i-j|<|i-\ell|$;
    \item $b\succ_a c$ whenever $a\in A_i,b\in A_j, c\in B$ for some $i,j\in [p]$;
    \item $b\succ_a c$ whenever $a,b\in B, c\in A_i$ for some $i\in [p]$;
    \item $b\succ_a c$ whenever $a\in B,b\in A_i,c\in A_j$ for some $i,j\in [p]$ with $i>j$;
\end{enumerate}
and the remaining pairwise comparisons are arbitrary.
We consider the election $\instance=(A,k,\succ)$ with $A=[n]$.

In what follows, we distinguish whether $f$ selects all $q$ agents in $B$ or not and construct appropriate distance metrics to show that, in either case, the distortion can be arbitrarily large.
Intuitively, if $f$ selects $B$ we will consider this set to be relatively close to $A_p$, so that picking $q$ agents from each set $A_1,\ldots,A_p$ would give a much lower social cost.
On the contrary, if $f$ does not select $B$, we will place this set extremely far from all others, so that the social cost of the selected set is huge compared to the social cost of a committee containing $B$.

Formally, we first consider the case with $B\subseteq S$ and define the distance metric $d_1$ on $A$ given by the following positions $x\in (-\infty,\infty)^n$: $x_a=i-1$ for every $a\in A_i$ and every $i\in [p]$, and $x_a=2(p-1)$ for every $a\in B$.
Note that $d_1\ \rhd \succ$; see \Cref{fig:utilitarian-qcost} for an illustration.
Since $B\subseteq S$, we have that $\big|S \cap \bigcup_{i\in [p]}A_i\big| \leq k-q$.
Hence, from an averaging argument, there exists $j\in [p]$ with 
\[
    |S\cap A_j| \leq \frac{k-q}{p} = \frac{q}{k}(k-q) < q.
\]
From the definition of $q$-cost, we thus have 
\begin{equation}
    \SC(S,a;d_1) \geq \min\{d_1(a,b)\mid b\in A\setminus A_j\}\geq 1 \quad \text{for every } a\in A_j.\label{ineq:uti-qcost-lb-alg-case1}
\end{equation}
On the other hand, consider the set $S=\cup_{i\in [p]}S_i$, where $S_i\subseteq A_i$ and $|S_i|\geq q$ for every $i\in [p]$.
Note that this set exists because $pq=k$ and
\[
    |A_i| \geq \bigg\lfloor \frac{n-q}{p} \bigg\rfloor \geq \bigg\lfloor \frac{2k}{k}q \bigg\rfloor  \geq q,
\]
where we used our assumption $n\geq 2k+q$.
From the definition of $q$-cost, we have that $\SC(S,a;d_1)=0$ for every $a\in A_i$ and every $i\in [p]$.
For each $a\in B$, we have $\SC(S,a;d_1)=p-1$.
Combining these facts with \cref{ineq:uti-qcost-lb-alg-case1}, we obtain
\[
    \dist(f(\succ),\instance) \geq \frac{\SC(S,A;d_1)}{\SC(S,A;d_1)} \geq \frac{|A_j|}{(p-1)|B|} \geq \bigg\lfloor \frac{n-q}{p} \bigg\rfloor \frac{1}{(p-1)q} = \bigg\lfloor \frac{(n-q)q}{k} \bigg\rfloor \cdot \frac{1}{k-q}.
\]

We now consider the case with $B\not\subseteq S$ and define the distance metric $d_2$ on $A$ given by the following positions $x\in (-\infty,\infty)^n$: $x_a=i-1$ for every $a\in A_i$ and every $i\in [p]$, and $x_a=p-1+Mn$ for every $a\in B$.
Note that $d_2\ \rhd \succ$; see \Cref{fig:utilitarian-qcost} for an illustration.
Since $B\not\subseteq S$, we have that $|S\cap B| < q$ and thus, by the definition of $q$-cost, we have
\begin{equation}
    \SC(S,a;d_2) \geq \min\{d_2(a,b)\mid b\in A\setminus B\}\geq Mn \quad \text{for every } a\in B.\label{ineq:uti-qcost-lb-alg-case2}
\end{equation}
On the other hand, consider the set $T=B\cup \cup_{i\in [p-1]}T_i$, where $T_i\subseteq A_i$ and $|T_i|\geq q$ for every $i\in [p-1]$.
Note that this set exists because $(p-1)q=k-q$ and
\[
    |A_i| \geq \bigg\lfloor \frac{n-q}{p} \bigg\rfloor \geq \bigg\lfloor \frac{2k}{k}q \bigg\rfloor  \geq q,
\]
where we used our assumption $n\geq 2k+q$.
From the definition of $q$-cost, we have that $\SC(T,a;d_2)=0$ for every $a\in A_i$ and every $i\in [p-1]$ and $\SC(T,a;d_2)=0$ for every $a\in B$.
For each $a\in A_p$, we have $\SC(T,a;d_2)=1$.
Combining these facts with \cref{ineq:uti-qcost-lb-alg-case2}, we obtain
\[
    \dist(f(\succ),\instance) \geq \frac{\SC(S,A;d_2)}{\SC(T,A;d_2)} \geq \frac{Mn|B|}{|A_p|} \geq \frac{1}{\big\lceil \frac{n-q}{p} \big\rceil}Mnq = \frac{1}{\big\lceil \frac{(n-q)q}{k} \big\rceil}Mnq.
\]

Since $\dist(f(\succ),\instance) \geq \big\lfloor \frac{(n-q)q}{k} \big\rfloor \cdot \frac{1}{k-q}$ if $B\subseteq S$ and $\dist(f(\succ),\instance) \geq \frac{1}{\big\lceil \frac{(n-q)q}{k} \big\rceil}Mnq$ otherwise, we conclude that
\[
    \dist(f) \geq \min\bigg\{ \bigg\lfloor \frac{(n-q)q}{k} \bigg\rfloor \cdot \frac{1}{k-q}, \frac{1}{\big\lceil \frac{(n-q)q}{k} \big\rceil}Mnq \bigg\},
\]
which can be unbounded by taking $n$ and $M$ arbitrarily large.
\end{proof}

\begin{figure}
\centering
\begin{tikzpicture}[scale=1, every node/.style={font=\footnotesize}]
    \node[anchor=east] at (-0.5, 0) {\textbf{Metric $d_1$}};
    \filldraw[black] (0,0) circle (2.5pt) node[anchor=north] {$A_1$};
    \filldraw[black] (1,0) circle (2.5pt) node[anchor=north] {$A_2$};
    \filldraw[black] (2,0) circle (2.5pt) node[anchor=north] {$A_3$};
    \filldraw[black] (2.5,0) circle (0.5pt) node[anchor=north] {};
    \filldraw[black] (3,0) circle (0.5pt) node[anchor=north] {};
    \filldraw[black] (3.5,0) circle (0.5pt) node[anchor=north] {};
    \filldraw[black] (4,0) circle (2.5pt) node[anchor=north] {$A_{p-1}$};
    \filldraw[black] (5,0) circle (2.5pt) node[anchor=north] {$A_{p}$};
    \filldraw[black] (10,0) circle (2.5pt) node[anchor=north] {$B$};
    \draw[-] (0,0) -- (1,0) node[midway, anchor=south] {$1$};
    \draw[-] (1,0) -- (2,0) node[midway, anchor=south] {$1$};
    \draw[-] (4,0) -- (5,0) node[midway, anchor=south] {$1$};
    \draw[-] (5,0) -- (10,0) node[midway, anchor=south] {$p-1$};

\begin{scope}[yshift=-1cm],
    \node[anchor=east] at (-0.5, 0) {\textbf{Metric $d_2$}};
    \filldraw[black] (0,0) circle (2.5pt) node[anchor=north] {$A_1$};
    \filldraw[black] (1,0) circle (2.5pt) node[anchor=north] {$A_2$};
    \filldraw[black] (2,0) circle (2.5pt) node[anchor=north] {$A_3$};
    \filldraw[black] (2.5,0) circle (0.5pt) node[anchor=north] {};
    \filldraw[black] (3,0) circle (0.5pt) node[anchor=north] {};
    \filldraw[black] (3.5,0) circle (0.5pt) node[anchor=north] {};
    \filldraw[black] (4,0) circle (2.5pt) node[anchor=north] {$A_{p-1}$};
    \filldraw[black] (5,0) circle (2.5pt) node[anchor=north] {$A_{p}$};
    \filldraw[black] (11,0) circle (2.5pt) node[anchor=north] {$B$};
    \draw[-] (0,0) -- (1,0) node[midway, anchor=south] {$1$};
    \draw[-] (1,0) -- (2,0) node[midway, anchor=south] {$1$};
    \draw[-] (4,0) -- (5,0) node[midway, anchor=south] {$1$};
    \draw[-] (5,0) -- (7,0) node[midway, anchor=south] {};
    \draw[dashed] (7,0) -- (9,0) node[midway, anchor=south] {$Mn$};
    \draw[-] (9,0) -- (11,0) node[midway, anchor=south] {};
\end{scope}
\end{tikzpicture}
    \caption{Metrics considered in the proof of \Cref{thm:utilitarian-1cost}. In this and all similar figures throughout the paper, the (sets of) agents are represented by circles, with the identity of the agents or sets below them, and the distances between them are written on top of the corresponding line segments. All figures consider indistinguishable metrics for a certain preference profile of the agents and thus any voting rule must select the same subsets for any of these metrics.}
    \label{fig:utilitarian-qcost}
\end{figure}

Next, we prove a lower bound of $2 - \frac{k-q}{4q-k-3}$ for the distortion of any voting rule for utilitarian $q$-cost when $\big\lceil\frac{k}{2} \big\rceil < q \leq k$ and $k \geq 3$. 
\begin{restatable}{theorem}{thmUtilitarianLargeq}\label{thm:utilitarian-lowerbound-qcost}
    For every $k\in \NN$ with $k\geq 3$ and $q\in \NN$ with $\frac{k}{2} +1\leq q \leq k$, there exists $n\in \NN$ with $n\geq k$ such that, for every $(n,k)$-voting rule $f$, $\dist(f)$ is at least $2 - \frac{k-q}{4q-k-3}$ for utilitarian $q$-cost.
\end{restatable}
\begin{proof}
We let $k$ and $q$ be as in the statement and fix $n=2(3q-k-2)$, and consider an arbitrary $(n,k)$-voting rule $f$.
We partition the agents into four sets $A=\dot\bigcup_{i=1}^4 A_i$ such that $|A_1|=|A_4|=q-1$ and $|A_2|=|A_3|=2q-k-1$.
Note that all these values lie between $1$ and $q-1$. 
Indeed, this is trivial for $|A_1|$ and $|A_4|$, whereas for $|A_2|$ and $|A_3|$ we have $2q-k-1\geq 2\big(\frac{k}{2}+1\big)-k-1=1$ and $2q-k-1\leq 2q-q-1=q-1$, where we have used that $q$ lies between $\frac{k}{2}+1$ and $k$.

We consider the profile $\succ\in \calL^n(n)$, where
\begin{enumerate}[label=(\roman*)]
    \item $b\succ_a c$ whenever $a\in A_i,b\in A_j,c\in A_\ell$ for some $i,j,\ell\in [4]$ with $|i-j|<|i-\ell|$;
    \item $b \succ_a c$ whenever $a\in A_2,b\in A_1,c\in A_3$;
    \item $b \succ_a c$ whenever $a\in A_3,b\in A_4,c\in A_2$;
\end{enumerate}
and the remaining pairwise comparisons are arbitrary.
We consider the election $\instance=(A,k,\succ)$ with $A=[n]$.

In what follows, we distinguish whether $f$ selects $q$ or more agents from $A_1\cup A_2$, from $A_3\cup A_4$, or from none of them, and construct appropriate distance metrics to show that, in either case, the distortion is at least the one claimed in the statement.
Intuitively, if $f$ selects less than $q$ agents from both $A_1\cup A_2$ and from $A_3\cup A_4$, we will consider $A_1\cup A_2$ on one extreme and $A_3\cup A_4$ on the other, so that picking $q$ agents from any of these sets would lead to a lower social cost.
If $f$ selects $q$ or more agents from $A_1\cup A_2$ we will consider a metric where $A_1$ lies in one extreme, $A_2$ in the middle, and both $A_3$ and $A_4$ in the other extreme, so that picking all agents from $A_4$ would lead to a lower social cost.
If $f$ selects $q$ or more agents from $A_3\cup A_4$, we will construct a symmetric instance.

Formally, we first consider the case with $|S\cap (A_1\cup A_2)| < q$ and $|S\cap (A_3\cup A_4)| < q$ and define the distance metric $d_1$ on $A$ by the following positions $x\in (-\infty,\infty)^n$: $x_a=0$ for every $a\in A_1\cup A_2$ and $x_a=2$ for every $a\in A_3\cup A_4$.
Note that $d_1\ \rhd \succ$; see \Cref{fig:plot-uti-qcost}.(b) for an illustration.
It is clear that $\SC(S,a;d_1)=2$ for every $a\in A$.
If we consider the alternative committee $S'=A_1\cup A_2 \in {A\choose k}$, we have $\SC(S',a;d_1)=0$ for every $a\in A_1\cup A_2$ and $\SC(S',a;d_1)=2$ for every $a\in A_3\cup A_4$.
We obtain
\[
    \dist(f(\succ),\instance) \geq \frac{\SC(S,A;d_1)}{\SC(S',A;d_1)} = \frac{2\cdot n}{2\frac{n}{2}} = 2.
\]

If $|S\cap (A_3\cup A_4)| \geq q$, we define the distance metric $d_2$ on $A$ by the following positions $x\in (-\infty,\infty)^n$: $x_a=0$ for every $a\in A_1\cup A_2$, $x_a=1$ for every $a\in A_3$, and $x_a=2$ for every $a\in A_4$.
Note that $d_2\ \rhd \succ$; see \Cref{fig:plot-uti-qcost}.(b) for an illustration.
Since $|S\cap (A_1\cup A_2 \cup A_3)| \leq (k-q)+|A_3| = q-1<q$, we have that $\SC(S,a;d_2)=2$ for every $a\in A_1\cup A_2$.
Furthermore, since both $|A_3|<q$ and $|A_4|<q$, we have that $\SC(S,a;d_2)=1$ for every $a\in A_3\cup A_4$.
If we consider an alternative committee $S'\subseteq A_1\cup A_2 \in {A\choose k}$, which exists due to $|A_1\cup A_2| = 3q-k-2\geq q$, we have $\SC(S',a;d_2)=0$ for every $a\in A_1\cup A_2$, $\SC(S',a;d_2)=1$ for every $a\in A_3$, and $\SC(S',a;d_2)=2$ for every $a\in A_4$.
Thus, we obtain
\begin{align*}
    \dist(f(\succ),\instance) & \geq \frac{\SC(S,A;d_2)}{\SC(S',A;d_2)} \\
    & = \frac{2|A_1\cup A_2|+|A_3\cup A_4|}{|A_3|+2|A_4|} \\
    & = \frac{3(3q-k-2)}{(2q-k-1)+2(q-1)} \\
    & = 2 - \frac{k-q}{4q-k-3}.
\end{align*}

Analogously, if $|S\cap (A_1\cup A_2)| \geq q$, we define the distance metric $d_3$ on $A$ by the following positions $x\in (-\infty,\infty)^n$: $x_a=0$ for every $a\in A_1$, $x_a=1$ for every $a\in A_2$, and $x_a=2$ for every $a\in A_3\cup A_4$.
Note that $d_3\ \rhd \succ$; see \Cref{fig:plot-uti-qcost}.(b) for an illustration.
Since $|S\cap (A_2\cup A_3 \cup A_4)| \leq (k-q)+|A_2| = q-1<q$, we have that $\SC(S,a;d_3)=2$ for every $a\in A_3\cup A_4$.
Furthermore, since both $|A_1|<q$ and $|A_2|<q$, we have that $\SC(S,a;d_3)=1$ for every $a\in A_1\cup A_2$.
If we consider an alternative committee $S'\subseteq A_3\cup A_4 \in {A\choose k}$, which exists due to $|A_3\cup A_4| = 3q-k-2\geq q$, we have $\SC(S',a;d_3)=0$ for every $a\in A_3\cup A_4$, $\SC(S',a;d_3)=1$ for every $a\in A_2$, and $\SC(S',a;d_3)=2$ for every $a\in A_1$.
Thus, we obtain
\begin{align*}
    \dist(f(\succ),\instance) & \geq \frac{\SC(S,A;d_3)}{\SC(S',A;d_3)} \\
    & = \frac{2|A_3\cup A_4|+|A_1\cup A_2|}{|A_2|+2|A_1|} \\
    & = \frac{3(3q-k-2)}{(2q-k-1)+2(q-1)} \\
    & = 2 - \frac{k-q}{4q-k-3}.
\end{align*}

Since $\dist(f(\succ),\instance)\geq 2 - \frac{k-q}{4q-k-3}$ regardless of $f(\succ)$, we conclude that $\dist(f) \geq 2 - \frac{k-q}{4q-k-3}$.
\end{proof}
The lower bound provided in this theorem increases in $q$ and varies between $\frac{3}{2}+\frac{3}{2(k+1)}$ for $q=\frac{k}{2}+1$ and $2$ for $q=k$; \Cref{fig:plot-uti-qcost}.(a) illustrates it for $k=100$ and $q$ between $51$ and $100$.

\begin{figure}
    \centering
\begin{subfigure}{0.47\textwidth}
    \centering
\begin{tikzpicture}[xscale=0.07,yscale=0.09]
\discontarrow(2,2)(2,4)(2,8)(2,40);
\discontarrow(2,2)(4,2)(8,2)(65,2);
\Text[x=67,y=2]{$q$}
\Text[x=2,y=43]{distortion}
\draw (10,3) -- (10,1) node[below] {$50$};
\draw (20,3) -- (20,1) node[below] {$60$};
\draw (30,3) -- (30,1) node[below] {$70$};
\draw (40,3) -- (40,1) node[below] {$80$};
\draw (50,3) -- (50,1) node[below] {$90$};
\draw (60,3) -- (60,1) node[below] {$100$};

\draw (3,10) -- (1,10) node[left] {$1.5$};
\draw (3,15) -- (1,15) node[left] {$1.6$};
\draw (3,20) -- (1,20) node[left] {$1.7$};
\draw (3,25) -- (1,25) node[left] {$1.8$};
\draw (3,30) -- (1,30) node[left] {$1.9$};
\draw (3,35) -- (1,35) node[left] {$2$};

\draw[semithick,blue] plot[mark=., mark color=none,mark options={fill=black,draw opacity=0},mark size=3pt] coordinates { (11, 10.7426) (12, 12.1429) (13, 13.4404) (14, 14.646) (15, 15.7692) (16, 16.8182) (17, 17.8) (18, 18.7209) (19, 19.5865) (20, 20.4015) (21, 21.1702) (22, 21.8966) (23, 22.5839) (24, 23.2353) (25, 23.8535) (26, 24.441) (27, 25.0) (28, 25.5325) (29, 26.0405) (30, 26.5254) (31, 26.989) (32, 27.4324) (33, 27.8571) (34, 28.2642) (35, 28.6548) (36, 29.0299) (37, 29.3902) (38, 29.7368) (39, 30.0704) (40, 30.3917) (41, 30.7014) (42, 31.0) (43, 31.2882) (44, 31.5665) (45, 31.8354) (46, 32.0954) (47, 32.3469) (48, 32.5904) (49, 32.8261) (50, 33.0545) (51, 33.2759) (52, 33.4906) (53, 33.6989) (54, 33.9011) (55, 34.0975) (56, 34.2883) (57, 34.4737) (58, 34.654) (59, 34.8294) (60, 35.0) };
\end{tikzpicture}
\caption{Lower bounds for $k=100$, $q\in \{51,\ldots,100\}$.}
    \end{subfigure}
    \hfill
    \begin{subfigure}{0.52\textwidth}

\begin{tikzpicture}[scale=1, every node/.style={font=\footnotesize}]
    \node[anchor=east] at (-0.5, 0) {\textbf{Metric $d_1$}};
    \filldraw[black] (0,0) circle (2.5pt) node[anchor=north] {$A_1\cup A_2$};
    \filldraw[black] (5,0) circle (2.5pt) node[anchor=north] {$A_3\cup A_4$};
    \draw[-] (0,0) -- (5,0) node[midway, anchor=south] {$2$};

\begin{scope}[yshift=-1cm],
    \node[anchor=east] at (-0.5, 0) {\textbf{Metric $d_2$}};
    \filldraw[black] (0,0) circle (2.5pt) node[anchor=north] {$A_1\cup A_2$};
    \filldraw[black] (2.5,0) circle (2.5pt) node[anchor=north] {$A_3$};
    \filldraw[black] (5,0) circle (2.5pt) node[anchor=north] {$A_4$};
    \draw[-] (0,0) -- (2.5,0) node[midway, anchor=south] {$1$};
    \draw[-] (2.5,0) -- (5,0) node[midway, anchor=south] {$1$};
\end{scope}

\begin{scope}[yshift=-2cm],
    \node[anchor=east] at (-0.5, 0) {\textbf{Metric $d_3$}};
    \filldraw[black] (0,0) circle (2.5pt) node[anchor=north] {$A_1$};
    \filldraw[black] (2.5,0) circle (2.5pt) node[anchor=north] {$A_2$};
    \filldraw[black] (5,0) circle (2.5pt) node[anchor=north] {$A_3\cup A_4$};
    \draw[-] (0,0) -- (2.5,0) node[midway, anchor=south] {$1$};
    \draw[-] (2.5,0) -- (5,0) node[midway, anchor=south] {$1$};
    \filldraw[white] (0,-1.2) circle (.1pt);
\end{scope}
\end{tikzpicture}
\caption{Metrics considered in the proof.}
\end{subfigure}
    \caption{Lower bound on the distortion of any rule for utilitarian $q$-cost stated in \Cref{thm:utilitarian-lowerbound-qcost}, and metrics used to prove it.}
    \label{fig:plot-uti-qcost}
\end{figure}

\subsubsection{A Voting Rule for $k=2$}\label{subsec:utilitarian-kcost}
In this section, we focus on the special case of utilitarian $q$-cost when $q = k = 2$. In this setting, the social cost of a committee $S$ for an agent $a$ is determined by the distance to the \textit{farthest} agent in the committee $S'$. Formally, for a set of agents $A$ and a committee $S' \in {A \choose 2}$, the social cost is:
\[
\SC(S', A; d) = \sum_{a \in A} \max_{s \in S'} d(a, s).
\]

On an intuitive level, the goal is to select agents that are both close to each other and close to the median agent(s). 
In particular, note that the optimal committee always consists of two consecutive agents: For any committee of non-consecutive agents, replacing the most extreme agent among the selected one with another closer to the median cannot decrease the social cost.

A visual aid for computing the social cost of a committee is what we call \textit{stair diagrams}, illustrated in \Cref{fig:stair-diagram-comparison}.
The area below both \textit{staircases} is a cost that every committee must incur. A specific committee $\{s_1,s_2\}$ must incur, in addition, a cost equal to the area of the rectangle whose basis is the line segment between both selected candidates and whose height is $n$ (and potentially an additional area to reach this point from the median).
The figure illustrates the common area incurred by any committee and the additional cost of two possible committees for each $n\in \{8,9\}$.
It also provides an intuition of the intrinsic difference between the cases with odd and even $k$.
The following lemma bounds the social cost of any committee from below and provides intuition about this objective. 

\begin{restatable}{lemma}{lemTwoSocialCost}\label{lem:2-social-cost}
Let $\instance = (A, k, \succ)$ be an election and $d\ \rhd \succ$ a consistent metric. Then, for every committee $S' = \{s_1, s_2\}\in {A\choose 2}$, 
\[
    \SC(S',A; d) \geq \begin{cases}
        \sum_{i=1}^{\frac{n-1}{2}} d(i, n-i+1) + \frac{n-1}{2}\cdot d(s_1, s_2) + \SC\big(S',\big\{\frac{n+1}{2}\big\}; d\big) & \text{if } n \text{ is odd,}\\[5pt]
        \sum_{i=1}^{\frac{n}{2}} d(i, n-i+1) + \frac{n}{2}\cdot d(s_1, s_2) & \text{if } n \text{ is even.}
    \end{cases}
\]
\end{restatable}
\begin{proof}
Let $\instance=(A,k,\succ)$ with $A=[n]$ and $d$ be as in the statement and $S'=\{s_1,s_2\}\in {A\choose k}$ an arbitrary committee.
We assume that $s_1<s_2$ w.l.o.g..
Let $i\in \{1,\ldots,\big\lfloor\frac{n}{2}\rfloor\}$ be a fixed agent.
If $i\leq s_1<s_2\leq n-i+1$, we have that the cost of the committee for agents $i$ and $n-i+1$ is at least
\[
    \SC(S',i;d) + \SC(S',n-i+1;d) = d(i,s_2)+d(s_1,n-i+1) \geq d(i,n-i+1)+d(s_1,s_2).
\]
Similarly, if $s_2<i$, we have
\[
    \SC(S',i;d) + \SC(S',n-i+1;d) = d(s_1,i)+d(s_1,n-i+1) \geq d(i,n-i+1)+d(s_1,s_2),
\]
and if $s_1>n-i+1$,
\[
    \SC(S',i;d) + \SC(S',n-i+1;d) = d(i,s_2)+d(n-i+1,s_2) \geq d(i,n-i+1)+d(s_1,s_2).
\]
Summing up over all agents, we obtain
\[
    \SC(S',A;d) = \sum_{i=1}^{\frac{n}{2}} (\SC(S',i;d) + \SC(S',n-i+1;d)) \geq \sum_{i=1}^{\frac{n}{2}} d(i, n-i+1) + \frac{n}{2} d(s_1, s_2)
\]
if $n$ is even, and
\begin{align*}
    \SC(S',A;d) & = \sum_{i=1}^{\frac{n-1}{2}} (\SC(S',i;d) + \SC(S',n-i+1;d)) +\SC\bigg(S',\frac{n+1}{2};d\bigg) \\
    & \geq \sum_{i=1}^{\frac{n-1}{2}} d(i, n-i+1) + \frac{n-1}{2} d(s_1, s_2)+\SC\bigg(S',\frac{n+1}{2};d\bigg)
\end{align*}
if $n$ is odd.
\end{proof}

\begin{figure}[t]
    \centering
    \begin{subfigure}{0.48\textwidth}
        \centering
        \begin{tikzpicture}[scale=0.7, every node/.style={font=\footnotesize}]
            \draw[thick] (0,0.6) -- (7.4,0.6);
            
            \fill[gray!30] (0, 0.6) -- (7.4, 0.6) -- (7.4, 1.2) -- (0, 1.2) -- cycle;
            \fill[gray!30] (0.7, 1.2) -- (6.7, 1.2) -- (6.7, 1.8) -- (0.7, 1.8) -- cycle;
            \fill[gray!30] (1.4, 1.8) -- (0.7*5 +2.5, 1.8) -- (0.7*5 +2.5, 2.4) -- (1.4, 2.4) -- cycle;

            \fill[gray!30] (2.1, 2.4) -- (0.7*5 +2.1, 2.4) -- (0.7*5 +2.1, 3) -- (2.1, 3) -- cycle;

            
            \fill[red!30] (1.4, 2.4) -- (2.35, 2.4) -- (2.35, 6) -- (1.4, 6) -- cycle;
            \fill[red!30] (2.1, 3) -- (3.5, 3) -- (3.5, 3.6) -- (2.1, 3.6) -- cycle;


            \fill[green!30] (3.5, 3) -- (5.6, 3) -- (5.6, 3.6) -- (3.5, 3.6) -- cycle;
            
           \fill[green!30] (0.7*5-0.7 + 2.8,  2.4) -- (6, 2.4) -- (6, 6) -- (0.7*5-0.7 +2.8 , 6) -- cycle;

            \foreach \x in {1,...,2} {
                \draw[thick] (0.7*\x-0.7,\x*0.6) -- (0.7*\x-0.7,\x*0.6+0.6);
                \draw[thick] (0.7*\x-0.7,\x*0.6+0.6) -- (0.7*\x ,\x*0.6+0.6);
            }
            \draw[thick] (0.7*3-0.7,3*0.6) -- (0.7*3-0.7,3*0.6+0.6);
            \draw[thick] (0.7*3-0.7,3*0.6+0.6) -- (2.35 ,3*0.6+0.6);
            
            \draw[thick] (2.35,4*0.6) -- (2.35,4*0.6+0.6);
            \draw[thick] (2.35,4*0.6+0.6) -- (0.7*4 + 0.7 ,4*0.6+0.6);

            \draw[thick] (0.7*5,5*0.6) -- (0.7*5,5*0.6+0.6);
            \draw[thick] (0.7*5,5*0.6+0.6) -- (0.7*5 + 2.1 ,5*0.6+0.6);

            \draw[thick] (0.7*5 +2.1,6*0.6) -- (0.7*5 +2.1,6*0.6+0.6);
            \draw[thick] (0.7*5 +2.1,6*0.6+0.6) -- (0.7*5 + 2.5 ,6*0.6+0.6);

            \draw[thick] (0.7*5 +2.5,7*0.6) -- (0.7*5 +2.5,7*0.6+0.6);
            \draw[thick] (0.7*5 +2.5,7*0.6+0.6) -- (0.7*5 + 3.2 ,7*0.6+0.6);

            \draw[thick] (0.7*5 +3.2,8*0.6) -- (0.7*5 +3.2,8*0.6+0.6);
            \draw[thick] (0.7*5 +3.2,8*0.6+0.6) -- (0.7*5 + 3.9 ,8*0.6+0.6);

            \draw[thick] (7.4, 5.4) -- (7.4, 6);

            \draw[thick] (0, 8*0.6+0.6) -- (0, 8*0.6+1.2);

            \foreach \x in {1,...,2} {
            \draw[thick] (\x* 0.7 - 0.7,8*0.6+0.6 - \x*0.6 + 0.6) -- (\x* 0.7,8*0.6+0.6- \x*0.6 + 0.6);
            \draw[thick] (\x* 0.7,8*0.6+0.6- \x*0.6 + 0.6) -- (\x* 0.7,8*0.6- \x*0.6 + 0.6);
            }

            \draw[thick] (3* 0.7 - 0.7,8*0.6+0.6 - 3*0.6 + 0.6) -- (2.35,8*0.6+0.6- 3*0.6 + 0.6);
            \draw[thick] (2.35,8*0.6+0.6- 3*0.6 + 0.6) -- (2.35,8*0.6- 3*0.6 + 0.6);
            
            \draw[thick] (2.35,8*0.6+0.6 - 4*0.6 + 0.6) -- (3.5,8*0.6+0.6- 4*0.6 + 0.6);
            \draw[thick] (3.5,8*0.6+0.6- 4*0.6 + 0.6) -- (3.5,8*0.6- 4*0.6 + 0.6);

            \draw[thick] (3.5,3) -- (5.6,3);
            \draw[thick] (5.6,3) -- (5.6,2.4);

            \draw[thick] (5.6,2.4) -- (6,2.4);
            \draw[thick] (6 ,2.4) -- (6,1.8);

            \draw[thick] (6,1.8) -- (6.7,1.8);
            \draw[thick] (6.7,1.8) -- (6.7,1.2);

            \draw[thick] (6.7,1.2) -- (7.4,1.2);
            \draw[thick] (7.4,1.2) -- (7.4,0.6);

            \foreach \x in {1,...,3} {
                \filldraw[black] (0.7*\x-0.7,0.6) circle (2pt) node[anchor=north] {\x};
            }
            \filldraw[black] (2.35,0.6) circle (2pt) node[anchor=north] {4};
            \filldraw[black] (3.5,0.6) circle (2pt) node[anchor=north] {5};
            \filldraw[black] (5.6,0.6) circle (2pt) node[anchor=north] {6};
            \filldraw[black] (6,0.6) circle (2pt) node[anchor=north] {7};
            \filldraw[black] (6.7,0.6) circle (2pt) node[anchor=north] {8};
            \filldraw[black] (7.4,0.6) circle (2pt) node[anchor=north] {9};

        \end{tikzpicture}
        \caption{Stair diagram for $n=9$. The red area corresponds to the committee $\{3,4\}$; the green area to $\{6,7\}$.}
    \end{subfigure}
    \hfill
    \begin{subfigure}{0.48\textwidth}
        \centering
        \begin{tikzpicture}[scale=0.7, every node/.style={font=\footnotesize}]
            \draw[thick] (0,0.6) -- (7.7,0.6);
            
            \foreach \x/\y in {1/8, 2/7, 3/6, 4/5} {
                \fill[gray!30] (0.7*\x-0.7, 0.6) -- (0.7*\y +2.1, 0.6) -- (0.7*\y + 2.1, \x*0.6+0.6) -- (0.7*\x-0.7, \x*0.6+0.6) -- cycle;
            }

            \foreach \x/\y in {4/5} {
                \fill[red!30] (0.7*\x-0.7, 0.6 + 2.4) -- (0.7*\y +2.1, 0.6 +2.4) -- (0.7*\y + 2.1, \x*0.6+0.6 + 2.4) -- (0.7*\x-0.7 , \x*0.6+0.6 + 2.4) -- cycle;
            }

            \foreach \x/\y in {5/6} {
                \fill[green!30] (0.7*\x-0.7 + 2.8,  2.4) -- (0.7*\y +2.1, 2.4) -- (0.7*\y + 2.1, \x*0.6 + 2.4) -- (0.7*\x-0.7 +2.8 , \x*0.6+ 2.4) -- cycle;
            }

            \foreach \x in {1,...,3} {
                \draw[thick] (0.7*\x-0.7,\x*0.6) -- (0.7*\x-0.7,\x*0.6+0.6); 
                \draw[thick] (0.7*\x-0.7,\x*0.6+0.6) -- (0.7*\x,\x*0.6+0.6); 
            }

            \draw[thick] (2.1,2.4) -- (2.1,3); 
            \draw[thick] (2.1,3) -- (5.6,3); 
            
            \foreach \x in {5,...,7} {
                \draw[thick] (0.7*\x+ 2.1,\x*0.6) -- (0.7*\x+2.1,\x*0.6+0.6); 
                \draw[thick] (0.7*\x+ 2.1,\x*0.6+0.6) -- (0.7*\x+ 2.8,\x*0.6+0.6); 
            }

            \draw[thick] (7.7, 7*0.6+0.6) -- (7.7, 7*0.6+1.2);

            \foreach \x in {1,...,4} {
                \draw[thick] (7.7-\x*0.7 + 0.7, \x*0.6) -- 
                             (7.7-\x*0.7+ 0.7, \x*0.6+0.6); 
                \draw[thick] (7.7-\x*0.7+ 0.7, \x*0.6+0.6) -- 
                             (7.7-\x*0.7, \x*0.6+0.6); 
            }

            \foreach \x in {5,...,7} {
                \draw[thick] (7.7-\x*0.7 - 2.1, \x*0.6) -- 
                             (7.7-\x*0.7-2.1, \x*0.6+0.6); 
                \draw[thick] (7.7-\x*0.7-2.1, \x*0.6+0.6) -- 
                             (7.7-\x*0.7-2.8, \x*0.6+0.6); 
            }

            \draw[thick] (0, 7*0.6+0.6) -- (0, 7*0.6+1.2);

            \foreach \x in {1,...,4} {
                \filldraw[black] (0.7*\x-0.7,0.6) circle (2pt) node[anchor=north] {\x};
            }
            \foreach \x in {5,...,8} {
                \filldraw[black] (0.7*\x+ 3*0.7 ,0.6) circle (2pt) node[anchor=north] {\x};
            }

        \end{tikzpicture}
        \caption{Stair diagram for $n=8$. The red area corresponds to the committee $\{4,5\}$; the green area to $\{5,6\}$.}
    \end{subfigure}
    \caption{Stair diagrams for 9 and 8 agents. The common cost incurred by any committee is shown in gray; the additional cost of two specific committees is shown in red and green.}
    \label{fig:stair-diagram-comparison}
\end{figure}

In the remainder of this section, we establish tight distortion bounds of $\frac{4}{3}$ and $2$ for odd and even values of $n$, respectively. To derive these bounds, we frequently utilize the inequalities above to bound the social cost of the optimal committee.

\paragraph*{\bf Odd number of agents} 
We first focus on odd values of $n$.
For $n=3$, it is easy to see that the optimal set corresponds to the median agent and the agent that the median prefers among the others, which yields a simple rule with distortion $1$.
For $n\geq 5$ we introduce a voting rule called \favCouple.
For an election $\instance=(A,k,\succ)$, we say that agents $a,b\in A$ are a \textit{couple} if they rank each other above all other agents; i.e., if $b\succ_a c$ and $a\succ_b c$ for every $c\in A\setminus \{a,b\}$.
Note that each agent can take part in at most one couple.
\favCouple\ selects the closest couple to the median when restricting to the five middle agents.

\begin{Vrule}[\favCouple]\label{closest couple}
For a preference profile $\succ$, compute the order from left to right $1,\ldots,n$ and let $m=\frac{n+1}{2}$ be the median agent. 
If there is a couple among the sets $\{m-1,m\}$ and $\{m,m+1\}$, return it.
Else, return $\{m+1,m+2\}$ if $m+2\succ_m m-2$ and return $\{m-2,m-1\}$ otherwise.
\end{Vrule}

On an intuitive level, this voting rule selects two consecutive agents who are both close to each other and to the median agent. 
The restriction to middle agents is necessary; simply choosing an arbitrary couple can lead to a distortion of up to $2$.
For example, this is the case if there are $n$ agents with distances $d(a,a+1)=1+a\varepsilon$ for all $a\in [n]$ and a small $\varepsilon>0$, as the only couple is $\{1,2\}$ with a social cost of approximately $\frac{n^2}{2}$, while the committee consisting of the median agent and any neighbor is close to $\frac{n^2}{4}$.
We now show that this rule provides the best-possible distortion of $\frac{4}{3}$ for an odd number of agents.
To establish the distortion of \favCouple, we address different cases depending on the set selected by this rule and the optimal set.
In each of them, we can use the selection condition of the rule to bound the social cost of the set selected by it from above and the optimal social cost from below and achieve a ratio of at most $\frac{4}{3}$ between them.
Intuitively, in situations like the one illustrated in \Cref{fig:stair-diagram-comparison}.(a), \favCouple\ can select a suboptimal committee (e.g.\ the red one instead of the green one), but this imposes several restriction on the distances, such as  $d(3,4)\leq \min\{d(4,5),d(6,7)\}$ and $2d(3,4)\leq d(5,7)$ in this example.
In particular, this implies lower bounds on the common cost incurred by any voting rule.
We study and apply these inequalities carefully to obtain our guarantee.

For the lower bound, on the other hand, we consider $n=5$ agents and two metrics illustrated in \Cref{fig:lower-bound-metrics}.
If a rule selects $\{1,2\}$, under metric $d_1$ the agents' costs, from left to right, are $1$, $1$, $2$, $4$, and $4$, accounting for a total cost of $12$, while the optimal committee $\{4,5\}$ induces a social cost of $4+3+2+0+0=9$.
Similarly, note that if a rule selects $\{2,3\}$, $\{3,4\}$, or $\{4,5\}$, then the social costs under metric $d_2$ are $8$, $8$, and $10$, respectively, while the optimal committee $\{1,2\}$ induces a social cost of $6$.
In all cases, the ratio is at least $\frac{4}{3}$, so the bound follows.

\begin{restatable}{theorem}{thmUtilitarianTwoCostOdd}\label{thm:utilitarian-2cost-k2-odd} 
    For every odd $n\geq 5$, \favCouple\ achieves a distortion of $\frac{4}{3}$ for utilitarian $2$-cost.  
    Moreover, there exists $n\in \NN$ such that, for every $(n,2)$-voting rule $f$, we have $\dist(f) \geq \frac{4}{3}$ for utilitarian $2$-cost.
\end{restatable}
\begin{proof}

We consider an arbitrary election $\instance=(A,k,\succ)$ with $n\geq 5$ and $A=[n]$, and a consistent metric $d\ \rhd \succ$.
We denote the five middle agents by $a_1, \dots, a_5$ from left to right, with $a_3$ being the median agent. 
We let $S$ denote the committee selected by \favCouple\ and $S^*$ denote the optimal committee for the metric $d$.
We analyze two main cases, depending on whether the rule selects the median agent or not.

    \paragraph*{\bf Case 1: $a_3\in S$}   
    w.l.o.g., we assume that $a_2 \succ_{a_3} a_4$, which implies that the selected committee is $S = \{a_2, a_3\}$. This implies that agents $a_2$ and $a_3$ form a couple, and both $d(a_2, a_3) \leq d(a_1, a_2)$ and $d(a_2, a_3) \leq d(a_3, a_4)$ hold.
    Therefore,
    \begin{equation}\label{eq:fSelectsM}
        d(a_1,a_5) \geq 3 \cdot d(a_2,a_3), \quad d(a_2,a_4) \geq 2 \cdot d(a_2,a_3).
    \end{equation}

For each $i \leq \frac{n-1}{2}$, the joint cost of $S$ for agents $i$ and $n-i+1$ is given by 
\[
    \SC(S,i;d)+\SC(S,n-i+1;d) = d(i, a_3) + d(a_2,n-1+1) = d(i,n-i+1) + d(a_2, a_3).
\]
Since the median agent incurs a cost of $\SC(A,a_3;d)=d(a_2, a_3)$, we obtain:
\begin{align*}
\SC(S, A; d) & = \sum_{i=1}^{\frac{n-3}{2}} d(i,n-i+1) + d(a_2,a_4) + \left(\frac{n-1}{2}\right) d(a_2,a_3) + d(a_2,a_3)\\
& = \sum_{i=1}^{\frac{n-3}{2}} d(i,n-i+1) + \left(\frac{n+1}{2}\right) d(a_2,a_3) + d(a_2,a_4).
\end{align*}
On the other hand, by Lemma~\ref{lem:2-social-cost}, we have:
\begin{align*}
\SC(S^*, A; d) & \geq \sum_{i=1}^{\frac{n-1}{2}} d(i,n-i+1) + \SC(\{a_3\}, A; d) \\
& \geq \sum_{i=1}^{\frac{n-3}{2}} d(i,n-i+1)  + d(a_2,a_4) + d(a_2,a_3),
\end{align*}
where we used, for the second inequality, that the cost of the median agent is at least $d(a_2,a_3)$ due to the assumption that $a_2 \succ_{a_3} a_4$.
Thus, the distortion is:
\begin{align*}
\dist(f) & = \frac{\SC(S, A; d)}{\SC(S^*, A; d)}\\
& \leq \frac{\sum_{i=1}^{\frac{n-3}{2}} d(i,n-i+1) + \left(\frac{n+1}{2}\right) d(a_2,a_3) + d(a_2,a_4)}{\sum_{i=1}^{\frac{n-3}{2}} d(i,n-i+1) + d(a_2,a_3) + d(a_2,a_4)}\\
&\leq \frac{\left(\frac{n-3}{2}\right) \cdot 3 \cdot d(a_2,a_3) + \left(\frac{n+1}{2}\right) \cdot d(a_2,a_3) + 2 \cdot d(a_2,a_3)}{\left(\frac{n-3}{2}\right) \cdot 3\cdot d(a_2,a_3) + d(a_2,a_3) + 2\cdot d(a_2,a_3)} = \frac{\frac{4n-8}{2} +2}{\frac{3n-9}{2} + 3} = \frac{4n-4}{3n-3} = \frac{4}{3},
\end{align*}
where the second inequality follows from inequalities~\eqref{eq:fSelectsM} and the fact that $d(i,n-i+1)\geq d(1,5)$ for every $i\leq \frac{n-3}{2}$.
This concludes the proof for this case.
 
\paragraph{\bf Case 2: $a_3\notin S$}     
In this case, we either have $S= \{a_1, a_2\}$ or $S= \{a_4, a_5\}$; we assume the former w.l.o.g.. From the definition of \favCouple, this implies that $\{a_2,a_3\}$ and $\{a_3,a_4\}$ are not couples, so we must have $a_1\succ_{a_2} a_3$ and $a_5\succ_{a_4} a_3$.
It also implies that $a_1\succ_{a_3} a_5$, since $\{a_4,a_5\}$ would be selected otherwise. 
In terms of distances:
\begin{equation}\label{eq:mNotInf1}
    d(a_2, a_3) \geq d(a_1, a_2), \quad d(a_3,a_4) \geq d(a_4,a_5), \quad d(a_3,a_5) \geq d(a_1,a_3).
\end{equation}
    
Similarly as before, the social cost of the selected committee is
\begin{align*}
\SC(S, A; d) &= \sum_{i=1}^{\frac{n-3}{2}} d(i,n-i+1) + \left(\frac{n-3}{2}\right) d(a_1,a_2) + d(a_1,a_2) + d(a_1,a_3) + d(a_1,a_4)\\
&= \sum_{i=1}^{\frac{n-3}{2}} d(i,n-i+1) + \left(\frac{n+3}{2}\right) d(a_1,a_2) + d(a_2, a_3) + d(a_2,a_4).
\end{align*}

We now consider two cases depending on whether $a_3$ is in the optimal committee.

\paragraph*{\bf Case 2.1: $a_3 \in S^*$.}  
If the median agent is selected in the optimal committee, we have from \Cref{lem:2-social-cost} that
\begin{equation}
\SC(S^*, A; d) \geq \sum_{i=1}^{\frac{n-3}{2}} d(i,n-i+1) + d(a_2,a_4) + \left(\frac{n-1}{2} + 1\right) \min \{d(a_2,a_3), d(a_3,a_4)\}.\label{eq:lb-committee-m-not-alg-opt}
\end{equation}

We now claim that $\big(\frac{n-1}{2} + 1\big) \min \{d(a_2,a_3), d(a_3,a_4)\} \geq \frac{3}{2}d(a_1,a_3)$.
Indeed, if we have $\min \{d(a_2,a_3), d(a_3,a_4)\} = d(a_2, a_3)$, this holds because $\frac{n-1}{2} + 1\geq 3$ and, due to inequalities \eqref{eq:mNotInf1}, $3 d(a_2, a_3)\geq \frac{3}{2}d(a_1,a_3)$.
If $\min \{d(a_2,a_3), d(a_3,a_4)\} = d(a_3, a_4)$, this holds because $\frac{n-1}{2} + 1\geq 3$ and, due to inequalities \eqref{eq:mNotInf1}, $3 d(a_3,a_4) \geq \frac{3}{2}d(a_3,a_5) \geq \frac{3}{2}d(a_1,a_3)$.

Replacing in inequality \eqref{eq:lb-committee-m-not-alg-opt}, we obtain
\[
\SC(S^*, A; d) \geq \sum_{i=1}^{\frac{n-3}{2}} d(i,n-i+1) + d(a_2,a_4) + \frac{3}{2}\cdot d(a_1, a_2) +\frac{3}{2}\cdot d(a_2,a_3).
\]
Thus, the distortion is
\begin{align*}
\dist(f) = \frac{\SC(S, A; d)}{\SC(S^*, A; d)} & \leq \frac{\sum_{i=1}^{\frac{n-3}{2}} d(i,n-i+1) + \left(\frac{n+3}{2}\right) d(a_1,a_2) + d(a_2, a_3) + d(a_2,a_4)}{\sum_{i=1}^{\frac{n-3}{2}} d(i,n-i+1) +d(a_2,a_4) + \frac{3}{2}\cdot d(a_1, a_2)+ \frac{3}{2}\cdot d(a_2,a_3)}\\
& \leq \frac{\left(\frac{n-3}{2}\right) \cdot 4 \cdot d(a_1,a_2) + \left(\frac{n+3}{2}\right) d(a_1,a_2) + d(a_1,a_2) + 2\cdot d(a_1,a_2)}{\left(\frac{n-3}{2}\right) \cdot 4 \cdot d(a_1,a_2) + 2 \cdot d(a_1,a_2) + \frac{3}{2}\cdot d(a_1, a_2)+ \frac{3}{2}\cdot d(a_1, a_2)}\\
& \leq \frac{\left(\frac{n-3}{2}\right) \cdot 4 + \left(\frac{n+3}{2}\right) +1+ 2}{\left(\frac{n-3}{2}\right) \cdot 4  + 2  + \frac{3}{2}  + \frac{3}{2}} \\
& = \frac{4n -12 +n+3+2+4}{4n -12 + 4 + 3 + 3} = \frac{5n-3}{4n-2} \leq \frac{5}{4} \leq \frac{4}{3},
\end{align*}
where the second inequality follows by applying inequalities \eqref{eq:mNotInf1} and the fact that $d(i,n-i+1)\geq d(1,5)$ for every $i\leq \frac{n-3}{2}$. We conclude the distortion bound of $\frac{4}{3}$ for this case.
    
\paragraph{\bf Case 2.2: $a_3 \notin S^*$.}  
We begin by rewriting the social cost of $S$ more conveniently as
\begin{align*}
\SC(S, A; d) & = \sum_{i=1}^{\frac{n-3}{2}} d(i,n-i+1) + \left(\frac{n-3}{2}\right) d(a_1,a_2) + d(a_1,a_2) + d(a_1,a_3) + d(a_1,a_4)\\
& = \sum_{i=1}^{\frac{n-3}{2}} d(i,n-i+1) + \left(\frac{n+1}{2}\right)d(a_1,a_2) + d(a_1, a_3) + d(a_2,a_3) + d(a_3, a_4)\\
& \leq \sum_{i=1}^{\frac{n-3}{2}} d(i,n-i+1) + \left(\frac{n+1}{2}\right) d(a_1,a_2) + d(a_3, a_5) + d(a_2,a_3) + d(a_3, a_5),
\end{align*}
where the last inequality follows from inequalities \eqref{eq:mNotInf1}.
We distinguish two further cases to bound the social cost of the optimal committee from below, depending on whether the optimal committee selects agents from the left or from the right side of the median.

\paragraph*{\bf Case 2.2.1: $S^*\subseteq \{a_4,a_5,\ldots,n\}$}
If the optimal committee selects an agent on the right side of the median agent, its social cost satisfies
\[
\SC(S^*, A; d) \geq \sum_{i=1}^{\frac{n-3}{2}} d(i,n-i+1) + d(a_2,a_5) + d(a_3, a_5) =\sum_{i=1}^{\frac{n-3}{2}} d(i,n-i+1) + d(a_2,a_3) + 2 d(a_3,a_5).
\]

Thus, the distortion is
\begin{align*}
\dist(f) = \frac{\SC(S, A; d)}{\SC(S^*, A; d)} & \leq \frac{\sum_{i=1}^{\frac{n-3}{2}} d(i,n-i+1) + \left(\frac{n+1}{2}\right) d(a_1,a_2) + d(a_2,a_3) + 2 d(a_3, a_5)}{\sum_{i=1}^{\frac{n-3}{2}} d(i,n-i+1) + d(a_2,a_3) + 2 d(a_3,a_5)}\\
& \leq \frac{\left(\frac{n-3}{2}\right) \cdot 4 \cdot d(a_1,a_2) + \left(\frac{n+1}{2}\right) d(a_1,a_2) + d(a_1,a_2) + 2 \cdot 2\cdot  d(a_1, a_2)}{\left(\frac{n-3}{2}\right) \cdot 4 \cdot d(a_1,a_2) + d(a_1,a_2) + 2 \cdot 2\cdot d(a_1,a_2)}\\
& \leq \frac{\left(\frac{n-3}{2}\right) \cdot 4 + \left(\frac{n+1}{2}\right)  + 1 + 4}{\left(\frac{n-3}{2}\right) \cdot 4  +1 + 4}\\
& = \frac{4n - 12 + n + 1 + 2+8 }{4n-12 +2+8} = \frac{5n - 1}{4n -2} \leq \frac{4}{3},
\end{align*}
where we used inequalities \eqref{eq:mNotInf1} for the second inequality.

\paragraph*{\bf Case 2.2.2: $S^*\subseteq \{1,\ldots, a_1,a_2\}$.}  
If $S^*=S$, the distortion is trivially $1$ and we conclude. Otherwise, the social cost of $S^*$ satisfies
\[
\SC(S^*, A; d) \geq \sum_{i=1}^{\frac{n-3}{2}} d(i,n-i+1) + d(a_1, a_2) + d(a_1,a_3) + d(a_1, a_4).
\]

Thus, the distortion is
\begin{align*}
\dist(f) & = \frac{\SC(S, A; d)}{\SC(S^*, A; d)}  \\
& \leq \frac{\sum_{i=1}^{\frac{n-3}{2}} d(i,n-i+1) + \left(\frac{n-3}{2}\right) d(a_1,a_2)+ d(a_1,a_2) + d(a_1,a_3) + d(a_1,a_4)}{\sum_{i=1}^{\frac{n-3}{2}} d(i,n-i+1) + d(a_1, a_2) + d(a_1,a_3) + d(a_1, a_4)}\\
& \leq \frac{\left(\frac{n-3}{2}\right) \cdot 4 d(a_1,a_2) + \left(\frac{n-3}{2}\right) d(a_1,a_2)+ d(a_1,a_2) + 2 d(a_1,a_2) + 3 d(a_1,a_2)}{\left(\frac{n-3}{2}\right) \cdot 4 d(a_1,a_2) + d(a_1, a_2) + 2 d(a_1,a_2) + 3 d(a_1, a_2)}\\
&\leq \frac{4n-12 + n-3 + 2+4+6}{4n-12+2+4+6} = \frac{5n - 3}{4n} < \frac{4}{3},
\end{align*}
where we used inequalities \eqref{eq:mNotInf1} for the second inequality.
This concludes the proof of the distortion of \favCouple.

For the lower bound, we fix $n=5$ and an arbitrary $(n,2)$-voting rule $f$, consider the profile $\succ\in \calL^5(5)$ defined as
    \begin{align*}
        1 & \succ_1 2 \succ_1 3 \succ_1 4 \succ_1 5,\\
        2 & \succ_2 1 \succ_2 3 \succ_2 4 \succ_2 5,\\
        3 & \succ_3 2 \succ_3 1 \succ_3 4 \succ_3 5,\\
        4 & \succ_4 5 \succ_4 3 \succ_4 2 \succ_4 1,\\
        5 & \succ_5 4 \succ_5 3 \succ_5 2 \succ_5 1,
    \end{align*}
    and consider the election $\instance=(A,2,\succ)$ with $A=[5]$. We distinguish two cases depending on the set of agents $S=f(\succ)$ selected by the rule.
    
    Suppose first that $S= \{1,2\}$.
    We take the distance metric $d_1$ on $A$ given by positions $x_1=0$, $x_2=1$, $x_3=2$, and $x_4=x_5=4$.
    Note that $d_1\ \rhd \succ$; see \Cref{fig:lower-bound-metrics} for an illustration.
    Since $\SC(\{1,2\},A;d_1)=12$, and $\SC(\{4,5\},A;d_1)=9$, we obtain
    \[
        \dist(f(\succ), \instance) \geq \frac{\SC(S,A;d_1)}{\min_{S'\in {A\choose 2}}\SC(S',A;d_1)} \geq \frac{12}{9} = \frac{4}{3}.
    \]
    
    If $S\in \{\{2,3\},\{3,4\},\{4,5\}\}$, we consider the distance metric $d_2$ on $A$ given by positions $x_1=x_2=0$, $x_3=1$, $x_4=2$, and $x_5=3$.
    Note that $d_2\ \rhd \succ$; see \Cref{fig:lower-bound-metrics} for an illustration.
    Since $\SC(\{2,3\},A;d_2)=\SC(\{3,4\},A;d_2)=8$ and $\SC(\{4,5\},A;d_2)=10$, whereas $\SC(\{1,2\},A;d_2)=6$, we obtain
    \[
        \dist(f(\succ), \instance) \geq \frac{\SC(S,A;d_2)}{\min_{S'\in {A\choose 2}}\SC(S',A;d_2)} \geq \frac{8}{6} = \frac{4}{3}.
    \]
    
    Since $\dist(f(\succ),\instance) \geq \frac{4}{3}$ in all these cases and sets of non-consecutive agents can only induce a larger social cost, we conclude that $\dist(f)\geq \frac{4}{3}$.  
\end{proof}

\begin{figure}[t]
\centering

\begin{tikzpicture}[scale=1, every node/.style={font=\footnotesize}]
    \node[anchor=east] at (-0.5, 0) {\textbf{Metric $d_1$}};
    \filldraw[black] (0,0) circle (2.5pt) node[anchor=north] {$1$};
    \filldraw[black] (2,0) circle (2.5pt) node[anchor=north] {$2$};
    \filldraw[black] (4,0) circle (2.5pt) node[anchor=north] {$3$};
    \filldraw[black] (8,0) circle (2.5pt) node[anchor=north] {$4,5$};
   
    \draw[-] (0,0) -- (2,0) node[midway, anchor=south] {$1$};
    \draw[-] (2,0) -- (4,0) node[midway, anchor=south] {$1$};
    \draw[-] (4,0) -- (8,0) node[midway, anchor=south] {$2$};

\begin{scope}[yshift=-1cm]
    \node[anchor=east] at (-0.5, 0) {\textbf{Metric $d_2$}};
    \filldraw[black] (0,0) circle (2.5pt) node[anchor=north] {$1,2$};
    \filldraw[black] (2,0) circle (2.5pt) node[anchor=north] {$3$};
    \filldraw[black] (4,0) circle (2.5pt) node[anchor=north] {$4$};
    \filldraw[black] (6,0) circle (2.5pt) node[anchor=north] {$5$};
    \draw[-] (0,0) -- (2,0) node[midway, anchor=south] {$1$};
    \draw[-] (2,0) -- (4,0) node[midway, anchor=south] {$1$};
    \draw[-] (4,0) -- (6,0) node[midway, anchor=south] {$1$};
\end{scope}
\end{tikzpicture}

\caption{Metrics considered in the proof of \Cref{thm:utilitarian-2cost-k2-odd}.}
\label{fig:lower-bound-metrics}
\end{figure}

\paragraph*{\bf Even number of agents}
When $n$ is even, we show that the voting rule that selects the two median agents attains the best-possible distortion of $2$.

\begin{restatable}{proposition}{propUtilitarianTwoCostEven}\label{prop:utilitarian-2cost-k2-even}
For an even number of agents $n$, the voting rule that selects the two median agents achieves a distortion of $2$ for utilitarian $2$-cost.
Moreover, there exists $n\in \NN$ such that, for every $(n,2)$-voting rule $f$, we have $\dist(f)\geq 2$ for utilitarian $2$-cost.
\end{restatable}
\begin{proof}
We consider an arbitrary election $\instance=(A,k,\succ)$ with even $n\geq 4$ and $A=[n]$, and a consistent metric $d\ \rhd \succ$. Note that the assumption $n\geq 4$ is w.l.o.g.\ since, for $n=2$, a distortion of $1$ is trivially achieved.
We let $m_1 = \frac{n}{2}$ and $m_2 = \frac{n}{2} + 1$ denote the left and right median, respectively, $S=\{m_1,m_2\}$ denote the committee selected by the rule, and $S^*$ denote the optimal committee for the metric $d$.
The social cost of $S$ is
\[
\SC(S, A; d) = \sum_{i=1}^{\frac{n}{2}} d(i, n-i+1) + \frac{n}{2} d(m_1, m_2),
\]
whereas \Cref{lem:2-social-cost} implies a lower bound on the social cost of the optimal committee of
\[
\SC(S^*, A; d) \geq \sum_{i=1}^{\frac{n}{2}} d(i, n-i+1).
\]
Thus, the distortion of the voting rule is
\[
\dist(f) = \frac{\SC(S, A; d)}{\SC(S^*, A; d)} \leq \frac{\sum_{i=1}^{\frac{n}{2}} d(i, n-i+1) + \frac{n}{2} d(m_1, m_2)}{\sum_{i=1}^{\frac{n}{2}} d(i, n-i+1)} \leq \frac{2\sum_{i=1}^{\frac{n}{2}} d(i, n-i+1)}{\sum_{i=1}^{\frac{n}{2}} d(i, n-i+1)} = 2,
\]
where the second inequality follows from the fact that $d(m_1, m_2) \leq d(i, n-i+1)$ for any $i \leq \frac{n}{2}$.
Thus, the voting rule achieves a distortion of at most $2$.

For the lower bound, we fix $n=4$ and an arbitrary $(n,2)$-voting rule $f$, consider the profile $\succ\in \calL^4(4)$ defined as
    \begin{align*}
        1 & \succ_1 2 \succ_1 3 \succ_1 4,\\
        2 & \succ_2 1 \succ_2 3 \succ_2 4,\\
        3 & \succ_3 4 \succ_3 2 \succ_3 1,\\
        4 & \succ_4 3 \succ_4 2 \succ_4 1,
    \end{align*}
    and consider the election $\instance=(A,2,\succ)$ with $A=[4]$. We distinguish three cases depending on the set of agents $S=f(\succ)$ selected by the rule.
    
    Suppose first that $S= \{3,4\}$.
    We take the distance metric $d_1$ on $A$ given by positions $x_1=x_2=0$, $x_3=1$, and $x_4=2$.
    Note that $d_1\ \rhd \succ$; see \Cref{fig:2-winner-instances} for an illustration.
    Since $\SC(\{3,4\},A;d_1)=6$, and $\SC(\{1,2\},A;d_1)=3$, we obtain
    \[
        \dist(f(\succ), \instance) \geq \frac{\SC(S,A;d_1)}{\min_{S'\in {A\choose 2}}\SC(S',A;d_1)} \geq \frac{6}{3} = 2.
    \]
    
    If $S=\{1,2\}$, we consider the distance metric $d_2$ on $A$ given by positions $x_1=0$, $x_2=1$, and $x_3=x_4=2$.
    Note that $d_2\ \rhd \succ$; see \Cref{fig:2-winner-instances} for an illustration.
    Since $\SC(\{1,2\},A;d_2)=6$ and $\SC(\{3,4\},A;d_2)=3$, we obtain
    \[
        \dist(f(\succ), \instance) \geq \frac{\SC(S,A;d_2)}{\min_{S'\in {A\choose 2}}\SC(S',A;d_2)} \geq \frac{6}{3} = 2.
    \]

    Finally, if $S=\{2,3\}$, we consider the distance metric $d_3$ on $A$ given by positions $x_1=x_2=0$ and $x_3=x_4=2$.
    Note that $d_3\ \rhd \succ$; see \Cref{fig:2-winner-instances} for an illustration.
    Since $\SC(\{2,3\},A;d_3)=8$ and $\SC(\{1,2\},A;d_3)=\SC(\{3,4\},A;d_3)=4$, we obtain
    \[
        \dist(f(\succ), \instance) \geq \frac{\SC(S,A;d_3)}{\min_{S'\in {A\choose 2}}\SC(S',A;d_3)} \geq \frac{8}{4} = 2.
    \]
    
    Since $\dist(f(\succ),\instance) \geq 2$ in all these cases and sets of non-consecutive agents can only induce a larger social cost, we conclude that $\dist(f)\geq 2$. 
\end{proof} 

\begin{figure}[t]
\centering

\begin{tikzpicture}[scale=1, every node/.style={font=\footnotesize}]
    \node[anchor=east] at (-0.5, 0) {\textbf{Metric $d_1$}};
    \filldraw[black] (0,0) circle (2.5pt) node[anchor=north] {$1 , 2$};
    \filldraw[black] (3,0) circle (2.5pt) node[anchor=north] {$3$};
    \filldraw[black] (6,0) circle (2.5pt) node[anchor=north] {$4$};
    \draw[-] (0,0) -- (3,0) node[midway, anchor=south] {$1$};
    \draw[-] (3,0) -- (6,0) node[midway, anchor=south] {$1$};

\begin{scope}[yshift=-1cm]
    \node[anchor=east] at (-0.5, 0) {\textbf{Metric $d_2$}};
    \filldraw[black] (0,0) circle (2.5pt) node[anchor=north] {$1$};
    \filldraw[black] (3,0) circle (2.5pt) node[anchor=north] {$2$};
    \filldraw[black] (6,0) circle (2.5pt) node[anchor=north] {$3, 4$};

    \draw[-] (0,0) -- (3,0) node[midway, anchor=south] {$1$};
    \draw[-] (3,0) -- (6,0) node[midway, anchor=south] {$1$};
\end{scope}

\begin{scope}[yshift=-2cm]
    \node[anchor=east] at (-0.5, 0) {\textbf{Metric $d_3$}};
    \filldraw[black] (0,0) circle (2.5pt) node[anchor=north] {$1, 2$};
    \filldraw[black] (6,0) circle (2.5pt) node[anchor=north] {$3,4$};
    \draw[-] (0,0) -- (6,0) node[midway, anchor=south] {$2$};
\end{scope}
\end{tikzpicture}

\caption{Metrics considered in the proof of \Cref{prop:utilitarian-2cost-k2-even} and \Cref{prop:egalitarian-k1}.}
\label{fig:2-winner-instances}
\end{figure}

%% file: a-Egalitarian.tex
In this section, we study the worst-case distortion achievable by voting rules in the context of peer selection with egalitarian social cost.
Recall that, in this case, given a set of agents $A$, a committee size $k$, and a distance metric $d$, the social cost of a committee $S'\in {A\choose k}$ corresponds to the maximum cost of this committee for some agent $a\in A$:
\[
    \SC(S',A;d) = \max\{\SC(S',a;d) \mid a\in A\}.
\]
We will start the section with the simple case $k=1$, where $S'=\{s\}$ for some $s\in A$ and thus $\SC(S',a;d)$ is simply $d(a,s)$ for every $a\in A$.
In \Cref{subsec:egalitarian-additive,subsec:egalitarian-qcost} we explore the case of general committee size under additive and $q$-cost candidate-aggregation functions, respectively.

\subsection{Warm-up: Selecting a Single Candidate}\label{subsec:egalitarian-k1}

In the context of single-candidate elections, any voting rule achieves a distortion of $2$, since for any election and any selected candidate, one of the extreme agents will incur a cost of at least half of its distance to the other extreme agent, and no rule can induce a social cost larger than the distance between the extreme agents.
Despite its simplicity, it is not hard to see that this is the best possible distortion by considering four agents and two metrics as in \Cref{fig:2-winner-instances}. 
If a rule selects an agent in $\{1,2\}$, the cost incurred by agent $4$ under metric $d_1$ is $2$, while selecting agent $3$ would imply a cost of at most $1$ for every agent.
The situation is analogous under metric $d_2$ if a rule selects agent $3$ or $4$.
The following proposition formally states these facts.
\begin{restatable}{proposition}{propEgalitariankOne}\label{prop:egalitarian-k1}
    For every $n\in \NN$, any $(n,1)$-voting rule has distortion $2$ for egalitarian social cost. 
    There exists $n\in \NN$ such that, for every $(n,1)$-voting rule $f$, $\dist(f) \geq 2$ for egalitarian social cost.
\end{restatable}

\begin{proof}
    Fix $n\in \NN$ and an $(n,1)$-voting rule $f$ arbitrarily. 
    Let $\succ\ \in \calL^n(n)$ be any preference profile on $A=[n]$ and let $s$ be the agent that $f$ outputs for this profile, i.e., $S = f(\succ)$ and $S =\{s\}$.
    We denote the agents by $\{1,\ldots,n\}$ from left to right, and we let $d\ \rhd \succ$ be any consistent distance metric.
    It is clear that, on the one hand, we have
    \begin{equation}
        \SC(\{s\},A;d) = \max\{d(a,s)\mid a\in A\} \leq \max\{d(a,b)\mid a,b\in A\} = d(1,n).\label{ineq:egalitarian-k1-ub-sc-alg}
    \end{equation}
    On the other hand, for every agent $b\in A$ we have that $d(1,b)+d(b,n)=d(1,n)$ and, therefore, $\max\{d(1,b),d(b,n)\} \geq \frac{d(1,n)}{2}$.
    This implies
    \begin{equation}
        \min_{S'\in {A\choose 1}} \SC(S',A;d) = \min_{b\in A} \max\{d(a,b)\mid a\in A\} = \min_{b\in A}\max\{d(1,b),d(b,n)\} \geq \frac{d(1,n)}{2}.\label{ineq:egalitarian-k1-lb-sc-opt}
    \end{equation}
    Combining \cref{ineq:egalitarian-k1-ub-sc-alg,ineq:egalitarian-k1-lb-sc-opt}, we directly obtain that $\dist(f)\leq 2$.

    For the second claim, we denote $S = f(\succ)$, and we fix $n=4$ and an arbitrary $(n,1)$-voting rule $f$, consider the profile $\succ\in \calL^4(4)$ defined as
    \begin{align*}
        1 & \succ_1 2 \succ_1 3 \succ_1 4,\\
        2 & \succ_2 1 \succ_2 3 \succ_2 4,\\
        3 & \succ_3 4 \succ_3 2 \succ_3 1,\\
        4 & \succ_4 3 \succ_4 2 \succ_4 1,
    \end{align*}
    and consider the election $\instance=(A,1,\succ)$ with $A=[4]$. We distinguish two cases depending on the agent selected by $f$.
    
    Suppose first that $S\in \{1,2\}$.
    We take the distance metric $d_1$ on $A$ given by positions $x_1=x_2=0$, $x_3=1$, and $x_4=2$.
    It is not hard to check that $d_1\ \rhd \succ$; see \Cref{fig:2-winner-instances} for an illustration.
    Since $\SC(\{1\},A;d_1)=2$, $\SC(\{2\},A;d_1)=2$, and $\SC(\{3\},A;d_1)=1$, we obtain
    \[
        \dist(f(\succ), \instance) \geq \frac{\SC(S,A;d_1)}{\min_{a\in A}\SC(\{a\},A;d_1)} \geq \frac{\SC(\{2\},A;d_1)}{\SC(\{3\},A;d_1)} = 2.
    \]

    Similarly, if $S\in \{3,4\}$, we consider the distance metric $d_2$ on $A$ given by positions $x_1=0$, $x_2=1$, $x_3=x_4=2$.
    It is not hard to check that $d_2\ \rhd \succ$; see \Cref{fig:2-winner-instances} for an illustration.
    Since $\SC(\{3\},A;d_2)=2$, $\SC(\{4\},A;d_2)=2$, and $\SC(\{2\},A;d_2)=1$, we obtain
    \[
        \dist(f(\succ), \instance) \geq \frac{\SC(S,A;d_2)}{\min_{a\in A}\SC(\{a\},A;d_2)} \geq \frac{\SC(\{3\},A;d_2)}{\SC(\{2\},A;d_2)} = 2.
    \]

    Since $\dist(f(\succ),\instance) \geq 2$ both when $S\in \{1,2\}$ and when $S\in \{3,4\}$, we conclude that $\dist(f)\geq 2$.
\end{proof}

\subsection{Egalitarian Additive Social Cost}\label{subsec:egalitarian-additive}

In this section, we study voting rules in the context of egalitarian additive social cost, defined as the maximum over agents of the sum of the distances from the agent to all selected candidates.
That is, for a set of agents $A$, a committee size $k$, and a distance metric $d$, the social cost of a committee $S'\in {A\choose k}$ is 
\[
    \SC(S',A;d) = \max\bigg\{\sum_{s\in S'} d(a,s) \;\Big\vert\; a\in A\bigg\}.
\]

We begin with a simple  observation: When $k=2$ candidates are to be selected, a simple rule selecting both extreme candidates achieves the best-possible distortion of $1$.
Intuitively, this voting rule makes sense because, for any selected committee, (1) the cost of the committee is maximized for one of the extreme agents, and (2) the sum of the costs of the committee for both extreme agents is fixed (and equal to two times the distance between them).
Thus, selecting both extreme agents ensures they incur the same cost and minimizes the maximum cost between them.
This rule and its distortion will be covered as a special case of the rule and result we introduce in what follows.

For larger $k$, the above intuition about the cost of any committee being maximized for the extreme agents remains true.
We state this property, which will be exploited in the development and analysis of a voting rule guaranteeing a constant distortion, in the following lemma.

\begin{restatable}{lemma}{lemEgAddExtremes}\label{lem:eg-add-extreme-agents}
For every set of agents $A=[n]$, committee size $k$, committee $S'\in {A\choose k}$, and distance metric $d$, it holds that
\[
    \SC(S',A;d) = \max\{\SC(S',1;d),\SC(S',n;d)\}.
\]
\end{restatable}
\begin{proof}
Let $A=[n]$, $k$, $S'$, and $d$ be as in the statement, and recall that we refer to the agents sorted from left to right by $\{1,\ldots,n\}$.
We suppose towards a contradiction that there exists $a\in A$ such that $\SC(S',a;d) > \max\{\SC(S',1;d),\SC(S',n;d)\}$; i.e.,
\begin{equation}
    \sum_{s\in S'} d(a,s) > \max\bigg\{ \sum_{s\in S'} d(1,s), \sum_{s\in S'} d(s,n)\bigg\}.\label{ineq:eg-add-extreme-better}
\end{equation}
We now distinguish two cases. 
If $a$ has at least as many agents in $S'$ weakly to its left as strictly to its right; i.e., $|\{s\in S'\mid s\leq a\}| \geq |\{s\in S'\mid s>a\}|$, then
\begin{align*}
    \sum_{s\in S'} d(s,n) & = \sum_{s\in S': s\leq a} (d(a,s)+d(a,n)) + \sum_{s\in S': s>a} (d(a,s)-(d(a,s)-d(s,n))) \\
    & \geq \sum_{s\in S': s\leq a} (d(a,s)+d(a,n)) + \sum_{s\in S': s>a} (d(a,s)-d(a,n))\\
    & = \sum_{s\in S'} d(a,s) + (|\{s\in S': s\leq a\}| - |\{s\in S': s>a\}|) d(a,n) \\
    & \geq \sum_{s\in S'} d(a,s),
\end{align*}
a contradiction to \cref{ineq:eg-add-extreme-better}.
Analogously, if $|\{s\in S'\mid s\leq a\}| < |\{s\in S'\mid s>a\}|$, then
\begin{align*}
    \sum_{s\in S'} d(1,s) & = \sum_{s\in S': s> a} (d(1,a)+d(a,s)) + \sum_{s\in S': s\leq a} (d(a,s)-(d(a,s)-d(1,s))) \\
    & \geq \sum_{s\in S': s> a} (d(1,a)+d(a,s)) + \sum_{s\in S': s\leq a} (d(a,s)-d(1,a))\\
    & = \sum_{s\in S'} d(a,s) + (|\{s\in S': s>a\}|-|\{s\in S': s\leq a\}|) d(1,a) \\
    & \geq \sum_{s\in S'} d(a,s),
\end{align*}
a contradiction to \cref{ineq:eg-add-extreme-better}.
\end{proof}

Since for any set of agents $A$, committee size $k$, committee $S'\in {A\choose k}$, and distance metric $d$ we have that
\begin{equation}
    \SC(S',1;d)+\SC(S',n;d) = \sum_{a\in S'} (d(1,a)+d(a,n)) = kd(1,n),\label{eq:eg-add-sum-sc-extremes}
\end{equation}
where the last expression does not depend on $S'$, the previous lemma implies that the optimal committee will be the set that balances the cost for the extreme agents as much as possible.
With this in mind, it is natural to consider a generalization of the 
rule that select both extreme agents to larger committees. 
The rule, which we denote \kExtremes, simply returns the $\big\lfloor \frac{k}{2} \big\rfloor$ agents closest to one extreme and the $\big\lceil \frac{k}{2} \big\rceil$ agents closest to the other extreme.
\begin{Vrule}[\kExtremes]
    For a preference profile $\succ$, compute the order of agents from left to right $1,\ldots,n$ and return $S=\big\{1,\ldots,\big\lfloor \frac{k}{2} \big\rfloor\big\} \cup \big\{n-\big\lceil \frac{k}{2} \big\rceil+1,\ldots,n\big\}$.
\end{Vrule}
The following theorem states the distortion of this voting rule. 
It captures the previously claimed distortion of $1$ for $k=2$, and it approaches $\frac{3}{2}$ as $k$ grows.
This is best possible up to $O\big(\frac{1}{k}\big)$ terms, which vanish as $k$ grows.
The upper and lower bounds stated in this theorem are depicted in \Cref{fig:plot-egalitarian-add}.
\begin{figure}
    \centering
\begin{tikzpicture}[xscale=0.07,yscale=0.07]
\draw[thick,-latex] node[left] {$1$} (-1,0) -- (105,0) node[right] {$k$};
\draw[thick,-latex] node[below] {$0$} (0,-1) -- (0,55) node[above] {distortion};
\foreach \n in {10,20,30,40,50,60,70,80,90,100}{
    \draw (\n,1) -- (\n,-1) node[below] {$\n$};
}
\draw (1,10) -- (-1,10) node[left] {$1.1$};
\draw (1,20) -- (-1,20) node[left] {$1.2$};
\draw (1,30) -- (-1,30) node[left] {$1.3$};
\draw (1,40) -- (-1,40) node[left] {$1.4$};
\draw (1,50) -- (-1,50) node[left] {$1.5$};

\draw[semithick,blue] plot[mark=., mark color=none,mark options={fill=black,draw opacity=0},mark size=3pt] coordinates { (2, 0.0) (3, 33.3333) (4, 33.3333) (5, 45.0) (6, 40.0) (7, 47.619) (8, 42.8571) (9, 48.6111) (10, 44.4444) (11, 49.0909) (12, 45.4545) (13, 49.359) (14, 46.1538) (15, 49.5238) (16, 46.6667) (17, 49.6324) (18, 47.0588) (19, 49.7076) (20, 47.3684) (21, 49.7619) (22, 47.619) (23, 49.8024) (24, 47.8261) (25, 49.8333) (26, 48.0) (27, 49.8575) (28, 48.1481) (29, 49.8768) (30, 48.2759) (31, 49.8925) (32, 48.3871) (33, 49.9053) (34, 48.4848) (35, 49.916) (36, 48.5714) (37, 49.9249) (38, 48.6486) (39, 49.9325) (40, 48.7179) (41, 49.939) (42, 48.7805) (43, 49.9446) (44, 48.8372) (45, 49.9495) (46, 48.8889) (47, 49.9537) (48, 48.9362) (49, 49.9575) (50, 48.9796) (51, 49.9608) (52, 49.0196) (53, 49.9637) (54, 49.0566) (55, 49.9663) (56, 49.0909) (57, 49.9687) (58, 49.1228) (59, 49.9708) (60, 49.1525) (61, 49.9727) (62, 49.1803) (63, 49.9744) (64, 49.2063) (65, 49.976) (66, 49.2308) (67, 49.9774) (68, 49.2537) (69, 49.9787) (70, 49.2754) (71, 49.9799) (72, 49.2958) (73, 49.981) (74, 49.3151) (75, 49.982) (76, 49.3333) (77, 49.9829) (78, 49.3506) (79, 49.9838) (80, 49.3671) (81, 49.9846) (82, 49.3827) (83, 49.9853) (84, 49.3976) (85, 49.986) (86, 49.4118) (87, 49.9866) (88, 49.4253) (89, 49.9872) (90, 49.4382) (91, 49.9878) (92, 49.4505) (93, 49.9883) (94, 49.4624) (95, 49.9888) (96, 49.4737) (97, 49.9893) (98, 49.4845) (99, 49.9897) (100, 49.4949) } node[above right] {\small\color{blue}upper bound};

\draw[semithick,magenta] plot[mark=., mark color=none,mark options={fill=black,draw opacity=0},mark size=3pt] coordinates {
(2, 0.0) (3, 16.6667) (4, 25.0) (5, 30.0) (6, 33.3333) (7, 35.7143) (8, 37.5) (9, 38.8889) (10, 40.0) (11, 40.9091) (12, 41.6667) (13, 42.3077) (14, 42.8571) (15, 43.3333) (16, 43.75) (17, 44.1176) (18, 44.4444) (19, 44.7368) (20, 45.0) (21, 45.2381) (22, 45.4545) (23, 45.6522) (24, 45.8333) (25, 46.0) (26, 46.1538) (27, 46.2963) (28, 46.4286) (29, 46.5517) (30, 46.6667) (31, 46.7742) (32, 46.875) (33, 46.9697) (34, 47.0588) (35, 47.1429) (36, 47.2222) (37, 47.2973) (38, 47.3684) (39, 47.4359) (40, 47.5) (41, 47.561) (42, 47.619) (43, 47.6744) (44, 47.7273) (45, 47.7778) (46, 47.8261) (47, 47.8723) (48, 47.9167) (49, 47.9592) (50, 48.0) (51, 48.0392) (52, 48.0769) (53, 48.1132) (54, 48.1481) (55, 48.1818) (56, 48.2143) (57, 48.2456) (58, 48.2759) (59, 48.3051) (60, 48.3333) (61, 48.3607) (62, 48.3871) (63, 48.4127) (64, 48.4375) (65, 48.4615) (66, 48.4848) (67, 48.5075) (68, 48.5294) (69, 48.5507) (70, 48.5714) (71, 48.5915) (72, 48.6111) (73, 48.6301) (74, 48.6486) (75, 48.6667) (76, 48.6842) (77, 48.7013) (78, 48.7179) (79, 48.7342) (80, 48.75) (81, 48.7654) (82, 48.7805) (83, 48.7952) (84, 48.8095) (85, 48.8235) (86, 48.8372) (87, 48.8506) (88, 48.8636) (89, 48.8764) (90, 48.8889) (91, 48.9011) (92, 48.913) (93, 48.9247) (94, 48.9362) (95, 48.9474) (96, 48.9583) (97, 48.9691) (98, 48.9796) (99, 48.9899) (100, 49.0)} node[below right] {\small\color{magenta}lower bound};
\end{tikzpicture}
    \caption{Distortion of \kExtremes\ and lower bound stated in \Cref{thm:egalitarian-additive} for $k\in \{2,\ldots,99\}$.}
    \label{fig:plot-egalitarian-add}
\end{figure}
\begin{restatable}{theorem}{thmEgalitarianAdditive}\label{thm:egalitarian-additive}
    For every $n,k\in \NN$ with $n\geq k\geq 2$, \kExtremes\ has a distortion for egalitarian additive social cost of at most $\frac{3}{2} - \frac{1}{2(k-1)}$ if $k$ is even and at most $\frac{3}{2} - \frac{1}{k(k-1)}$ if $k$ is odd.
    Conversely, for every $k\in \NN$ with $k\geq 3$ there exists $n\in \NN$ with $n\geq k$ such that, for every $(n,k)$-voting rule $f$, $\dist(f) \geq \frac{3}{2}-\frac{1}{k}$ for egalitarian additive social cost.
\end{restatable}
\begin{proof}
We first show the bound on the distortion of \kExtremes. 
We fix $n,k\in \NN$ with $n\geq k\geq 2$, a linear order $\succ$ on $A=[n]$, and a consistent distance metric $d\ \rhd \succ$.
We write $\instance=(A,k,\succ)$ for the corresponding election and denote \kExtremes\ by $f$ and the outcome by $S$ in this part of the proof for compactness.

We claim that, if $d$ is such that $\SC(S,1;d)< \SC(S,n;d)$, there exists an alternative distance metric $d'$ with $\SC(S,1;d')\geq \SC(S,n;d')$ and  $\dist(f(\succ),\instance;d')\geq \dist(f(\succ),\instance;d)$.
Indeed, consider such $d$ defined by positions $x\in (-\infty,\infty)^n$, and let $d'$ be defined by positions $x'\in (-\infty,\infty)^n$, where $x'_a=x_{n+1-a}$ for every $a\in [n]$.
Since $f$ selects $\big\lfloor \frac{k}{2}\big\rfloor$ agents closest to the left-most agent and the $\big\lceil \frac{k}{2} \big\rceil$ agents closest to the right-most agent, we have
\[
    \SC(S,1;d')\geq \SC(S,n;d) > \SC(S,1;d) \geq \SC(S,n;d').
\]
Furthermore, this chain of inequalities combined with \Cref{lem:eg-add-extreme-agents} imply that $\SC(S,A;d') \geq \SC(S,A;d)$.
Since $\min\big\{\SC(S',A;d')\mid S'\in {A\choose k}\big\} = \min\big\{\SC(S',A;d)\mid S'\in {A\choose k}\big\}$, this yields $\dist(f(\succ),\instance;d')\geq \dist(f(\succ),\instance;d)$, so the claim follows.
Thanks to this claim, we can assume in what follows that $\SC(S,1;d)\geq \SC(S,n;d)$ and thus, by \Cref{lem:eg-add-extreme-agents}, $\SC(S,A;d)=\SC(S,1;d)$.

We distinguish three cases depending on the distances from agent $1$ to other agents and show the claimed distortion for each of them.
We first suppose that $d\big(1,\big\lfloor \frac{k}{2}\big\rfloor\big) \leq \frac{d(1,n)}{2}$.
In this case,
\begin{align*}
    \SC(S,1;d) & = \sum_{s=1}^{\lfloor k/2\rfloor} d(1,s) + \sum_{s=n-\lceil k/2\rceil + 1}^{n} d(1,s) \\
    & \leq \bigg(\bigg\lfloor\frac{k}{2}\bigg\rfloor-1\bigg) \frac{d(1,n)}{2} + \bigg\lceil \frac{k}{2}\bigg\rceil d(1,n) \\
    & = \bigg(k+\bigg\lceil \frac{k}{2}\bigg\rceil-1\bigg) \frac{d(1,n)}{2},
\end{align*}
where we used the assumption $d\big(1,\big\lfloor \frac{k}{2}\big\rfloor\big) \leq d/2$ and the fact that $d(1,1)=0$ for the inequality.
From \Cref{lem:eg-add-extreme-agents} and \cref{eq:eg-add-sum-sc-extremes} we know that $\SC(S',A;d)\geq \frac{kd(1,n)}{2}$ for any $S'\in {A\choose k}$, so we obtain
\[
    \dist(f(\succ), \instance) = \frac{\SC(S,1;d)}{\min_{S'\in {A\choose k}} \SC(S',A;d)} \leq \frac{\big(k+\big\lceil \frac{k}{2}\big\rceil-1\big) \frac{d(1,n)}{2}}{\frac{kd(1,n)}{2}} = \frac{3}{2} - \frac{2-k\bmod 2}{2k},
\]
which is smaller than $\frac{3}{2} - \frac{1}{2(k-1)}$ for even $k \geq 2$ and smaller than $\frac{3}{2} - \frac{1}{k(k-1)}$ for odd $k \geq 3$.
Thus, we conclude the result in this case.

We next suppose that $d\big(1,\big\lfloor \frac{k}{2}\big\rfloor\big) > \frac{d(1,n)}{2}$ and $\sum_{s=2}^{\lfloor k/2\rfloor} d(1,s) \leq \frac{k-2-k\bmod2}{k-1} \cdot \frac{kd(1,n)}{4}$.
In a similar way as before, we now have
\begin{align*}
    \SC(S,1;d) & = \sum_{s=1}^{\lfloor k/2\rfloor} d(1,s) + \sum_{s=n-\lceil k/2\rceil + 1}^{n} d(1,s) \\
    & \leq \frac{k-2-k\bmod2}{k-1} \cdot \frac{kd(1,n)}{4} + \bigg\lceil \frac{k}{2}\bigg\rceil d(1,n) \\
    & = \bigg(3k - \frac{k-(k-2)k\bmod 2}{k-1}\bigg)\frac{d(1,n)}{4},
\end{align*}
where the inequality follows from the assumption $\sum_{s=2}^{\lfloor k/2\rfloor} d(1,s) \leq \frac{k-2-k\bmod2}{k-1} \cdot \frac{kd(1,n)}{4}$ and the fact that $d(1,1)=0$.
From \Cref{lem:eg-add-extreme-agents} and \cref{eq:eg-add-sum-sc-extremes} we know that $\SC(S',A;d)\geq \frac{kd(1,n)}{2}$ for any $S'\in {A\choose k}$, so we obtain
\begin{align*}
    \dist(f(\succ), \instance) & = \frac{\SC(S,1;d)}{\min_{S'\in {A\choose k}} \SC(S',A;d)} \\
    & \leq \frac{\big(3k - \frac{k-(k-2)k\bmod 2}{k-1}\big)\frac{d(1,n)}{4}}{\frac{kd(1,n)}{2}} = \frac{3}{2} - \frac{k-(k-2)k\bmod 2}{2k(k-1)},
\end{align*}
which corresponds to the expression in the statement.

We finally consider the case with $d\big(1,\big\lfloor \frac{k}{2}\big\rfloor\big) > \frac{d(1,n)}{2}$ and $\sum_{s=2}^{\lfloor k/2\rfloor} d(1,s) > \frac{k-2-k\bmod2}{k-1} \cdot \frac{kd(1,n)}{4}$.
Since the distance between $1$ and the right-most point among $\big\{2,\ldots,\big\lfloor \frac{k}{2}\big\rfloor\big\}$, namely $d\big(1,\big\lfloor \frac{k}{2}\big\rfloor\big)$, is at least its average distance to points within this set, we know that
\begin{equation}
    d\bigg(1,\bigg\lfloor \frac{k}{2}\bigg\rfloor\bigg) \geq \frac{1}{\big\lfloor \frac{k}{2}\big\rfloor -1} \sum_{s=2}^{\lfloor k/2\rfloor}d(1,s) \geq \frac{1}{\big\lfloor \frac{k}{2}\big\rfloor -1} \cdot \frac{k-2-k\bmod2}{k-1} \cdot \frac{kd(1,n)}{4} = \frac{kd(1,n)}{2(k-1)}.\label{ineq:eg-add-lb-right-most}
\end{equation}
Let now $S'\in {A\setminus \{1\}\choose k-1}$ be any set of $k-1$ agents without $1$.
Since $\big\{2,\ldots,\big\lfloor \frac{k}{2}\big\rfloor\big\}$ are the closest agents to $1$, we know that $\frac{1}{k-1}\sum_{s\in S'} d(1,s) \geq \frac{1}{\lfloor k/2\rfloor -1}\sum_{s=2}^{\lfloor k/2\rfloor} d(1,s)$.
Rearranging this expression and using our assumption once again, we obtain
\[
    \sum_{s\in S'} d(1,s) \geq \frac{k-1}{\big\lfloor \frac{k}{2}\big\rfloor -1} \sum_{s=2}^{\lfloor k/2\rfloor}d(1,s) \geq \frac{kd(1,n)}{2},
\]
where we used \cref{ineq:eg-add-lb-right-most} for the last inequality.
For any committee $S'\in {A\choose k}$, this implies that $\SC(S',1;d)\geq \frac{kd(1,n)}{2}$, \cref{eq:eg-add-sum-sc-extremes} implies that $\SC(S',1;d) \geq \SC(S',n;d)$, and \Cref{lem:eg-add-extreme-agents} implies that $\SC(S',A;d)=\SC(S',1;d)$.
Therefore,
\begin{equation}
    \min_{S'\in {A\choose k}} \SC(S',A;d) = \min_{S'\in {A\choose k}} \SC(S',1;d) = \sum_{s=2}^k d(1,s);\label{eq:eg-add-opt-char}
\end{equation}
i.e., the optimal set in this case corresponds to $\{1,\ldots,k\}$.
Combining the previous expressions, we obtain the following chain of inequalities:
\allowdisplaybreaks
\begin{align*}
    \dist(f(\succ), \instance) & = \frac{\SC(S,1;d)}{\min_{S'\in {A\choose k}} \SC(S',A;d)}\\
    & = 1+ \frac{\SC(S,1;d) - \min_{S'\in {A\choose k}} \SC(S',A;d)}{\min_{S'\in {A\choose k}} \SC(S',A;d)}\\
    & \leq 1+ \frac{2}{kd(1,n)} \bigg(\sum_{s\in S} d(1,s) - \sum_{s=2}^{k} d(1,s)\bigg)\\
    & = 1 + \frac{2}{kd(1,n)} \bigg(\sum_{s=n-\lceil k/2\rceil +1}^{n} d(1,s) -  \sum_{s=\lfloor k/2\rfloor +1}^{k} d(1,s) \bigg)\\
    & \leq 1 + \frac{2}{kd(1,n)} \cdot \bigg\lceil \frac{k}{2} \bigg\rceil \bigg(d(1,n)-d\bigg(1,\bigg\lfloor \frac{k}{2}\bigg\rfloor\bigg)\bigg)\\
    & \leq 1 + \frac{2}{kd(1,n)} \cdot \bigg\lceil \frac{k}{2} \bigg\rceil \bigg( d(1,n) - \frac{kd(1,n)}{2(k-1)}\bigg)\\
    & = \frac{3}{2} - \frac{k-(k-2)k\bmod 2}{2k(k-1)}.
\end{align*}
Indeed, the first inequality follows from \cref{eq:eg-add-opt-char} and the fact that $\SC(S',A;d)\geq \frac{kd(1,n)}{2}$ for every $S'\in {A\choose k}$ due to \cref{eq:eg-add-sum-sc-extremes}, the third equality from the definition of $f$, the second inequality from simple bounds on $d(1,s)$ for different values of $s$, and the last inequality from \cref{ineq:eg-add-lb-right-most}.
The other equalities come from simple calculations.
Since the last expression again corresponds to the expression in the statement, we conclude.

For the lower bound, we consider any $k\in \NN$ with $k\geq 3$, we fix $n=2(k+1)$, and consider an arbitrary $(n,k)$-voting rule $f$.
We partition the agents into four sets $A=\dot\bigcup_{i=1}^4 A_i$ such that $A_1=\{1\},$ $A_4=\{n\}$ and $|A_2|=|A_3|=k$.
We consider the profile $\succ\in \calL^n(n)$, where $S = f(\succ)$, and
\begin{enumerate}[label=(\roman*)]
    \item $b\succ_a c$ whenever $a\in A_i,b\in A_j,c\in A_\ell$ for some $i,j,\ell\in [4]$ with $|i-j|<|i-\ell|$;
    \item $1\succ_a b$ whenever $a\in A_2,b\in A_3\cup A_4$;
    \item $n\succ_a b$ whenever $a\in A_3,b\in A_1\cup A_2$;
\end{enumerate}
and the remaining pairwise comparisons are arbitrary.
We consider the election $\instance=(A,k,\succ)$ with $A=[n]$.

In what follows, we distinguish whether $f$ selects more agents from $A_1\cup A_2$ or from $A_3\cup A_4$ and construct appropriate distance metrics to show that, in either case, the distortion is at least the one claimed in the statement.
Intuitively, if $f$ selects more agents from $A_1\cup A_2$ we will consider a metric where these sets lie on one extreme, $A_4=n$ on the other extreme, and all agents $A_3$ in the middle, so that picking all agents from $A_3$ would lead to a much lower social cost.
In the opposite case, we will construct a symmetric instance.

Formally, we first consider the case with $|S \cap (A_1\cup A_2)| \geq \frac{k}{2}$ and define the distance metric $d_1$ on $A$ by the following positions $x\in (-\infty,\infty)^n$: $x_a=0$ for every $a\in A_1\cup A_2$, $x_a=1$ for every $a\in A_3$, and $x_n=2$.
It is not hard to check that $d_1\ \rhd \succ$; see \Cref{fig:egalitarian-add} for an illustration.
Since $|S \cap (A_1\cup A_2)| \geq \frac{k}{2}$, we obtain
\[
    \dist(f(\succ), \instance) \geq \frac{\SC(S,A;d_1)}{\SC(A_3,A;d_1)} \geq \frac{\SC(S,n;d_1)}{\SC(A_3,n;d_1)} \geq \frac{(k-1)+|S\cap(A_1\cup A_2)|}{k} \geq\frac{3}{2}-\frac{1}{k}.
\]

Conversely, if $|S\cap (A_3\cup A_4)| \geq \frac{k}{2}$, we define the distance metric $d_2$ on $A$ by the following positions $x\in (-\infty,\infty)^n$: $x_1=0$, $x_a=1$ for every $a\in A_2$, and $x_a=2$ for every $a\in A_3\cup A_4$.
It is not hard to check that $d_2\ \rhd \succ$; see \Cref{fig:egalitarian-add} for an illustration.
Since $|S\cap (A_3\cup A_4)| \geq \frac{k}{2}$, we obtain
\[
    \dist(f(\succ), \instance) \geq \frac{\SC(S,A;d_2)}{\SC(A_2,A;d_2)} \geq \frac{\SC(S,1;d_2)}{\SC(A_2,1;d_2)} \geq \frac{(k-1)+|S\cap(A_3\cup A_4)|}{k} \geq\frac{3}{2}-\frac{1}{k}.
\]

Since $\dist(f(\succ), \instance)\geq \frac{3}{2}-\frac{1}{k}$ regardless of $f(\succ)$, we conclude that $\dist(f) \geq \frac{3}{2}-\frac{1}{k}$.
\end{proof}

\begin{figure}
\centering
\begin{tikzpicture}[scale=1, every node/.style={font=\footnotesize}]
    \node[anchor=east] at (-0.5, 0) {\textbf{Metric $d_1$}};
    \filldraw[black] (0,0) circle (2.5pt) node[anchor=north] {$A_1\cup A_2$};
    \filldraw[black] (3,0) circle (2.5pt) node[anchor=north] {$A_3$};
    \filldraw[black] (6,0) circle (2.5pt) node[anchor=north] {$A_4$};
    \draw[-] (0,0) -- (3,0) node[midway, anchor=south] {$1$};
    \draw[-] (3,0) -- (6,0) node[midway, anchor=south] {$1$};

\begin{scope}[yshift=-1cm],
    \node[anchor=east] at (-0.5, 0) {\textbf{Metric $d_2$}};
    \filldraw[black] (0,0) circle (2.5pt) node[anchor=north] {$A_1$};
    \filldraw[black] (3,0) circle (2.5pt) node[anchor=north] {$A_2$};
    \filldraw[black] (6,0) circle (2.5pt) node[anchor=north] {$A_3\cup A_4$};
    \draw[-] (0,0) -- (3,0) node[midway, anchor=south] {$1$};
    \draw[-] (3,0) -- (6,0) node[midway, anchor=south] {$1$};
\end{scope}
\end{tikzpicture}
    \caption{Metrics considered in the proof of \Cref{thm:egalitarian-additive}.}
    \label{fig:egalitarian-add}
\end{figure}

\subsection{Egalitarian $q$-Cost}\label{subsec:egalitarian-qcost}

We now turn our attention to the $q$-cost aggregation function of candidates, so that the social cost is the maximum over agents of the distance from each agent to its $q$th closest candidate; i.e.,
\[
    \SC(S',A;d) = \max\big\{\tilde{d}(a)_q \mid a\in A\big\}
\]
for a set of agents $A$, a committee size $k$, a committee $S'\in {A\choose k}$, and a distance metric $d$, where $\tilde{d}(a)\in \RR^{S'}_+$ contains the values $\{d(a,s)\mid s\in S'\}$ in increasing order.

We begin by showing that no voting rule can guarantee a constant distortion for $q$-cost when $q \leq \frac{k}{3}$. This implies that the unbounded distortion for this objective, previously established in the setting of disjoint voters and candidates~\citep{caragiannis2022metric}, also holds in our setting.
\begin{restatable}{theorem}{thmEgalitarianqCostSmallq}\label{thm:egalitarian-qcost-smallq}
    For every $k,q\in \NN$ with $\frac{k}{3}\geq q$, there exists $n\in \NN$ with $n\geq k$ such that, for every $(n,k)$-voting rule $f$, $\dist(f)$ is unbounded for egalitarian $q$-cost.
\end{restatable}
\begin{proof}
We let $k,q\in \NN$ with $\frac{k}{3}\geq q$ be arbitrary, define $p=\big\lfloor \frac{k}{q}\big\rfloor$, and take $n=(p+1)q$.
We partition the agents into $p+1\geq 4$ sets $A=\dot\bigcup_{i\in [p+1]}A_i$ such that $|A_i|\geq q$ for every $i\in [p+1]$; note that this is possible since $(p+1)q \leq \big(\frac{k}{q}+1)q=k+q=n$.
We consider any fixed $(n,k)$-voting rule $f$ and the profile $\succ\in \calL^n(n)$, where $S = f(\succ)$, and
\begin{enumerate}[label=(\roman*)]
    \item $b\succ_a c$ whenever $a\in A_i,b\in A_j,c\in A_\ell$ with $|i-j|<|i-\ell|$ for some $i,j,\ell\in [p+1]$;
    \item $b\succ_a c$ whenever $a\in A_i,b\in A_1,c\in A_j$ with $|i-1|=|i-j|$ for some $i,j\in [p]$;
    \item $b\succ_a c$ whenever $a\in A_i,b\in A_{p+1},c\in A_j$ with $|i-(p+1)|=|i-j|$ for some $i,j\in \{2,\ldots,p+1\}$;
\end{enumerate}
and the remaining pairwise comparisons are arbitrary.
We consider the election $\instance=(A,k,\succ)$ with $A=[n]$.
Since $(p+1)q>\frac{k}{q}q=k$, we know that there exists $j\in [p+1]$ such that $|S\cap A_j| < q$.
We distinguish two cases depending on the identity of $j$.

If $j\notin \{p,p+1\}$, we consider the distance metric $d_1$ on $A$ given by the following positions $x\in (-\infty,\infty)^n$: $x_a = i-1$ for every $a\in A_i$ and $i\in [p]$, and $x_a = p-1$ for every $a\in A_{p+1}$.
It is not hard to see that $d_1\ \rhd \succ$; see \Cref{fig:egalitarian-qcost-smallq} for an illustration.
Since $|S \cap A_j| < q$ for some $j\notin \{p,p+1\}$, we have that $\SC(S,A_j;d_1)=1$. 
On the other hand, we can define an alternative committee $S'=\bigcup_{i\in [p]}S'_i$ such that $|S'_i\cap A_i|\geq q$ for every $i\in [p]$, which is possible because $|A_i|\geq q$ for every $i\in [p]$ and $pq\leq \frac{k}{q}q=k$.
Since $\SC(S',A;d_1)=0$ and $\dist(f(\succ), \instance) \geq \frac{\SC(S,A;d_1)}{\SC(S',A;d_1)}$, we conclude that $\dist(f(\succ), \instance)$ is unbounded.

If $j\in \{p,p+1\}$, we consider the distance metric $d_2$ on $A$ given by the following positions $x\in (-\infty,\infty)^n$: $x_a=0$ for every $a\in A_q$, and $x_a = i-2$ for every $a\in A_i$ and $i\in \{2,\ldots,p+1\}$.
It is not hard to see that $d_2\ \rhd \succ$; see \Cref{fig:egalitarian-qcost-smallq} for an illustration.
Since $|S\cap A_j| < q$ for some $j\in \{p,p+1\}$, we have that $\SC(S,A_j;d_2)=1$. 
On the other hand, we can define an alternative committee $S'=\bigcup_{i\in \{2,\ldots,p+1\}}S'_i$ such that $|S'_i\cap A_i|\geq q$ for every $i\in \{2,\ldots,p+1\}$, which is possible because $|A_i|\geq q$ for every $i\in \{2,\ldots,p+1\}$ and $pq\leq \frac{k}{q}q=k$.
Since $\SC(S',A;d_2)=0$ and $\dist(f(\succ), \instance) \geq \frac{\SC(S,A;d_2)}{\SC(S',A;d_2)}$, we conclude that $\dist(f(\succ), \instance)$ is unbounded.

Since $\dist(f(\succ), \instance)$ is unbounded regardless of $f(\succ)$, we conclude that $\dist(f)$ is unbounded.
\end{proof}

\begin{figure}
\centering
    \begin{tikzpicture}[scale=1, every node/.style={font=\footnotesize}]
    \node[anchor=east] at (-0.5, 0) {\textbf{Metric $d_1$}};
    \filldraw[black] (0,0) circle (2.5pt) node[anchor=north] {$A_1$};
    \filldraw[black] (1.5,0) circle (2.5pt) node[anchor=north] {$A_2$};
    \filldraw[black] (3,0) circle (2.5pt) node[anchor=north] {$A_3$};
    \filldraw[black] (3.5,0) circle (0.5pt) node[anchor=north] {};
    \filldraw[black] (4,0) circle (0.5pt) node[anchor=north] {};
    \filldraw[black] (4.5,0) circle (0.5pt) node[anchor=north] {};
    \filldraw[black] (5,0) circle (2.5pt) node[anchor=north] {$A_{p-2}$};
    \filldraw[black] (6.5,0) circle (2.5pt) node[anchor=north] {$A_{p-1}$};
    \filldraw[black] (8,0) circle (2.5pt) node[anchor=north] {$A_p\cup A_{p+1}$};
    \draw[-] (0,0) -- (1.5,0) node[midway, anchor=south] {$1$};
    \draw[-] (1.5,0) -- (3,0) node[midway, anchor=south] {$1$};
    \draw[-] (5,0) -- (6.5,0) node[midway, anchor=south] {$1$};
    \draw[-] (6.5,0) -- (8,0) node[midway, anchor=south] {$1$};

\begin{scope}[yshift=-1cm],
    \node[anchor=east] at (-0.5, 0) {\textbf{Metric $d_2$}};
    \filldraw[black] (0,0) circle (2.5pt) node[anchor=north] {$A_1\cup A_2$};
    \filldraw[black] (1.5,0) circle (2.5pt) node[anchor=north] {$A_3$};
    \filldraw[black] (3,0) circle (2.5pt) node[anchor=north] {$A_4$};
    \filldraw[black] (3.5,0) circle (0.5pt) node[anchor=north] {};
    \filldraw[black] (4,0) circle (0.5pt) node[anchor=north] {};
    \filldraw[black] (4.5,0) circle (0.5pt) node[anchor=north] {};
    \filldraw[black] (5,0) circle (2.5pt) node[anchor=north] {$A_{p-1}$};
    \filldraw[black] (6.5,0) circle (2.5pt) node[anchor=north] {$A_{p}$};
    \filldraw[black] (8,0) circle (2.5pt) node[anchor=north] {$A_{p+1}$};
    \draw[-] (0,0) -- (1.5,0) node[midway, anchor=south] {$1$};
    \draw[-] (1.5,0) -- (3,0) node[midway, anchor=south] {$1$};
    \draw[-] (5,0) -- (6.5,0) node[midway, anchor=south] {$1$};
    \draw[-] (6.5,0) -- (8,0) node[midway, anchor=south] {$1$};
\end{scope}
\end{tikzpicture}
    \caption{Metrics considered in the proof of \Cref{thm:egalitarian-qcost-smallq}.}
    \label{fig:egalitarian-qcost-smallq}
\end{figure}

In the context of egalitarian $q$-cost for $q> \frac{k}{3}$, much better results are possible.
The case with $q> \frac{k}{2}$ behaves similarly to the setting where a single candidate is to be selected: Any voting rule achieves a distortion of $2$ and this is best possible.
When $\frac{k}{3} < q \leq \frac{k}{2}$, the best-possible distortion a voting rule can achieve is again $2$, but not any rule does so.
We show that \kExtremes\ attains it.
\begin{restatable}{theorem}{thmEgalitarianqCostqLarge}\label{thm:egalitarian-qcost-qlarge}
    Let $n,k,q\in \NN$ be such that $n\geq k\geq 2$ and $q> \frac{k}{3}$.
    If $q>\frac{k}{2}$, any $(n,k)$-voting rule has distortion $2$ for egalitarian $q$-cost. 
    If $q>\frac{k}{3}$, \kExtremes\ has distortion $2$ for egalitarian $q$-cost.
    For every $k,q\in \NN$ with $q>\frac{k}{3}\geq 1$, there exists $n\in \NN$ with $n\geq k$ such that, for every $(n,k)$-voting rule $f$, $\dist(f)\geq 2$.
\end{restatable}
\begin{proof}
    Let $n,k\in \NN$ be such that $n\geq k\geq 2$.
    Let first $q\in\NN$ be such that $q> \frac{k}{2}$. 
    Let $f$ be any $(n,k)$-voting rule and let $\succ\ \in \calL^n(n)$ be an arbitrary preference profile on $A=[n]$.
    We denote, as usual, agents by $\{1,\ldots,n\}$ from left to right, $S = f(\succ)$, and we let $d\ \rhd \succ$ be any consistent distance metric.
    For a committee $S'\in {A\choose k}$, we let $\tilde{d}(S',a)\in \RR^k_+$ denote the vector with the values $\{d(a,s)\mid s\in S'\}$ in increasing order.
    It is clear that
    \begin{equation}
        \SC(S,A;d) = \max\{\tilde{d}(S,a)_q \mid a\in A\} \leq \max\{d(a,b)\mid a,b\in A\} = d(1,n).\label{ineq:egalitarian-qcost-ub-sc-alg}
    \end{equation}
    On the other hand, for every committee $S'\in {A\choose k}$, if we denote the agents in $S'$ in increasing order by $s_1,\ldots,s_k$ we have that $s_q>s_{k-q}$ because $q>\frac{k}{2}$.
    This implies that, for every committee $S'\in {A\choose k}$, we have
    \[
        \tilde{d}(S',1)_q+\tilde{d}(S',n)_q = s(1,s_q)+ d(s_{k-q},n) > d(1,n),
    \]
    and thus $\max\{\tilde{d}(S',1)_q,\tilde{d}(S',n)_q\} \geq \frac{d(1,n)}{2}$.
    Therefore,
    \begin{align}
        \min_{S'\in {A\choose k}} \SC(S',A;d) & = \min_{S'\in {A\choose k}} \max\{\tilde{d}(S',a)_q \mid a\in A\} \nonumber\\
        & \geq \min_{S'\in {A\choose k}} \max\{\tilde{d}(S',1)_q,\tilde{d}(S',n)_q\}\geq \frac{d(1,n)}{2}.\label{ineq:egalitarian-qcost-lb-sc-opt}
    \end{align}
    Combining \cref{ineq:egalitarian-qcost-ub-sc-alg,ineq:egalitarian-qcost-lb-sc-opt}, we directly obtain that $\dist(f)\leq 2$.

    Let now $q\in \NN$ be such that $\frac{k}{3}<q\leq \frac{k}{2}$, $\succ\ \in \calL^n(n)$ be an arbitrary preference profile on $A=[n]$, and $d\ \rhd \succ$ be a consistent distance metric; we consider the election $\instance=(A,k,\succ)$.
    We denote the outcome of \kExtremes\ for this profile by $S$ for compactness.
    We denote agents by $\{1,\ldots,n\}$ from left to right and, for $S'\in {A\choose k}$, we let $\tilde{d}(S',a)\in \RR^k_+$ denote the vector with the values $\{d(a,s)\mid s\in S'\}$ in increasing order.
    We finally let $a^*\in \arg\max\{\min\{d(1,a),d(a,n)\}\mid a\in A\}$ denote the agent with maximum distance from both extreme agents, assume w.l.o.g.\ that this is its distance to $1$, i.e., $d(1,a^*)\leq d(a^*,n)$, and write $d^*=d(1,a^*)$ for this distance. 
    Observe that 
    \begin{equation}
        \min \{d(a^*,n),d(1,a^*+1)\} \geq \frac{d(1,n)}{2}.\label{ineq:lbs-distances-a*}
    \end{equation}
    Indeed, $d(a^*,n) \geq \frac{d(1,n)}{2}$ follows directly from the inequality $d(1,a^*)\leq d(a^*,n)$ and the equality $d(1,a^*)+d(a^*,n)=d(1,n)$.
    Having $d(1,a^*+1)< \frac{d(1,n)}{2}$ would imply $\min\{d(1,a^*+1),d(a^*+1,n)\} > d^*$, a contradiction to the definition of $a^*$. 

    We first tackle two simple cases. If $a^*<q$, i.e., there are less than $q$ agents between $1$ and $a^*$, then for any committee $S'\in {A\choose k}$ we have $\SC(S',A;d)\geq \SC(S',1;d)\geq d(1,a^*+1)\geq \frac{d(1,n)}{2}$, where the second inequality follows from \cref{ineq:lbs-distances-a*}. 
    Since $\SC(S',A;d)\leq d(1,n)$ holds for any committee $S'\in {A\choose k}$, we know that in particular $\SC(S,A;d)\leq d(1,n)$ and thus $\dist(f)\leq 2$.
    Similarly, if $n-a^*<q$, i.e., there are less than $q$ agents between $a^*+1$ and $n$, then for any committee $S'\in {A\choose k}$ we have $\SC(S',A;d)\geq \SC(S',1;d)\geq d(a^*,n)\geq \frac{d(1,n)}{2}$, where the second inequality follows from \cref{ineq:lbs-distances-a*}.
    As before, $\dist(f)\leq 2$ thus follows directly.

    If none of the previous cases hold, we have both $a^*\geq q$ and $n-a^*\geq 2$, so that from the definition of \kExtremes\ we have $|S\cup \{1,\ldots,a^*\}| = \big\lfloor \frac{k}{2}\big\rfloor \geq q$ and $|S\cup \{a^*+1,\ldots,n\}| = \big\lceil \frac{k}{2}\big\rceil \geq q$.
    This implies that
    \begin{equation}
        \SC(S,A;d)\leq \max\{d(1,a^*),d(a^*+1,n)\}\leq d^*.\label{ineq:ub-sc-kextremes-a^*}
    \end{equation}
    We claim that, for every $S'\in {A\choose k}$, we have $\SC(S',A;d)\geq \frac{d^*}{2}$. Together with \cref{ineq:ub-sc-kextremes-a^*}, this would immediately imply $\dist(f)\leq 2$ and conclude the proof.
    To prove this fact, suppose for the sake of contradiction that $\SC(S',A;d)< \frac{d^*}{2}$ for some $S'\in {A\choose k}$. 
    This is equivalent to the fact that
    \[
        \SC(S',a;d) < \frac{d^*}{2} \Longleftrightarrow \bigg| S'\cup \bigg\{b\in A: d(a,b)< \frac{d^*}{2}\bigg\} \bigg| \geq q
    \]
    for every $a\in A$. 
    Since the sets $\big\{b\in A\mid d(a,b)< \frac{d^*}{2}\big\}$ for $a\in \{1,a^*,n\}$ are disjoint, we conclude that $|S'|\geq 3q > k$, a contradiction.
    
    For the lower bound, we consider the same instances as in the proof of \Cref{thm:egalitarian-additive}; we repeat the construction for completeness.
    Naturally, the proof of the lower bound in the end differs from the additive case.
    We consider any $k\in \NN$ with $k\geq 2$, we fix $n=2(k+1)$, and consider an arbitrary $(n,k)$-voting rule $f$.
We partition the agents into four sets $A=\dot\bigcup_{i=1}^4 A_i$ such that $A_1=\{1\},$ $A_4=\{n\}$ and $|A_2|=|A_3|=k$.
We consider the profile $\succ\in \calL^n(n)$, where $S = f(\succ)$, and
\begin{enumerate}[label=(\roman*)]
    \item $b\succ_a c$ whenever $a\in A_i,b\in A_j,c\in A_\ell$ for some $i,j,\ell \in [4]$ with $|i-j|<|i-\ell|$;
    \item $1\succ_a b$ whenever $a\in A_2,b\in A_3\cup A_4$;
    \item $n\succ_a b$ whenever $a\in A_3,b\in A_1\cup A_2$;
\end{enumerate}
and the remaining pairwise comparisons are arbitrary.
We consider the election $\instance=(A,k,\succ)$ with $A=[n]$.

In what follows, we distinguish whether $f$ selects more agents from $A_1\cup A_2$ or from $A_3\cup A_4$ and construct appropriate distance metrics to show that, in either case, the distortion is at least the one claimed in the statement.
Intuitively, if $f$ selects more agents from $A_1\cup A_2$ we will consider a metric where these sets lie on one extreme, $A_4=n$ on the other extreme, and all agents $A_3$ in the middle.
This way, the selected committee gives twice the social cost as picking all agents from $A_3$
In the opposite case, we will construct a symmetric instance.

Formally, we first consider the case with $|S \cap (A_1\cup A_2)| \geq \frac{k}{2}$ and define the distance metric $d_1$ on $A$ by the following positions $x\in (-\infty,\infty)^n$: $x_a=0$ for every $a\in A_1\cup A_2$, $x_a=1$ for every $a\in A_3$, and $x_n=2$.
It is not hard to check that $d_1\ \rhd \succ$; see \Cref{fig:egalitarian-add} for an illustration.
Since $|S\cap (A_1\cup A_2)| \geq \frac{k}{2}$, we obtain $\SC(S,n;d_1)=2$ and thus
\[
    \dist(f(\succ),\instance) \geq \frac{\SC(S,A;d_1)}{\SC(A_3,A;d_1)} \geq \frac{\SC(S,n;d_1)}{\SC(A_3,n;d_1)} \geq 2.
\]

Conversely, if $|S\cap (A_3\cup A_4)| \geq \frac{k}{2}$, we define the distance metric $d_2$ on $A$ by the following positions $x\in (-\infty,\infty)^n$: $x_1=0$, $x_a=1$ for every $a\in A_2$, and $x_a=2$ for every $a\in A_3\cup A_4$.
It is not hard to check that $d_2\ \rhd \succ$; see \Cref{fig:egalitarian-add} for an illustration.
Since $|S\cap (A_3\cup A_4)| \geq \frac{k}{2}$, we obtain $\SC(S,1;d_2)=2$ and thus
\[
    \dist(f(\succ),\instance) \geq \frac{\SC(S,A;d_2)}{\SC(A_2,A;d_2)} \geq \frac{\SC(S,1;d_2)}{\SC(A_2,1;d_2)} \geq 2.
\]

Since $\dist(f(\succ),\instance)\geq 2$ regardless of $f(\succ)$, we conclude that $\dist(f) \geq 2$.
\end{proof}